\documentclass[a4paper,10pt,reqno]{article}
\usepackage[colorlinks=true]{hyperref}
\usepackage{amssymb}
\usepackage{mathtools}
\usepackage{comment}
\usepackage[makeroom]{cancel}
\usepackage{tikz-cd}
\usepackage[linguistics]{forest}
\usepackage{tikz-qtree}
\usepackage{caption}
\usepackage[all]{hypcap}
\forestset{sn edges/.style={for tree={parent anchor= north, child anchor=south}}}
\forestset{
fairly nice empty nodes/.style={
delay={where content={}
{shape=coordinate, }{}},
for tree={s sep=4mm}
}
}
\forestset{
nice empty nodes/.style={
    for tree={calign=fixed edge angles},
    delay={where content={}{shape=coordinate, for children={anchor=south} }{}}
},
}

\newcommand{\p}{\mathrm p^{ \geq 2}}

\def \ptr {{\mathrm{p}}^1_{|}}

\def \pv {{\mathrm{p}}^1_{\vee}}

\usepackage{adjustbox}

\usepackage[utf8]{inputenc}

\usepackage[textsize=tiny]{todonotes}
\usepackage{latexsym,amsfonts,amsthm,amsmath,amscd,amssymb,color}
\usepackage[all]{xy}

\theoremstyle{plain}
\newtheorem{lemma}{Lemma}[section]
\newtheorem{theorem}[lemma]{Theorem}
\newtheorem{proposition}[lemma]{Proposition}

\newtheorem{corollary}[lemma]{Corollary}
\newtheorem{convention}[lemma]{Convention}
\newtheorem{question}[lemma]{Question}

\theoremstyle{definition}
\newtheorem{definition}[lemma]{Definition}
\newtheorem{deff}[lemma]{{Definition}}
\newtheorem{example}[lemma]{Example}
\newtheorem{remark}[lemma]{Remark}
\newtheorem*{definition*}{Definition}
\theoremstyle{remark}

\newcommand{\V}{\mathcal{V}}
\newcommand{\lb}{\left[ \cdot\,,\cdot\right] }
\usepackage[left=2cm, right=3cm, top=3cm, bottom=1.7cm]{geometry}
\newcommand{\dd}{\mathrm{d}}

\newcommand{\X}{\mathfrak {X} }

\def\cO{\mathcal O}
\def \cI {\mathcal I}

\def \cA {\mathcal A}
\def \cV {\mathcal V}

\def \FM {\mathfrak M}

\def \r  {\mathfrak r}
\def \cTr{\mathcal T ree}

\title{
On construction of differential $\mathbb Z$-graded varieties}
\usepackage{authblk}

\author{Aliaksandr Hancharuk\thanks{Department of Mathematics, Jilin University, Changchun, China.\\ Email: hancharuk@jlu.edu.cn}\;\, and\;Ruben Louis\thanks{Department of Mathematics, University of Illinois Urbana-Champaign, Urbana, IL, USA.\\Email: rlouis@illinois.edu}}

\begin{document}
\tikzset{mystyle/.style={xshift = - 2.4ex, yshift = - 1.5ex, align=left}}
\tikzset{mystyle2/.style={xshift = + 2.4ex, yshift = - 1.5ex, align=right}}
\maketitle

\begin{abstract}


 Given a commutative unital algebra $\cO$, a proper ideal $\cI$ in $\cO$, and a positively graded differential variety over $\cO/\cI$, we provide a $\mathbb Z$-graded extension, whose negative part is an arborescent Koszul-Tate resolution of $\cO/\cI$. This extension is obtained through an algorithm exploiting the explicit homotopy retract data of the arborescent Koszul-Tate resolution, so that the number of homological computations in the construction is significantly reduced. For a positively graded differential variety over $\cO$ that preserves the ideal $\cI$, the extension admits a manifest description in terms of decorated trees and computed data.

 As a by-product, to every Lie--Rinehart algebra over the coordinate ring of an affine variety \( W \subseteq M = \mathbb{C}^d \), one associates an explicit differential \(\mathbb{Z}\)-graded variety over \(M\) whose negative component is the arborescent Koszul--Tate resolution of the coordinate ring $\mathbb C[x_1, \ldots, x_d]/\cI_W$ of \(W\), and whose positive component is the universal dg-variety of the given Lie--Rinehart algebra. Concrete examples are given.

\end{abstract}
\tableofcontents
\section*{Introduction}
The study of singular spaces in differential and algebraic geometry has been approached using various distinct methods. One notable approach involves examining singularities via Lie algebras of vector fields. The investigation of the singularities of singular spaces through their associated Lie algebras of vector fields is a topic with a deep and rich history. This line of inquiry began with the  work of Pursell and Shanks \cite{Pursell-Shanks}, who formulated the problem of characterizing a manifold by its Lie algebra of vector fields. Subsequently, similar problems were studied by many authors in different contexts, including H. Omori \cite{Omori1980} and J. Grabowski \cite{Grabowski1978}. T. Siebert \cite{Siebert1996} further developed the work of J. Grabowski \cite{Grabowski1978} (see also \cite{Grabowski-Kotov-Poncin2010,Grabowski-Kotov-Poncin2011,Grabowski-Kotov-Poncin2013} for the superalgebra case), employing a purely algebraic approach in the context of affine varieties. Siebert proved, as a byproduct, the result that an affine variety $W$ with a ring of coordinates $\mathcal O_W$ is smooth if and only if the Lie algebra $\mathrm{Der}(\mathcal{O}_W)$ is simple. The direction stating that the smoothness of $W$ implies the simplicity of $\mathrm{Der}(\mathcal{O}_W)$ was previously established by D.A. Jordan \cite{D.A.Jordan}. These foundational ideas have since inspired further research, emphasizing the utility of Lie-theoretic approaches in understanding both geometric and algebraic structures on a singular space.

The modern approach to geometric problems, as well as to mathematical physics ones, is unified by the language of $Q$-structures. The general motto is to replace a singular “hard” object with a collection of “gentle” ones, which are easier to work with. Typically, such a collection comes together with an algebraic structure equipped with a squared-zero operator $Q$. {Such objects are often referred to as \emph{manifolds up to homotopy} and admit several equivalent descriptions. Two particularly important incarnations are:
\begin{itemize}
    \item Lie $\infty$-algebroids, also known as Lie algebroids up to homotopy;
    \item $Q$-manifolds (or $Q$-varieties), also called differential graded (dg-) manifolds (or dg-varieties).
\end{itemize}
These two notions are equivalent, in the sense that they are dual to one another \cite{Voronov,Voronov2}.} In this framework, the original questions about the singular object might be encoded in the homology of the underlying differential $Q$.  A typical example in differential geometry arises when one considers the intersection of a section of a vector bundle with the base manifold. More precisely, let $E\to M$ be a vector bundle over a smooth manifold $M$ and $s \colon M\to E$ a section, and let $\Sigma=M\cap \mathrm{Graph}(s)=\{s=0\}$. The latter intersection is singular in general (if $0$ is not a regular value for $s$).

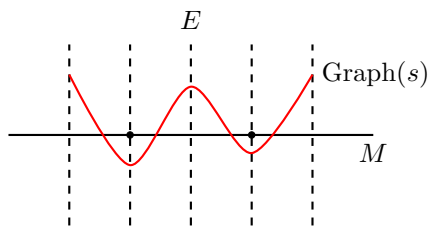
\begin{center}

\begin{tikzpicture}[scale=0.8] 

\draw[thick] (-3,0) -- (3,0);
\node[below] at (3,0) {$M$}; 

\foreach \x in {-2,-1,0,1,2} {
    \draw[thick,dashed] (\x,-1.5) -- (\x,1.5);
}

\draw[thick,red] plot [smooth] coordinates {
    (-2,1) (-1,-0.5) (0,0.8) (1,-0.3) (2,1)
};
\node[right] at (2,1) {$\mathrm{Graph}(s)$};


\node[above] at (0,1.6) {$E$};

\end{tikzpicture}
\end{center}
There exists a rather unusual Lie $\infty$-algebroid, or equivalently, a negatively graded differential-graded manifold over $M$, which can be defined by $(\Gamma(\wedge^\bullet E[-1]^*), Q=\iota_s)$ where $\iota_s$ is the contraction with $s$. Here, the dual vector bundle $E[-1]^*\to M$ is concentrated in degree $-1$.   It is clear that the $0$-th homology is  $H^0(\Gamma(\wedge^\bullet E[-1]^*), Q)\simeq~C^{\infty}(M)/\mathcal{I}_s$ with $\mathcal{I}_s$ being the ideal generated by functions  of the form  $\langle \xi, s\rangle$ for $\xi\in \Gamma(E^*)$. {If $s$ is transverse to the zero section, then $\Sigma$ is smooth and the cohomology of $(\Gamma(\wedge^\bullet E[-1]^*), Q)$ is concentrated in degree $0$ and $H^0$ essentially recovers $C^\infty(\Sigma)$}. {More generally, if the ideal $\mathcal I_s$ is locally a complete intersection (equivalently, it is locally generated by a regular sequence), then the cohomology of
$(\Gamma(\wedge^\bullet E[-1]^*),Q)$ is again concentrated in degree $0$. Thus, in this case, the Koszul complex provides a differential graded resolution of
$C^\infty(M)/\mathcal I_s$, and hence a differential graded model for the (possibly singular) space $\Sigma$.} 
In a similar fashion, it is possible to adjust the construction for the common zero locus of several sections $s_1, \ldots, s_k\in \Gamma(E)$. For example, for $M=\mathbb C^d$ and $E=\mathbb C^d\times \mathbb C^n$, the set $\Sigma\subset \mathbb C^d$ is the zero locus of some functions $\varphi_1,\ldots, \varphi_n\in \mathbb C[x_1, \ldots, x_d]$. Here, we have \[Q=\sum_{i=1}^n\varphi_i\frac{\partial}{\partial \eta_i}\;\;\text{and} \;\; H^0((\Gamma(\wedge^\bullet E[-1]^*), Q))\simeq C^{\infty}(M)/\langle \varphi_1,\ldots, \varphi_n \rangle.\] where  $\eta_1,\ldots, \eta_n$ corresponding to the dual basis of $E[-1]$. {In general, the differential graded manifold
\(
(\Gamma(\wedge^\bullet E[-1]^*),Q)
\) has nontrivial cohomology in negative degrees}. However, one can always choose a negatively graded dg-variety so that there is no cohomology in degree less than or equal to –1, namely a Koszul-Tate resolution of ${C^\infty(M)}/\mathcal{I}_\Sigma$. More generally, to any  unital commutative algebra $\cO$ and a proper ideal $\cI \subset \cO$ one can associate a Koszul-Tate resolution of $\cO/\cI$, which is a graded symmetric dg-algebra $(S_{\cO}(\oplus_{i\leq -1} \cV_i), \delta)$ acyclic in negative degrees \cite{Tate}. Such a negative $Q$-variety is a cornerstone of BV and BFV formalisms in mathematical physics \cite{HT,stasheff_poisson, Felder-Kazhdan}. {In particular, a Koszul-Tate resolution is a standard ingredient to construct a corresponding classical BFV (BRST) data using homological perturbation techniques. Even on the classical level, BFV (BRST) (co)homology contains important data such as generalized conservation laws, gauge anomalies, and possible deformations (interaction vertices), see e.g.,  \cite{Barnich_2000} for a detailed analysis for the Yang-Mills gauge theory.}

From the infinitesimal point of view, a (regular) foliation \(F \subset TM\) may be encoded as a positively graded \(Q\)-manifold -- the leaf-wise de Rham complex \((\Gamma(\wedge^\bullet F^{*}), d^{{DR}})\). The degree-\(0\) cohomology of the latter identifies with the algebra of functions on the (typically singular) leaf space \(M/F\).  Similar constructions exist for foliations with singularities, i.e.,  \emph{singular foliations} \cite{Androulidakis-Iakovos-Georges, LLL1} and for vector fields on singular spaces: one replaces \(F\) by an appropriate Lie \(\infty\)-algebroid. Recently,  Laurent-Gengoux, Lavau, and Strobl \cite{LLS} introduced an important framework by associating a class of positively graded $Q$-varieties (equivalently \underline{negatively} graded Lie $\infty$-algebroids) to any Lie subalgebra $\mathfrak A$ of the sheaf of vector fields $\mathfrak X$ on a manifold $M$ that admits a geometric resolution, i.e., those for which there exists an anchored complex of vector bundles 

$$ (E,\dd,
\rho) \colon  \xymatrix{  \ar[r] & E_{-i-1} \ar[r]^{{\dd^{(i+1)}}} \ar[d] & 
     E_{ -i} \ar[r]^{{\dd^{(i)}}}
     \ar[d] & E_{-i+1} \ar[r] \ar[d] & \ar@{..}[r] & \ar[r]^{{\dd^{(2)}}}& E_{-1} \ar[r]^{\rho} \ar[d]& TM \ar[d] \\ 
      \ar@{=}[r] & \ar@{=}[r] M  &  \ar@{=}[r] M 
      &  \ar@{=}[r] M  &\ar@{..}[r] & \ar@{=}[r]   &  \ar@{=}[r] M  & M}
    $$
    such that the following complex of sheaves      \begin{equation}
           \label{}{\longrightarrow} \Gamma({E_{ -i-1}})
     \stackrel{\dd^{(i+1)}}{\longrightarrow} \Gamma({E_{-i}})
     \stackrel{\dd^{(i)}}{\longrightarrow}{\Gamma(E_{-i+1})}{\longrightarrow}\cdots  {\longrightarrow} \Gamma(E_{-1})  
     \stackrel{\rho}{\longrightarrow} \mathfrak A\longrightarrow 0
      \end{equation}  is exact.  It is quite natural to work with this class of Lie sublagebras, as it contains the class of Lie subalgebras that are (locally) real analytic \cite[Proposition 2.3]{LLS}. It is also a natural object in the holomorphic setting, since $\mathfrak A$ is then a coherent sheaf and such geometric resolutions always exist locally.  This construction admits a  generalization \cite{CLRL,louis2023universalhigherliealgebras} in a purely algebraic setting, where $\mathfrak A$ is an arbitrary Lie-Reinhart algebra. With this generalization, we associate  a class of positively graded $Q$-varieties to any singular space where the concept of vector fields is well-defined, e.g., an affine variety. 
 
 It is natural to ask 

\begin{question}\label{question:main}
    Is there a $\mathbb Z$-graded $Q$-variety $(\cA, Q)$ encoding both the Koszul–Tate resolution of $\cO/\cI$ (a singular space), and a positively graded $Q$-variety $(\mathcal{A}^+, Q^+)$ associated with vector fields on the singular space?
\end{question}
A positive answer to this question was given in \cite{KOTOV2023104908} in the smooth geometric setting, i.e., for an ideal $\cI$ of a ring of smooth functions $\cO = C^{\infty}(M)$. However, the result \cite{KOTOV2023104908} is only an existence result. It is typically expected that a generic Koszul-Tate resolution $(S_{\cO}(\oplus_{i\leq -1}\cV_i), \delta)$ has infinitely many non-zero $\cV_i$. Thus, finding examples is problematic. 
Therefore, it is of interest to study the following question: 
\begin{question}
    \label{question:main2}
    To what extent is the $\mathbb Z$-graded $Q$-variety $(\cA, Q)$ of Question \ref{question:main} computable? Is there an algorithm that terminates in a finite number of steps?
\end{question}
Questions similar to Question \ref{question:main2} were extensively studied in the realm of projective resolutions of modules. More precisely, in the works of \cite{peeva2010graded, Scarf} and references therein, an efficient machinery was developed to study monomial resolutions of polynomial rings. A manifest description reveals underlying structures, such as a semi-simplicial complex that governs such resolutions. Moreover, the construction of such resolutions and related computations can now be performed by software \cite{Macaulay}.
In what concerns Koszul-Tate resolutions, a recent work \cite{hancharuk:tel-04692988, hancharuk2024} suggested a construction built on top of projective resolutions of $\cO/\cI$. This construction, called an \emph{arborescent} Koszul-Tate resolution, significantly reduces the number of computations, which becomes finite if the projective resolution is of finite length and has finite ranks at each degree. A manifest description provides insight into the generic structure of any Koszul-Tate resolution \cite[Proposition 5.5]{hancharuk2024}. Namely, any Koszul-Tate resolution $(S_{\cO}(\oplus_{i\leq -1}\cV_i), \delta)$ of $\cO/\cI$, where $\cO$ is a polynomial ring and $\cI$ is a monomial ideal, has $\cV_i \neq 0$ for all $i$ provided that $\cI$ is not generated by a regular sequence.

We address Question~\ref{question:main2} in the following setting: as a negative dg-variety we use an arborescent Koszul-Tate resolution of $\cO/\cI$, while the positive graded $Q$-variety is either a generic one over $\cO/\cI$, or it is over $\cO$ and preserves $\cI$. In \S \ref{sec:1} we recall the notions of $Q$-varieties, and recall (arborescent) Koszul-Tate resolutions. In \S \ref{sec:2} we state the main results of the paper. In \S \ref{sec:3}, we give manifest examples of our construction. In \S \ref{sec:appendix} we substitute the arborescent Koszul-Tate resolution by the Koszul complex and draw a comparison from our construction to it. In \S \ref{sec:h.mult} we give examples of new structures (higher-order multiplications) that our construction induces. In \S \ref{Lifting.derivations} we address the issue of lifting derivations from quotient algebras.\\

\noindent
\textbf{Acknowledgements}. We would like to express our sincere gratitude to the anonymous referees for their careful reading of the manuscript and for their valuable comments and suggestions, which have improved the presentation of this work. We thank C.~Laurent-Gengoux, A.~Kotov, L.~Ryvkin and V.~Salnikov for helpful discussions and valuable comments on the manuscript. We are also grateful to E.~Lerman for insightful discussions. R.~Louis would further like to thank C.~Zhu and the University of Göttingen for their hospitality during his visit as a Guest Researcher under the Ministry of Science and Culture (MWK) research program (RTG 2491) in the summer of 2026. Finally, R.~Louis gratefully acknowledges the financial support provided through the J.~L.~Doob Research Assistant Professorship at the University of Illinois Urbana--Champaign (UIUC).\\

\noindent
\textbf{Statements and Declarations}. The authors have no conflict of interest that is relevant to the content of this article.
\\

\noindent
\textbf{Data Availability}. Data sharing is not applicable to this article as no datasets were generated or analyzed during the current study.

\section{Definitions}\label{sec:1}

Throughout this article, $ \mathcal O$ is a commutative unital algebra over a field $ \mathbb K$ of characteristic zero, and ${\mathrm{Der}}(\mathcal O) $ stands for its Lie algebra of $\mathbb{K}$-linear derivations. We denote by $ X[f]$ the derivation $X \in  {\mathrm{Der}}(\mathcal O) $ applied to $f\in \mathcal O$. Also, an $\mathcal{O}$-module $\cV$ is seen as a $\mathbb K$-vector space in the
natural way, for $\lambda\in \mathbb K$, {$\lambda \cdot v := (\lambda \cdot 1_{\mathcal{O}})\, v$, where $1_{\mathcal{O}} = 1$ is the multiplicative identity of $\mathcal{O}$.} 
In the sequel, we will drop the notation ``$\cdot$''.

Geometrically, $\mathcal{O}$ can be understood as the algebra of smooth functions on a manifold $M$, 
or on an open subset $U \subset M$ of a complex manifold, or the coordinate ring of an affine variety $W$.\\

\noindent
\textbf{Some notations}. Recall that a \emph{$\mathbb Z$-graded $\cO$-module} is an $\mathcal{O}$-module of the form  $\cV=\oplus_{i \in \mathbb Z}\cV_i$ for   $\mathcal{O}$-modules $\cV_i, i\in \mathbb Z$.  Elements of $\cV_i$ are called \emph{homogeneous} of \emph{degree} $i$. A \emph{graded algebra} $\cA=\oplus_{i\in \mathbb Z}\cA_i$ over $\cO$ is a graded $\cO$-module equipped with a  multiplication that respects the grading, i.e., $\cA_i \cdot \cA_j \subset \cA_{i+j}$. In this paper, we deal with the graded symmetric algebra of a graded module. To fix the sign conventions and notations, let us recall the construction. 

For a graded $\cO$-module $\cV$, we denote by $\lvert x\rvert \in \mathbb{Z}$ the degree of a homogeneous element $x\in \cV$. 

\begin{itemize}
    \item We denote by $S_{}(\cV)  $ and call \emph{graded symmetric algebra of $\cV $ over $\mathcal O$} the quotient of the tensor algebra over $\mathcal O $ 
$$ T_{}(\cV):= \oplus_{k =1}^\infty
\underbrace{ \cV \otimes_\mathcal O \cdots \otimes_\mathcal O \cV}_{\hbox{\small{$k$ times}}}
$$
 by the ideal generated by $x\otimes_{\mathcal O} y-(-1)^{\lvert x\rvert\lvert y\rvert}y\otimes_{\mathcal O} x$, with $x,y$ arbitrary homogeneous elements of $\cV $. We denote by $\odot_{\cO}$ the product in $S_{}(\cV)$. In order to simplify notations we omit the subscript $\cO$ of tensor products $\otimes_{\cO}, \odot_{\cO}$, so that these products $\otimes, \odot$ are understood to be over $\cO$.
\item We denote by $ S_R( \cV ) $ the \emph{graded symmetric algebra} of $\cV $ tensored with the quotient algebra $R = \cO/\cI$ for some proper ideal $\cI\subset \cO$, namely, $S_R( \cV ) = S( \cV )\otimes_\cO R$.
\item {Let $V, W$ be graded $\cO$-modules. By abusing the notation, we denote by $V\odot W$ the quotient $V\otimes W \oplus W\otimes V$ by the submodule generated by the relations $x\otimes y- (-1)^{|x||y|}y\otimes x$, where $x$ and $y$ are arbitrary homogeneous elements of $V$ and $W$ respectively. In other words, $V\odot W$ can be seen as the direct summand of $S(V\oplus W)$.}

\end{itemize} 

The algebra $S(\cV)$  comes equipped with different notions of \underline{degrees} that must not be confused. 
\begin{enumerate}
    \item  We define the \emph{degree} of $x= x_1 \odot \cdots  \odot x_n\in S^n(\cV)\coloneqq 
\underbrace{ \cV \odot \cdots \odot \cV}_{\hbox{\small{$n$ times}}}$  by 
$$  | x_1 \odot \cdots  \odot x_n | = |x_1|+ \cdots  + |x_n|$$
for any homogeneous elements $ x_1, \dots, x_n\in\cV$. With respect to this degree, $S(\cV) $ is a  $\mathbb Z$-graded commutative algebra over $\cO$. For any homogeneous elements $x_1 , \ldots , x_k \in \cV$ and $\sigma\in\mathfrak{S}_k$ a permutation of $\{1, \ldots, k\}$, the \emph{Koszul sign} $\epsilon(\sigma; x_1 , \ldots , x_k )$ is defined by:
$$  x_{\sigma(1)} \odot \cdots \odot x_{\sigma(k)}= \epsilon(\sigma; x_1 , \ldots , x_k ) \, x_1 \odot \cdots \odot x_k.$$   We write $\epsilon(\sigma )$ for $ \epsilon(\sigma; x_1 , \ldots , x_k )$.
\item The \emph{polynomial degree} of $ x_1 \odot \cdots  \odot x_n\in S^n(\cV)$ is defined to be $n$. Elements of polynomial degree $k$ and degree $d$ in $ S(\cV)$ are denoted by $S^k(\cV)_d$.
\item There is also a notion of “positive” and “negative” gradings in $S(\cV)$. For an element $a \in \cV_i$ we prescribe them as follows:
 $$|a|_- =
      \begin{cases}
          0, \quad \hbox{if} \; i\geq 0, \\
          -i, \quad \hbox{if} \; i<0.
      \end{cases}
      $$
     $$|a|_+ =
      \begin{cases}
          i, \quad \hbox{if } i \geq0, \\
          0, \quad \hbox{if } i < 0.
      \end{cases}
      $$
      So that $|a| = |a|_+ - |a|_-$. The degrees $|\,\cdot \,|_+\,\,,\,\, |\,\cdot \,|_-$ are extended to elements of $S(\cV)$ by the formula $|ab|_k = |a|_k + |b|_k$, where $a,b$ are monomials of $S(\cV)$, with $k \in \lbrace +, - \rbrace$.
      A homogenous element of negative degree $i$ is denoted as $w_{(i)}$. For instance, elements of  $\cV_{1}\odot \cV_{-4}$ are of \begin{enumerate}
      \item  degree $-3$;
      \item polynomial degree $2$;
          \item negative degree $4$;
          \item positive degree $1$.
          
      \end{enumerate}

\end{enumerate}
\begin{itemize}
    \item Every linear map $\Phi\colon S(\cV) \to S(\cV')$ between graded symmetric algebras admits a  decomposition w.r.t the degree, as well as the negative degree. We  denote by  $\Phi_k$ (respectively $\Phi_{(k)}$) the component of $\Phi$ of degree $k$ (negative degree $k$), which means that it sends an element $a_i$ to an element of degree $i+k$ (or $a_{(i)}$ to $(\Phi(a_{(i)}))_{(i+k)}$).

\item A \emph{derivation of degree} $k$ is a $\mathbb K$-linear map $Q\colon S(\cV) \to S(\cV)$  satisfying the graded Leibniz rule, i.e., 
      $$
      Q(a\odot b) = Q(a)\odot b + (-1)^{k|a|}a\odot Q(b)
      $$
      for homogeneous elements $a$ and $b$ of $S(\cV)$. Such a derivation admits a decomposition with respect to the negative degree $$Q = \sum_{i\in \mathbb  Z}Q_{(i)}$$ where $Q_{(i)}\colon S(\cV)\to S(\cV)$ is a derivation of $S(\cV)$ of negative degree $i$ and {of degree $k$}.

\end{itemize}

\subsection{$\mathbb Z$-graded $Q$-varieties} 

When working with $\mathbb Z$-graded algebras, one is quickly led to consider infinite sums of homogeneous elements. It therefore becomes essential to specify which such infinite sums are meaningful, that is, to determine a notion of convergence. For instance, in the smooth geometrical setting, for $\mathbb Z$-graded manifolds, one would like to perform change of coordinates, e.g., $x_i \mapsto x_i + \eta \mathcal P$, where $x=(x_1, \ldots, x_n)$ is a local coordinate on a manifold $M$, and $\eta, \mathcal P$ are graded coordinates  of degrees $2$ and $-2$ respectively. Then the Taylor series for $\sin (x_1 + \eta \mathcal P)$ does not make sense\footnote{We are grateful to Vladimir Salnikov for pointing out this example.} in $S(\cV)$. In our paper,  this issue arises when doing homological perturbation -- we obtain homogeneous components of an $\cO$-linear map $\Phi$. In particular, if there are infinitely many components $\Phi_{(k)}$ of negative degree $k$ and of fixed degree $i$, then the sum $\sum_{k\in \mathbb Z} \Phi_{(k)}(a)$, $a\in \cV$, does not make sense in $S(\cV)$. Therefore it is necessary to find an appropriate completion of $S(\cV)$ for which  our construction converges. This necessity naturally leads to the introduction of filtrations and their associated topologies. For the sake of clarity in our notation and conventions, we recall the following
\begin{deff}
\label{def:filtration}\cite{Eisenbud1995,Bourbaki2006}
Let $\mathcal{G} = \oplus_{j\in \mathbb Z} \mathcal{G}_j$ be a $\mathbb Z$-graded commutative unital algebra. The \emph{negative filtration} is the filtration $$\cdots \subset F^{i+1}\mathcal{G} \subset F^{i} \mathcal{G} \subset \dots \subset  F^1 \mathcal{G} \subset  F^0\mathcal{G}=\mathcal{G}$$
where for every $i\in \mathbb Z$, $F^i\mathcal{G}$ is an ideal of $\mathcal{G}$ generated by elements of degree less or equal to $-i$, that is 
$$
F^i\mathcal{G} = \mathcal{G}\cdot(\oplus_{j\leq -i}\mathcal{G}_{j})
$$

\begin{enumerate}
\item   The cosets $(g+F^i\mathcal{G})_{g\in\mathcal{G},\,i\in \mathbb Z}$ generate a topology $\tau_F$ on $\mathcal{G}$, called the \emph{negative filtered topology}.
    \item The \emph{completion} of $\mathcal{G}$ is the $\mathbb  Z$-graded unital commutative algebra $\hat{\mathcal{G}}$ obtained as  the completion\footnote{Cauchy sequences are defined w.r.t the neighborhoods of $0$ in $\mathcal{G}$ of the negative filtered topology: $(g_n)_{n\in \mathbb{N}}$ is a \emph{Cauchy} sequence of $\mathcal{G}$ if for all $i\geq 0$, there exists $n_0\in \mathbb N$ such that $ n, p\geq n_0\Longrightarrow g_n-g_p\in F^i\mathcal{G}$.} w.r.t the negative filtered topology $\tau_F$. Equivalently, in terms of the graded projective limit of $\mathcal{G}$ 
\begin{equation}\label{eq:proj_lim_R1}
 \hat{\mathcal{G}}=\varprojlim_{j} \Big(  \mathcal{G}/F^j \mathcal{G} \Big).
\end{equation}

\end{enumerate}
\end{deff} 

\begin{remark}\label{rmk:completion} In Definition \ref{def:filtration},
    \begin{enumerate}
    \item we have $\bigcap_{i\geq 0} F^i\mathcal{G}=\{0\}$. This means that the negative filtered topology is Hausdorff. 
    
    \item some examples of elements of $\hat{\mathcal{G}}$ are: for a family $(g_i\in \mathcal{G}_i)_{i\in \mathbb Z}$, the infinite sums \(\sum_{i\geq 1}g_ig_{-i}\); \(\sum_{i\geq 1}g_{-i}\) are well-defined in $\hat{\mathcal{G}}$ since $(\sum_{i=1}^ng_ig_{-i})_{n\in \mathbb N}$; $(\sum_{i=1}^ng_{-i})_{n\in \mathbb N}$ are  Cauchy sequences in $\mathcal{G}$ with respect to the filtered topology. However, an infinite sum of the form  $\sum_{i\geq 1}g_0$ is not well-defined in  $\hat{\mathcal{G}}$ for $g_0\neq 0$. 
        \item if $\mathcal{G}$ is positively graded (that is, $\mathcal{G}_i=\{0\}$ for $i\leq -1$) or if $\mathcal{G}$ has only finitely many nonzero negative or positive degrees, then $\hat{\mathcal{G}}=\mathcal{G}$ since the Cauchy sequences in $\mathcal{G}$ (w.r.t in the negative filtered topology) eventually are constant. However, if  $\mathcal{G}$ is negatively graded (that is, $\mathcal{G}_i=\{0\}$ for $i\geq 1 $) and not bounded below, then $\mathcal{G}\subsetneq \hat{\mathcal{G}}$. For example, take a  Laurent-type algebra  $\mathcal{G}=\mathbb K[x^{-1}]$ with $\mathrm{deg}(x^{-1})=-1$. 

        \item In a graded commutative algebra $\mathcal{G}$ endowed with the negative filtration, the exponential $\mathrm{exp}(g)=\sum_{n=0}^{\infty}\frac{g^n}{n!}$ is well-defined for every $g\in F^i\mathcal{G}$ for some $i\geq 1$. 
    \end{enumerate}
\end{remark}

If our graded commutative unital algebra is concentrated in negative degrees, we do not want the completion to alter the algebra, as happens in the Laurent-type example in Remark \ref{rmk:completion} (2). To remedy this, we shall consider completion by degree, i.e., we complete at each degree separately with respect to the negative filtration. More precisely, let $\mathcal{G} = \oplus_{j\in \mathbb Z} \mathcal{G}_j$ be a $\mathbb Z$-graded commutative unital algebra and $\hat{\mathcal{G}}$ its completion w.r.t the negative filtered topology. We shall denote by $\Bar{\mathcal{G}}=\oplus_{j\in \mathbb Z} \Bar{\mathcal{G}}_j$ the $\mathbb Z$-graded commutative unital subalgebra of $\hat{\mathcal{G}}$ obtained by a “degree by degree” completion of $\mathcal{G}$ w.r.t the negative filtered topology, that is, for every $j\in \mathbb Z$

\[\bar{\mathcal{G}}_j=\{\text{the equivalence classes of all the Cauchy sequences of $\mathcal{G}_j$}\}.\]

Now, \begin{equation}\label{eq:proj_lim_R}
 \bar{\mathcal{G}}= \oplus_{i\in \mathbb Z} \varprojlim_{j}\Big(\mathcal{G}_i/F^j \mathcal{G}_i\Big)
\end{equation}
here, $F^j \mathcal{G}_i$ stands for the vector space of elements of degree $i$ in $F^j \mathcal{G}$.

\begin{remark}Notice that 
\begin{enumerate}
    \item     if $\mathcal{G}$ is negatively/positively graded (that is, $\mathcal{G}_i=\{0\}$ for $i\geq 1\;/\;i\leq -1$) or if $\mathcal{G}$ has only finitely many nonzero negative or positive degrees, then $\Bar{\mathcal{G}}=\mathcal{G}$ since the Cauchy sequences of homogeneous degree in $\mathcal{G}$ (w.r.t in the negative filtered topology)  are eventually constant.

    \item the infinite sum, $\sum_{i\geq 1}g_{-i}$ is well-defined in $\hat{\mathcal{G}}$ but not in $\Bar{\mathcal{G}}$. See \cite{KotovSalnikov} for more details. 
\end{enumerate}
\end{remark}

\begin{convention}
    In the sequel, all completions of graded commutative algebras will be “degreewise completion” taken with respect to the negative filtered topology introduced in Definition~\ref{def:filtration}.
\end{convention}


The following proposition is important and will be used to ensure convergence  in our main results.
\begin{proposition}
    Let $\mathcal{G} = \oplus_{j\in \mathbb Z} \mathcal{G}_j$ be a $\mathbb Z$-graded commutative unital algebra, and $(F^i\mathcal{G})_{i\in \mathbb N}$  be the negative filtration. Every derivation $Z \colon \mathcal{G}\to \mathcal{G}$ of degree $k$  extends to a well-defined derivation $\bar Z\colon \bar{\mathcal{G}}\to\bar{\mathcal{G}}$ on the completion.
\end{proposition}
\begin{proof}
Since $Z$ has degree $k$, it shifts the degree of homogeneous elements by $k$, and therefore preserves the negative filtration: for every $i\in\mathbb N$,
\[
Z(F^{i}\mathcal{G}) \subseteq
\begin{cases}
F^{\,i-k}\mathcal{G}, & k \ge 1,\quad \text{where } F^{\,i-k}\mathcal{G} := \mathcal{G} \text{ if } i-k \le 0,\\[0.3em]
F^{\,i}\mathcal{G}, & k \le 0.
\end{cases}
\]
In particular, $Z$ sends Cauchy sequences (for the filtered topology) to Cauchy sequences. Hence, $Z$ is continuous and extends uniquely to a derivation $\bar Z$ on the completion $\bar{\mathcal{G}}$. 
\end{proof}

We consider the following definition of a $\mathbb Z$-graded manifold/variety. Their properties 
are carefully studied in \cite{KotovSalnikov}.

The following definition is taken from \cite{KOTOV2023104908}. For the
original introduction of $\mathbb Z$-graded manifolds, see
\cite{voronov_graded_2002-1}; see also \cite{Voronov_2019}.
\begin{definition}
\label{Z_graded_manifold}
    A \emph{$\mathbb Z$-graded manifold} on a manifold $M$ is a pair $ (M, \cA)$, where 
    $\cA= \oplus_{i\in \mathbb Z}\cA_i$ is a sheaf of $\mathbb Z$-graded
commutative algebras over $M$ (referred to as its sheaf of \emph{functions}), such that every point
of $M$ has a neighborhood $U \subseteq M$ over which $\mathcal \cA (U)$ is isomorphic to  $ \Gamma( \Bar{S}(\oplus_{i\in \mathbb Z} V_{i}))$, where $\Gamma(\Bar{S}(\oplus_{i\in \mathbb Z} V_{i}))$ is the degreewise completion of the graded symmetric algebra $\Gamma( S(\oplus_{i\in \mathbb Z} V_{i}))$, and each  $V_i$ is a vector bundle over $U$ concentrated in degree $i$.

\begin{enumerate}
    \item If the grading of $(M,\mathcal A)$ goes from $-\infty$ to $0$, then $(M,\mathcal A)$ shall be referred to a \emph{negatively graded manifold}.

    \item If the grading of $(M,\mathcal A)$ goes from $0$ to $+\infty$, then $(M,\mathcal A)$ shall be referred to a \emph{positively graded manifold}.
\end{enumerate}

\end{definition}
\begin{convention}
    Throughout this paper, we shall assume\footnote{In other words, there are no  generators of $\cA$ of degree $0$  besides $\mathcal{O}_M$ \cite[Remark 1.5]{KOTOV2023104908}.} that $V_0=\{0\}$ of Definition \ref{Z_graded_manifold},  and  therefore omit it from the list $(V_i)_{i\in\mathbb{Z}}$.
\end{convention}

In Definition \ref{Z_graded_manifold}, the vector bundles $(V_i)_{i\in \mathbb Z^\times}$ are only defined on an open neighborhood of a point. However,  Batchelor’s theorem for $\mathbb Z$-graded manifold \cite{Pavol-Severa, Voronov2,KOTOV2023104908, KotovSalnikov} says in the smooth case 
that they can be glued into global vector bundles $(E^*_i)_{i\in \mathbb Z^\times}$ on $M$ such that $\cA\simeq \Gamma( S(\oplus_{i\in \mathbb Z^\times} E^*_{i}))$. The vector bundles $(E_{i})_{i\in \mathbb Z^\times}$ are canonically defined by Serre-Swan theorem so that the modules of sections of their duals $(E_i^*)_{i\in \mathbb Z^\times}$ are $\mathcal{I}/\mathcal{I}^2=\left(\mathcal{I}_{-i}/(\mathcal{I}^2)_{-i}\right)_{i\in \mathbb Z^\times}$, where $\mathcal{I}=\mathcal{I}_{-}+\mathcal{I}_{+}$ with $\mathcal{I}_{-}$ resp. $\mathcal{I}_{+}$ is the  ideal  of $\cA$ generated by negative/positive homogeneous degree functions \cite{KOTOV2023104908, KotovSalnikov}. In the sequel, our graded manifold will be of this form, under a choice of such an isomorphism. This choice will be referred to as a \emph{splitting}.

This justifies an algebraic version of the Definition \ref{Z_graded_manifold}.
\begin{definition}\label{def:Z-variety}
    A \emph{$\mathbb Z$-graded variety} over a commutative unital algebra\footnote{In algebraic geometry, $\mathcal{O}$ stands for the coordinate ring of an affine variety.} $\cO$ is a $\mathbb Z$-graded commutative algebra $\cA \coloneqq \bar{S}(\oplus_{i\in \mathbb Z^\times}\cV_i)$, where each $\cV_i$ is a projective $\cO$-module and $\bar{S}(\oplus_{i\in \mathbb Z^\times}\cV_i)$ is the degreewise completion of ${S}(\oplus_{i\in \mathbb Z^\times}\cV_i)$ with respect to the negative filtered topology. 
\end{definition}

{{Given a $\mathbb{Z}$-graded variety 
\(
\mathcal{A} \coloneqq \bar{S}\!\left(\oplus_{i \in \mathbb{Z}^\times} \mathcal{V}_i\right),
\)
the derivations of $\mathcal{A}$, equipped with the graded commutator, form a graded Lie algebra $\mathrm{Der}(\mathcal{A})$.}} 
\begin{convention}
    The product $\odot$ of  the symmetric algebra ${S}\!\left(\oplus_{i \in \mathbb{Z}^\times} \mathcal{V}_i\right)$ is extended to the completion $\bar{S}\!\left(\oplus_{i \in \mathbb{Z}^\times} \mathcal{V}_i\right)$ and is denoted by $\bar \odot$.
\end{convention}

\begin{definition}\label{Z_graded-Q}
A \emph{$\mathbb Z$-graded $Q$-variety}, or \emph{differential graded (dg-) variety},\footnote{When $\mathcal{O}$ is the algebra of functions on a manifold, each $\mathcal V_i$ is the  $\mathcal{O}$-module of sections of a vector bundle over $M$, it shall be called a $Q$-manifold.} $(\cA, Q)$ is a $\mathbb Z$-graded variety $\cA = \Bar{S}(\oplus_{i\in \mathbb Z^\times}\cV_i)$ equipped with a homological vector field $Q\in \mathrm{Der}(\mathcal A)$ of degree $+1$, i.e., $Q^2=\frac{1}{2}[Q,Q]=0$.

    We shall speak of \emph{negatively} or \emph{positively graded $Q$-variety} when $\cA = S(\oplus_{i<0}\cV_i)$ or $\cA = S(\oplus_{i>0}\cV_i)$. In this case, the completion $\Bar{S}(\oplus_{i<0}\cV_i)$ (or $\Bar{S}(\oplus_{i>0}\cV_i)$) coincides with the graded symmetric algebra
$S(\oplus_{i<0}\cV_i)$ (or $S(\oplus_{i>0}\cV_i)$).
\end{definition}

\begin{remark}Notice that\begin{enumerate}
    \item the homological vector field $Q$ of a $\mathbb Z$-graded $Q$-variety $(\cA, Q)$ ,  admits a decomposition with respect to the negative degree $$Q = \sum_{i\in \mathbb  \{-1\}\cup\mathbb N}Q_{(i)}$$where $Q_{(i)}\colon \mathcal{A}\to \mathcal{A}$ is a {vector field} (derivation of $\cA$) of negative degree $i$ and {of degree $+1$}.  The sum is allowed to be pointwise infinite. 
    \item if $(\mathcal{A}, Q)$ is negatively graded then $Q=Q_{(-1)}$. 
\end{enumerate}
\end{remark}

Given a split $\mathbb Z$-graded manifold $\mathcal
A\simeq \Gamma (\bar S(\oplus_{i\in \mathbb Z^\times} E^*_i))$ we shall denote by $\mathfrak X({E})$  the $\mathbb Z$-graded derivations of $\mathcal{A}$, they shall be called \emph{(graded) vector fields} on $E$. 

\begin{example}\label{ex:Lie-alg}
\begin{enumerate}
    \item Every Lie algebroid  $(A,[\cdot\,,\cdot\,]_A,\rho)$ ($A$ is concentrated in degree $-1$) corresponds to a positively graded $Q$-variety $(S(\Gamma(A^*)), Q)$ where the homological vector field $Q\in\mathfrak X(A)$ is given by
\begin{align*}
    \langle Q[f], a\rangle&=\rho(a)[f]\\
    \langle Q[\xi], a\odot b \rangle&= \rho(a)[ \langle \xi,b\rangle]-\rho(b)[ \langle \xi,a\rangle] -\langle \xi,[a,b]_A\rangle
\end{align*}
 $f\in \cO$ and $\xi\in \Gamma(A^*),\;a,b\in \Gamma(A)$, where $\cO$ is the ring of smooth functions on the base manifold $M$. Here, $Q$ acts on $S(\Gamma(A^*))$ by means of the graded Leibniz rule. The Jacobi identity for  $[\cdot\,,\cdot\,]_A$ is equivalent to $Q^2 = 0$. In particular, the tangent bundle $TM$ of a manifold $M$ is a Lie algebroid, the corresponding positively graded $Q$-variety is the de Rham complex $(\Omega^\bullet(M), d^{dR})$.

\item More generally,  recall that a \emph{negatively graded Lie $\infty$-algebroid} $\left(A_\bullet,(\ell_k)_{k\geq 1}, \rho\right)$ is a collection of vector bundles $A_\bullet =(A_{-i})_{i\geq 1}$ over $M$ endowed with a sheaf of Lie $\infty$-algebra
structures $(\ell_k)_{k\geq 1}$ over the sheaf of sections of $A_\bullet$ together with a vector bundle morphism $\rho\colon A_{-1}\to TM$, called the \emph{anchor map}, such that the $k$-ary-brackets $$\ell_k: \underbrace{\Gamma(A_\bullet)\times \cdots \times \Gamma(A_\bullet)}_{k {\text{-times}}}\longrightarrow \Gamma(A_\bullet)$$ are all $\cO$-linear 
in each of their arguments except when $k=2$ and at least one of the arguments of $\ell_2$ is of degree $-1$. The $2$-ary bracket  satisfies the Leibniz identity

\begin{equation}
    \ell_2(x, f y) = \rho(x)[f]y + f\ell_2(x, y),\; f\in C^{\infty}(M),\; x \in \Gamma(A_{-1}), y\in\Gamma(A_\bullet).
\end{equation}

If $A_{-i}=\{0\}$ for $i\geq n+1$  then we speak of \emph{Lie $n$-algebroid} over $M$. Equivalently, as in the case of Lie algebroids, a Lie $n$-algebroid corresponds to a derivation
\[
Q \colon \mathcal A=\Gamma\!\left(S(A_\bullet^{*})\right)
   \longrightarrow 
   \Gamma\!\left(S(A_\bullet^{*})\right)
\]
of total degree $+1$ on $A_\bullet \to M$, satisfying $Q^{2} = 0$. See \cite{Voronov,Voronov2, Poncin} or \cite{LLS} for more details. {There is a more general notion of Lie $\infty$-algebroid over a manifold $M$ in which the higher brackets are not required to be $C^\infty(M)$-multilinear \cite{khudaverdian2008higher}. The Lie $\infty$-algebroids considered in this example correspond to the special case in which all higher anchors vanish, so that the higher brackets are $C^\infty(M)$-multilinear.}
\end{enumerate}   
\end{example}

\subsubsection{The negative and positive components of a  $\mathbb Z$-graded $Q$-variety}\label{+and-components}
Let $(\cA, Q)$ be a $\mathbb Z$-graded $Q$-variety. The pair $(\cA, Q)$ breaks into two main components: a negative component, which is a negatively graded $Q$-variety that is denoted by $(\cA^-,Q^-)$, and a positive component, which is a $\mathbb N$-graded $Q$-variety as well, that we denote by $(\cA^+,Q^+)$. Let us now recall how these components are defined.\\

\begin{enumerate}
    \item \textbf{The negative {component} \((\cA^-,Q^-)\) of $(\cA, Q)$}. This component is essentially constructed by modding out the ideal of positively graded functions. More precisely, consider an ideal $\mathcal{I}_+$ of $\cA$ generated by positive homogeneous degree functions  $\cA_{\geq +1}$.
     We define $\cA^-$ to be the quotient $\mathcal A/\mathcal I_+$. Since the degree of the homological vector field  $Q$ is   $+1$, it preserves the ideal $\mathcal{I}_+$ , i.e.,  $Q[\mathcal I_+]\subset \mathcal{I}_+$. Hence,  $Q$ passes to quotient and yields a well-defined homological vector field $Q^-$ of degree $+1$. Therefore,  the pair $(\cA^-, Q^-)$ is a  negatively graded $Q$-variety that we call the \emph{negative component of $(\cA, Q)$}.  Notice that the vector field $Q^-$ is vertical, i.e., it is  $\mathcal{O}$-linear.

\item \textbf{The positive component \((\cA^+,Q^+)\) of $(\cA, Q)$}. Consider the ideal $\mathcal{I}_-$ of $\mathcal A$ generated by negative homogeneous degree functions   $\cA_{\leq -1}$. Unlike the negative component of $(\cA, Q)$, the vector field $Q$ does not immediately descend to the positively graded manifold $\cA/\mathcal I_-$, since $Q(\cA_{-1})\subset \cA_0$. 
However, the ideal $\mathcal{I}_-+Q[\mathcal{I}_{-}]$ is preserved by $Q$. $\cA^+$ is defined as the quotient $\cA/(\mathcal{I}_-+Q[\mathcal{I}_{-}])$. The explicit description of $\cA^+$ is as follows: let $\cI$ be an ideal of $\cO$ given by the image
\begin{tikzcd}
   Q[\cI_{-}] \arrow[r, "\mathrm{pr}_{(0)}"] &  \cO,
\end{tikzcd} 
where $\mathrm{pr}_{(0)}$ is the projection on negative degree $0$. Then a straightforward observation shows that $\mathcal{I}_-+Q[\mathcal{I}_{-1}]=\mathcal{I}_-+\mathcal{I}\cA$, \cite[Lemma 2.12]{KOTOV2023104908}. Therefore, if $\mathcal{A}= \bar S({\oplus_{i\in \mathbb Z^\times}\mathcal{V}_i})$, then the \emph{positive component} $(\cA^+, Q^+)$ of $(\cA, Q)$ is identified as $\cA^+ = S_{\cO/\cI}(\oplus_{i\geq 0}\cV_i)$ and $Q^+$ is $Q$ descended to the quotient.
\end{enumerate}

\subsection{Lie-Rinehart algebras and positively graded $Q$-varieties}\label{sec:LR-NQ-manifold}
Let $\mathcal O$ be a commutative unital algebra over $\mathbb K=\mathbb R$ or $\mathbb C$. Which geometrically may correspond to an algebra of admissible functions of a  manifold or an affine variety.

\begin{definition}\label{def:Lie-Rrt}
A \emph{Lie-Rinehart algebra} over $ \mathcal O$ is a triple $(\mathfrak A , [\cdot, \cdot]_\mathfrak A , \rho_{\mathfrak A})$ with $ \mathfrak A$ an $ \mathcal O$-module, $[\cdot, \cdot]_\mathfrak A  $ a Lie $\mathbb{K}$-algebra bracket on  $\mathfrak A $, and $ \rho_\mathfrak A \colon \mathfrak A \longrightarrow  {\mathrm{Der}}(\mathcal O)$ a $ \mathcal O$-linear Lie algebra morphism called \emph{anchor map}, satisfying the the so-called \emph{Leibniz identity}:
 $$   [  a,  f b ]_\mathfrak A  = \rho_\mathfrak A (a ) [f] \, b + f [a,b]_\mathfrak A  \hbox{ for all $ a,b \in \mathfrak A, f \in \mathcal O$}. $$
\end{definition}

\textbf{Restriction.}
Consider a Lie-Rinehart algebra $(\mathfrak A, [\cdot, \cdot]_\mathfrak A , \rho_{\mathfrak A})$  over $\mathcal O $.
For every \emph{Lie-Rinehart ideal} $\mathcal I \subset \mathcal O$, i.e. any ideal
such that
$$ \rho_\mathfrak A (\mathfrak A) [\mathcal I] \subset \mathcal I$$
the quotient space $\mathfrak A / \mathcal I \mathfrak A $ inherits a natural Lie-Rinehart algebra structure over $\mathcal O / \mathcal I$. We call this Lie-Rinehart algebra the \emph{restriction w.r.t the Lie-Rinehart ideal $ \mathcal I$}. In the context of affine varieties or an arbitrary subset $\Sigma\subseteq M$ of a manifold, when $\mathcal I$ is the ideal of functions vanishing on $\Sigma$, we shall denote $\frac{\mathfrak A}{\mathcal{I}\mathfrak A}$ by $\mathfrak i_{\Sigma}^* \mathfrak A $.

\begin{example}
    $\mathfrak A=\mathrm{Der}(\mathcal{O})$ is a Lie-Rinehart algebra whose anchor map is the identity. Lie subalgebras $\mathfrak{F}\subseteq \mathrm{Der}(\mathcal{O})$ that are  finitely generated as $\mathcal{O}$-modules are Lie-Rinehart algebras whose anchor map is the inclusion map.  Geometrically, when $\mathcal{ O}$ is the algebra of functions on a manifold,  these Lie-Rinehart algebras are called \emph{singular foliations} on $M$, see e.g., \cite{Debord,AndroulidakisIakovos} or \cite{LLL1}.
\end{example}
\begin{example}[$\mathbb Z$-graded $Q$-varieties induce Lie-Rinehart algebras] Consider a $\mathbb Z$-graded $Q$-variety $(\cA=\Gamma(\bar S(\oplus_{i\in \mathbb Z^\times} E^*_{-i})),Q)$ on a manifold $M$. The homological vector field $Q$ admits a formal decomposition by polynomial degrees of the form $$Q=\sum_{i\geq -1}Q^{(i)}$$where $Q^{(i)}$ stands for the homogeneous polynomial degree $i$ component. The component of polynomial degree 
$-1$ is the contraction with a section of $E_{+1}\to M$, that is,  \[Q^{(-1)}=\iota_{c}\;\;\;\;\; \text{for some}\; c\in \Gamma(E_{+1}).\] This section determines an ideal $\cI\subset \cO$, defined as the image of
\[\iota_{c} \colon \Gamma(E^*_{+1})\longrightarrow  \mathcal{O},\qquad \cI=~\{\langle \alpha, c\rangle, \alpha\in \Gamma(E_{+1}^*)\}.\]

\begin{enumerate}
    \item There is an $\cO$-submodule $\mathfrak F\subseteq \mathfrak X(M)$ given by the image of an \emph{anchor map} $\mathfrak F=\rho_1(\Gamma(E_{-1}))\subseteq \mathfrak X(M)$,  where $\rho_1\colon E_{-1}\to TM$ is  determined by the identity $\langle Q(f), e\rangle=\rho_1(e)[f]$ for all $e\in \Gamma(E_{-1})$ and $f\in C^\infty(M)$. The submodule $\mathfrak F$ is included in the module $\mathfrak X_\mathcal{I}(M)=\{X\in \mathfrak{X}(M)\,|\, X[\mathcal{I}]\subseteq \mathcal{I} \}$ of vector fields that are “tangent to the zero locus of $\cI$": this inclusion follows  from $Q^{2}=0$ applied to a degree $-1$ function $\alpha\in \Gamma(E_{+1}^*)$ and a section $e\in\Gamma(E_{-1})$ of degree $-1$.\item The submodule $\mathfrak F\subseteq\mathfrak X_\cI(M)\subset\mathfrak X(M)$ is not, in general, closed under the Lie bracket of vector fields. However, its restriction $\mathfrak A:=\frac{\mathfrak F}{\mathcal{I}\mathfrak F}$ is closed under the Lie bracket. Consequently, $\mathfrak A=\frac{\mathfrak F}{\mathcal{I}\mathfrak F}$ is a Lie-Rinehart algebra over $\cO/\cI$. We call $\mathfrak A$  the \emph{basic Lie-Rinehart algebra} of the $\mathbb Z$-graded variety  $(\cA,Q)$ over $\cO/\cI$.
\end{enumerate}

\end{example}

It is natural to ask whether, given a Lie--Rinehart algebra $\mathfrak{A}$ over $\mathcal{O}/\mathcal{I}$, there exists a $\mathbb{Z}$-graded $Q$-variety $(\mathcal{A}, Q)$ over $\mathcal{O}$ whose basic Lie--Rinehart algebra is precisely $\mathfrak{A}$. The answer to this question is provided in Theorem~\ref{main:smooth}. In that theorem, $(\mathcal{A}, Q)$ is constructed so that its negative part is an arborescent Koszul--Tate resolution of $\mathcal{O}/\mathcal{I}$, while its positive part is a universal positively graded $Q$-variety associated to  $\mathfrak{A}$, which we recall below.

\begin{theorem}[Existence of an $NQ$-variety \cite{CLRL}]\label{thm;Universal}Let $(\mathfrak A, [\cdot, \cdot]_\mathfrak A , \rho_{\mathfrak A})$ be a  Lie-Rinehart algebra over $\mathcal{O}$. Any free/projective resolution
    \begin{equation}\label{eq:free-resol}
\cdots \stackrel{\dd} \longrightarrow\mathcal P_{-3} \stackrel{\dd}{\longrightarrow} \mathcal P_{-2} \stackrel{\dd}{\longrightarrow} \mathcal P_{-1} \stackrel{\pi}{\longrightarrow} \mathfrak A \end{equation} 
  of $\mathfrak A$  over $\mathcal O$   lifts to a unique negatively graded Lie $\infty$-algebroid structure  whose $1$-ary bracket is $\dd$ and whose anchor map is  $\rho=\rho_\mathfrak A\circ \pi$. Moreover, any two such constructions are homotopy equivalent in the sense of \cite{CLRL}. This class is called the \emph{universal Lie $\infty$-algebroid  of $\mathfrak{A}$} and is denoted $\mathbb U_\mathfrak A$.
\end{theorem}

As a consequence, we obtain the following.

\begin{proposition}\label{thm;Universal2}
Let $(\mathfrak A, \lb_\mathfrak A, \rho_\mathfrak A)$ be a Lie–Rinehart algebra over $\mathcal{O}$. Assume that the resolution \eqref{eq:free-resol} of Theorem~\ref{thm;Universal} is such that $\mathcal P_{-i}$ is finitely generated for all $i \geq 1$. Then:
\begin{enumerate}
\item The universal Lie $\infty$-algebroid $\mathbb U_\mathfrak A$ of $\mathfrak A$ dualizes to a positively graded $Q$-variety over $\mathcal{O}$.
\item If $\mathcal{O} = C^\infty(M)$ and $\mathfrak A \subseteq \mathfrak X(M)$ is a singular foliation on a manifold $M$, then $\mathbb U_\mathfrak A$ coincides with the universal Lie $\infty$-algebroid of Laurent-Gengoux, Lavau, and Strobl.
\end{enumerate}
\end{proposition}

\begin{proof}
Item~1 follows from the duality between Lie $\infty$-algebroids and $Q$-manifolds \cite{Voronov,Voronov2}. Item~2 follows because $\mathfrak A$ admits a geometric resolution in the sense of \cite{LLS}.
\end{proof}

\begin{remark}
    If $\mathcal{O}$ is the coordinates ring of some affine variety $W\subseteq \mathbb C^d$ and $\mathfrak{A}$ is a Lie-Rinehart subalgebra of $\mathfrak X(W)$, \begin{enumerate}
        \item then $\mathfrak{A}$  and the modules $\mathcal{ P}_i$ in \eqref{eq:free-resol} are finitely generated since $\mathcal{ O}_W$ is Noetherian.  By the Hilbert Syzygy Theorem, those can be chosen to be of finite length when $W=\mathbb C^d$.
        \item the universal Lie $\infty$-algebroid  of $\mathfrak{A}$ corresponds to an $\mathbb N$-graded $Q$-variety over $\mathcal{O}$. It shall be called the  \emph{universal $NQ$-variety of $\mathfrak{A}$}.
    \end{enumerate}
    Notice that Theorem \ref{thm;Universal} is valid for $\mathcal{O}=C^\infty(M)/\mathcal{I}$, where $\mathcal{I}\subset C^\infty(M)$ is an ideal. This allows to associate a Lie $\infty$-algebroid to any subset $\Sigma\subseteq M$ by taking $\mathcal{I}$ to be the ideal of functions vanishing on $\Sigma$ and $\mathfrak{A}=\mathrm{Der}(C^\infty(M)/\mathcal{I})$ the Lie-Rinehart algebra of derivations of $C^\infty(M)/\mathcal{I}$.
\end{remark}
\subsection{Arborescent Koszul-Tate resolutions}
    An important class of examples of negatively graded $Q$-varieties/manifolds is the so-called \emph{Koszul-Tate resolution} of $\cO/\cI$ for some proper ideal $\cI$ of an algebra $\cO$. Koszul-Tate resolutions are also known in a more algebraic setting, see \cite{Tate} for a historical introduction. 
    \begin{definition}
    \label{def:KT.geom}
      Let $\cI$ be an ideal  of an algebra $\cO$. The \emph{Koszul-Tate resolution}  of $\cO/\cI$   is a negatively graded $Q$-variety $(\cA^- ,\delta)$ over $\cO$ such that 
        \begin{itemize}
        \item  $ \cA^- \simeq S(\oplus_{i\leq -1}\cV_i)$, for some collection of projective $\cO$-modules $\cV_i$.
\item  the homology, $H^{-i}(\cA^-, \delta) = 0$ for $i\geq 1$ and $H^0(\cA^-,\delta) = \mathcal{O}/\mathcal{I}$.
        \end{itemize}
    \end{definition}

A standard way to construct a Koszul-Tate resolution $(S(\mathcal V), \delta)$ of $\cO/\cI$ is to employ the Tate algorithm \cite{Tate}. The main idea is to consistently extend a differential graded commutative algebra (dgca) $S(\mathcal V_{-1}\oplus \dots \oplus \mathcal V_{-k})$ equipped with a differential of degree $+1$ \[\delta^{-k}\colon S(\mathcal V_{-1}\oplus \dots \oplus \mathcal V_{-k})\to S(\mathcal V_{-1}\oplus \dots \oplus \mathcal V_{-k})\] to a dgca $\left(S(\mathcal V_{-1}\oplus \dots \oplus \mathcal V_{-k} \oplus \V_{-k-1}), \delta^{-k-1}\right)$ such that
\begin{itemize}
    \item[$\bullet$] $\delta^{-k-1}$ coincides with $\delta^{-k}$ on $\V_{-1}\oplus \dots \oplus \mathcal V_{-k}$.
    \item[$\bullet$]  $H^{-k}(S(\mathcal V_{-1}\oplus \dots \oplus \mathcal V_{-k}\oplus \cV_{-k-1}),\; \delta^{-k-1}) = 0$. In other words, all non-trivial cycles of degree $-k$ of $(S(\mathcal V_{-1}\oplus \dots \oplus \mathcal V_{-k},\; \delta^{-k})$ are in the image $\delta^{-k-1}(\mathcal V_{-k-1})$.
 \end{itemize}
 Despite its apparent simplicity, this algorithm, in general, does not terminate even for Noetherian $\cO$. For such $\cO$ all $\mathcal V_{-j}$ can be chosen to be finitely generated, but one still deals with an infinite collection of them, at least for $\cO$ being a polynomial ring and a monomial ideal $\cI \subsetneq \cO$ \cite{hancharuk2024}. An alternative approach is given by arborescent Koszul-Tate resolutions \cite{hancharuk2024} which are obtained from a projective resolution $(\FM, d)$ of $\cO/\cI$, which is simpler in a number of examples. By projective resolution $(\FM, d)$ we mean a sequence of projective $\cO$-modules $(\FM_i)_{i<0}$ such that
 $$
\begin{tikzcd}
  \cdots \arrow[r,"d"]  & \mathfrak M_{-i}\arrow[r,"d"] & \cdots \arrow[r, "d"] & \FM_{-1} \arrow[r, "d"] & \cO \arrow[r] & 0
\end{tikzcd}
 $$
 is acyclic in negative degrees and $H^{0}(\FM, d) = \cO/\cI$.
 We shall fix a projective resolution $(\FM, d)$ of $\mathcal
  O/\mathcal{I}$ in the sequel. Let us briefly recall the construction.
 \subsubsection{The construction}\label{KT:construction}
In this section we use the notations and conventions of \cite{hancharuk2024}.
 \begin{enumerate}
     \item The starting point is the set of \emph{planar rooted trees}, i.e., trees that can be embedded in a plane with a distinguished vertex, labeled as the root, at the bottom of the tree. There is a natural partial ordering on the vertices (nodes) of such trees: $A<B$ if the unique path from the root to the node $B$ passes through $A$. Such $A$ is called an \emph{ascendant} of $B$ (and $B$ is a \emph{descendant} of $A$). If $A<B$ and there are no vertices on the path from $A$ to $B$, i.e., $A$ and $B$ are connected by an edge, then $A$ is called the \emph{parent} of $B$, while $B$ is called a \emph{child} of $A$. The maximal elements w.r.t the partial ordering "$<$" are referred to as \emph{leaves}, the minimal one is the \emph{root}, while to others we refer to as \emph{inner vertices}.
      \begin{center}
\begin{tabular}{cc} 
    \scalebox{0.5}{  \begin{forest}
for tree = {grow' = 90}, nice empty nodes, for tree={ inner sep=0 pt, s sep= 0 pt, fit=band, 
},
[,
[,{label=[mystyle]{\scalebox{2}{$A\quad $}}}
    [,
        [, tier =1] 
        [,
        [, {label=[mystyle2]{\scalebox{2}{$ B$}}}, tier =1]
        [, tier = 1]
        ]
    ]
    [, tier =1]
]
]
\path[fill=black] (.parent anchor) circle[radius=4pt]
(!1.child anchor) circle[radius=4pt]
(!11.child anchor) circle[radius=4pt]
(!112.child anchor) circle[radius=4pt];
\end{forest}
}
&
   \scalebox{1}{   \begin{forest}
for tree = {grow' = 90}, nice empty nodes, for tree={ inner sep=0 pt, s sep= 15 pt, fit=band, 
},
[,{label=[mystyle]{\scalebox{1}{$E\; $}}}
    [,{label=[mystyle]{\scalebox{1}{$\quad C$}}}, tier =1]
    [,{label=[mystyle2]{\scalebox{1}{$\;D$}}}
        [,{label=[mystyle]{\scalebox{1}{$\quad G$}}}, tier =1] 
        [, tier = 1]
        [, tier =1]
    ]
]
\path[fill=black] (.parent anchor) circle[radius=2pt]
(!2.child anchor) circle[radius=2pt];
\end{forest}
}
\end{tabular}
\captionof{figure}{Examples of rooted planar trees}
\label{fig:trees1}
\end{center}
     The terminology above is illustrated in the Figure \ref{fig:trees1}. $A$ is an ascendant of $B$, $B$ is a leaf vertex. $E$ is the root vertex of the second tree, and it is a parent of vertices $C$ and $D$. $G$ is a leaf and a child of $D$. We label the root vertex as well as inner vertices by $\bullet$. Leaves are not marked by $\bullet$.
     
     \item The next step is to consider a $\mathbb K$-vector space $Tree$ of planar rooted trees, satisfying an additional condition that the valency (i.e., the number of edges connected to the vertex) of each inner vertex $\geq 3$ and the root valency is $\geq 2$. This space is enlarged by a trivial tree, which, by our conventions, consists of only one leaf.
 \end{enumerate} Examples of such trees are depicted below:
    \begin{center}
\begin{tabular}{ccc} 
    \scalebox{0.5}{  \begin{forest}
for tree = {grow' = 90}, nice empty nodes, for tree={ inner sep=0 pt, s sep= 0 pt, fit=band, 
},
[,{label=[mystyle]{\scalebox{2}{$B\quad $}}}
    [,{label=[mystyle]{\scalebox{2}{$A\quad $}}}
        [, tier =1] 
        [, tier =1]
    ]
    [, tier =1]
]
\path[fill=black] (.parent anchor) circle[radius=4pt]
(!1.child anchor) circle[radius=4pt];
\end{forest}
}
&
   \scalebox{0.5}{   \begin{forest}
for tree = {grow' = 90}, nice empty nodes, for tree={ inner sep=0 pt, s sep= 15 pt, fit=band, 
},
[
    [,{label=[mystyle]{\scalebox{2}{$C $}}}, tier =1]
      [, tier =1]
    [,{label=[mystyle2]{\scalebox{2}{$\quad D$}}}
        [, tier =1] 
        [
        [, tier = 1]
        [, tier = 1]
        ]
        [, tier =1]
    ]
]
\path[fill=black] (.parent anchor) circle[radius=4pt]
(!3.child anchor) circle[radius=4pt]
(!32.child anchor)circle[radius=4pt];
\end{forest}
}
&          \begin{forest}
for tree = {grow' = 90}, nice empty nodes, for tree={ inner sep=0 pt, s sep= 0 pt, fit=band, 
},
[[, tier =1]]
;
\end{forest} \\
A $3$-leaves tree & A $6$-leaves tree & The trivial tree

\end{tabular}
\captionof{figure}{Admissible planar trees}
\label{fig:trees2}
\end{center}
It is our convention to visualize the trivial tree as “leaf” + “edge”. The first tree in Figure \ref{fig:trees1} is not admissible, since the root valency is $1$.
\begin{convention}
\label{conv:uptree}
\normalfont
The vocabulary of vertices can be extended to subtrees in a direct manner: first, if $A$ is a vertex of a tree $t$, a subtree $t_{\uparrow A}$ is defined as a tree obtained by removing from $t$ all vertices (and corresponding edges) which are not $A$ or descendants of $A$. If $A$ is a child of $B$, we define a child subtree of $B$ to be $t_{\uparrow A}$. A tree $t_{\downarrow A}$ is defined by replacing $t_{\uparrow A}$ with a leaf, i.e., by declaring a vertex $A$ to be a maximal element. In particular,
$$
\hbox{ if $t =$} \adjustbox{valign=c}{  \scalebox{0.5}{  \begin{forest}
for tree = {grow' = 90}, nice empty nodes, for tree={ inner sep=0 pt, s sep= 0 pt, fit=band, 
},
[,{label=[mystyle]{\scalebox{2}{$B\quad $}}}
    [,{label=[mystyle]{\scalebox{2}{$A\quad $}}}
        [, tier =1] 
        [, tier =1]
         [, tier =1]
    ]
    [, tier =1]
]
\path[fill=black] (.parent anchor) circle[radius=4pt]
(!1.child anchor) circle[radius=4pt];
\end{forest}
}} \hbox{, then $t_{\uparrow A} =$ }
\adjustbox{valign=c}{    \begin{forest}
for tree = {grow' = 90}, nice empty nodes, for tree={ inner sep=0 pt, s sep= 0 pt, fit=band, 
},
[,{label=[mystyle]{$A\enspace $}}
        [, tier =1] 
        [, tier =1]
         [, tier =1]
]
\path[fill=black] (.parent anchor) circle[radius=2pt];
\end{forest}
}
\hbox{ and $t_{\downarrow A} =$ }\adjustbox{valign=c}{\scalebox{1}{  \begin{forest}
for tree = {grow' = 90}, nice empty nodes, for tree={ inner sep=0 pt, s sep= 0 pt, fit=band, 
},
[,{label=[mystyle]{\scalebox{1}{$B\quad $}}},
    [, tier =1]
    [,{label=[mystyle]{\scalebox{1}{$A $}}}, tier =1]
]
\path[fill=black] (.parent anchor) circle[radius=2pt];
\end{forest}
}}.
$$
Also, if $A$ is the root vertex of a tree $t$, then $t_{\downarrow A}$ is a trivial tree, and $t_{\uparrow A} = t$. If $A$ is a leaf of $t$, then $t_{\downarrow A} = t$ and $t_{\uparrow A}  = A$.  
\end{convention}

\noindent
\textbf{Module of planar
rooted decorated trees}: Denote $Tree^n$ the $\mathbb K$-subspace of $Tree$ of trees with $n$ leaves. Let us describe the $\cO$-module of planar rooted trees decorated with $(\FM, d)$.  We set
\begin{equation}
    Tree[\FM]\coloneqq  \oplus_{n=1}^{\infty} Tree^n \otimes_{\mathbb K} \FM^{\otimes n}.
\end{equation}
 Here, $\FM^{\otimes n}$ stands for the tensor product over $\mathcal O$ of $\FM$ taken with itself $n$-times. The $\cO$-module structure of $Tree[\FM]$ is obvious. Since the number of leaves is equal to the tensor power of $\FM$, it is useful to interpret elements of the $\mathcal{O}$-module $Tree[\FM]$ as planar trees with leaves decorated (labeled) by elements of $\FM$. We denote by $t[a_1, \dots, a_n]$ the element $t\otimes_{\mathbb K} a_1 \otimes \dots \otimes a_n \in Tree[\FM]$, where $t$ is a rooted tree $\in Tree^n$ decorated with $a_1 \otimes \dots \otimes a_n$. The \emph{homological degree} of a  tree $t\in Tree^n$ decorated by elements $a_1, \dots, a_n$ of homogeneous degree is given by 
 \begin{equation}
    |t[a_1, \dots, a_n]| = \hbox{  $-$ root $\#$ $-\#$ (of inner vertices of $t$) } + |a_1| + \dots +|a_n|.
 \end{equation}
 Here the root $\#$ is equal to $0$ or $1$, depending on whether the tree is trivial (root $\# = 0$) or not (root $\# = 1$). In particular, for trees  
 $$
 t_1[a_1, \dots, a_5] = \adjustbox{valign =c }{
  \scalebox{0.5}{   \begin{forest}
for tree = {grow' = 90}, nice empty nodes, for tree={ inner sep=5 pt, s sep= 5 pt, fit=band, 
},
[
    [\scalebox{2}{$a_1$}, tier =1]
      [\scalebox{2}{$a_2$}, tier =1]
    [
        [\scalebox{2}{$a_3$}, tier =1] 
        [\scalebox{2}{$a_4$}, tier = 1]
        [\scalebox{2}{$a_5$}, tier = 1]
    ]
]
\path[fill=black] (.parent anchor) circle[radius=4pt]
(!3.child anchor) circle[radius=4pt];
\end{forest}
}}, \quad \hbox{ and } \quad
t_2[a_1] = \adjustbox{valign =c }{
  \begin{forest}
for tree = {grow' = 90}, nice empty nodes, for tree={ inner sep=5 pt, s sep= 5 pt, fit=band, 
},
[
    [$a_1$, tier =1]
];
\end{forest}}
 $$
 $|t_1[a_1,\dots, a_5]| = -2 + |a_1| + \dots + |a_5|$ and $|t_2[a_1]| = |a_1|$. For $n\in \mathbb N$, we denote  by $Tree[\FM]_{-n}$ the subspace of decorated trees of degree $-n$. Clearly, $Tree[\FM] = \oplus_{n\in \mathbb N}Tree[\FM]_{-n}$. 

 \noindent
 \textbf{Some operations on rooted decorated  trees}: The following two natural operations on rooted trees are employed in the construction of the arborescent Koszul-Tate resolution. The first one is the isomorphism $$T^{\geq 2}(Tree[\FM]) \cong Tree^{\geq 2}[\FM]$$ provided by the root map $\r \colon T^{\geq 2}(Tree[\FM]) \cong Tree^{\geq 2}[\FM]$. Graphically, $\r$ tantamounts to  joining the forest of trees into a tree by means of a root vertex, e.g., 
 $$
\r\colon \adjustbox{valign =c }{
  \scalebox{1}{   \begin{forest}
for tree = {grow' = 90}, nice empty nodes, for tree={ inner sep=5 pt, s sep= 5 pt, fit=band, 
},
[
    [\scalebox{1}{$a_1$}, tier =1]
      [\scalebox{1}{$a_2$}, tier =1]
]
\path[fill=black] (.parent anchor) circle[radius=2pt];
\end{forest}
}} \otimes 
\adjustbox{valign =c }{
  \scalebox{1}{   \begin{forest}
for tree = {grow' = 90}, nice empty nodes, for tree={ inner sep=5 pt, s sep= 5 pt, fit=band, 
},
    [
        [\scalebox{1}{$a_3$}, tier =1] 
        [\scalebox{1}{$a_4$}, tier = 1]
        [\scalebox{1}{$a_5$}, tier = 1]
    ]
\path[fill=black] (.parent anchor) circle[radius=2pt];
\end{forest}
}} \mapsto 
\adjustbox{valign =c }{
  \scalebox{0.5}{   \begin{forest}
for tree = {grow' = 90}, nice empty nodes, for tree={ inner sep=5 pt, s sep= 5 pt, fit=band, 
},
[
[
    [\scalebox{2}{$a_1$}, tier =1]
      [\scalebox{2}{$a_2$}, tier =1]
      ]
    [
        [\scalebox{2}{$a_3$}, tier =1] 
        [\scalebox{2}{$a_4$}, tier = 1]
        [\scalebox{2}{$a_5$}, tier = 1]
    ]
]
\path[fill=black] (.parent anchor) circle[radius=4pt]
(!1.child anchor) circle[radius=4pt]
(!2.child anchor) circle[radius=4pt];
\end{forest}
}}
 $$
The other operation, $\partial_A\colon T^{\geq 2}(Tree[\FM]) \cong Tree^{\geq 2}[\FM]$, consists in removing an inner vertex $A$ of a given decorated tree. This is done by means of merging $A$ with its parent vertex $P_A$. For instance,
$$
\partial_A\colon \adjustbox{valign =c }{
  \scalebox{0.5}{   \begin{forest}
for tree = {grow' = 90}, nice empty nodes, for tree={ inner sep=5 pt, s sep= 5 pt, fit=band, 
},
[,{label=[mystyle]{\scalebox{2}{$P_A\quad $}}}
[,{label=[mystyle]{\scalebox{2}{$A\quad $}}}
    [\scalebox{2}{$a_1$}, tier =2]
      [\scalebox{2}{$a_2$}, tier =2]
      ]
    [
        [\scalebox{2}{$a_3$}, tier =1] 
        [\scalebox{2}{$a_4$}, tier = 1]
        [\scalebox{2}{$a_5$}, tier = 1]
    ]
]
\path[fill=black] (.parent anchor) circle[radius=4pt]
(!1.child anchor) circle[radius=4pt]
(!2.child anchor) circle[radius=4pt];
\end{forest}
}} \mapsto
\adjustbox{valign =c }{
  \scalebox{0.5}{   \begin{forest}
for tree = {grow' = 90}, nice empty nodes, for tree={ inner sep=5 pt, s sep= 5 pt, fit=band, 
},
[,{label=[mystyle]{\scalebox{2}{$P_A\quad $}}}
    [\scalebox{2}{$a_1$}, tier =2]
      [\scalebox{2}{$a_2$}, tier =2]
    [
        [\scalebox{2}{$a_3$}, tier =1] 
        [\scalebox{2}{$a_4$}, tier = 1]
        [\scalebox{2}{$a_5$}, tier = 1]
    ]
]
\path[fill=black] (.parent anchor) circle[radius=4pt]
(!3.child anchor) circle[radius=4pt];
\end{forest}
}}
$$

\noindent
\textbf{Module of symmetric decorated trees}: Let $A$ be any vertex of a tree $t[a_1,\dots, a_n]$ and let $\theta_1, \dots, \theta_m$ be the decorated subtrees of $t$ which are children of $A$. We set the following equivalence relations:
\begin{equation}
\label{eq:sym.trees}
t[a_1, \dots, a_n] \sim \epsilon(\sigma, \theta) t[a_{\sigma(1)}, \dots, a_{\sigma(n)}],
\end{equation}
where $\sigma$ is the permutation of the children of $A$ and $\epsilon(\sigma, \theta)$ is the Koszul sign of this permutation with respect to the degrees $|\theta_1|, \dots, |\theta_m|$. More precisely, $\epsilon(\sigma,\theta)$ is deduced from the following equation:
$$
\theta_1 \cdots \theta_m  = \epsilon(\sigma, \theta)\cdot\theta_{\sigma(1)}\cdots \theta_{\sigma(m)}.
$$

Let us illustrate the equivalence on the following example:
$$
\adjustbox{valign =c }{
  \scalebox{0.5}{   \begin{forest}
for tree = {grow' = 90}, nice empty nodes, for tree={ inner sep=5 pt, s sep= 5 pt, fit=band, 
},
[,{label=[mystyle]{\scalebox{2}{$P\quad $}}}
[,{label=[mystyle]{\scalebox{2}{$A\quad $}}}
    [\scalebox{2}{$a_1$}, tier =2]
      [\scalebox{2}{$a_2$}, tier =2]
      ]
    [,{label=[mystyle]{\scalebox{2}{$B\quad $}}}
        [\scalebox{2}{$a_3$}, tier =1] 
        [\scalebox{2}{$a_4$}, tier = 1]
    ]
]
\path[fill=black] (.parent anchor) circle[radius=4pt]
(!1.child anchor) circle[radius=4pt]
(!2.child anchor) circle[radius=4pt];
\end{forest}
}} = (-1)^{(|a_1|+|a_2|+1)(|a_3| +|a_4| +1
)}
\adjustbox{valign =c }{
  \scalebox{0.5}{   \begin{forest}
for tree = {grow' = 90}, nice empty nodes, for tree={ inner sep=5 pt, s sep= 5 pt, fit=band, 
},
[,{label=[mystyle]{\scalebox{2}{$P\quad $}}}
    [,{label=[mystyle]{\scalebox{2}{$B\quad $}}}
        [\scalebox{2}{$a_3$}, tier =1] 
        [\scalebox{2}{$a_4$}, tier = 1]
    ]
    [,{label=[mystyle]{\scalebox{2}{$A\quad $}}}
    [\scalebox{2}{$a_1$}, tier =2]
      [\scalebox{2}{$a_2$}, tier =2]
      ]
]
\path[fill=black] (.parent anchor) circle[radius=4pt]
(!1.child anchor) circle[radius=4pt]
(!2.child anchor) circle[radius=4pt];
\end{forest}
}}
=(-1)^{|a_1||a_2|}
\adjustbox{valign =c }{
  \scalebox{0.5}{   \begin{forest}
for tree = {grow' = 90}, nice empty nodes, for tree={ inner sep=5 pt, s sep= 5 pt, fit=band, 
},
[,{label=[mystyle]{\scalebox{2}{$P\quad $}}}
[,{label=[mystyle]{\scalebox{2}{$A\quad $}}}
    [\scalebox{2}{$a_2$}, tier =2]
      [\scalebox{2}{$a_1$}, tier =2]
      ]
    [,{label=[mystyle]{\scalebox{2}{$B\quad $}}}
        [\scalebox{2}{$a_3$}, tier =1] 
        [\scalebox{2}{$a_4$}, tier = 1]
    ]
]
\path[fill=black] (.parent anchor) circle[radius=4pt]
(!1.child anchor) circle[radius=4pt]
(!2.child anchor) circle[radius=4pt];
\end{forest}
}}.
$$

The first equality is due to the permutation of subtrees with roots $A$ and $B$ (which are children of $P$), while the second one comes from a permutation of leaves $|\otimes_{\mathbb K}a_1$ and $|\otimes_{\mathbb K}a_2$.

\begin{definition}\label{def:symmetric-decorated-trees}
The $\cO$-module of \emph{symmetric decorated trees}  is defined by   $\mathcal{T}ree[\FM] = Tree[\FM]/\!\sim$, where $\sim$ is given in \eqref{eq:sym.trees}.
\end{definition}

\noindent
\textbf{Arborescent Koszul-Tate resolutions}: In what follows, we abbreviate elements of $\cTr[\FM]$ by a representative in $Tree[\FM]$, while bearing in mind that the branches of such trees can be permuted with the Koszul sign. Let $S(\cTr[\FM])$ be the graded symmetric algebra of symmetric decorated trees $\mathcal{T}ree[\FM]$ in Definition \ref{def:symmetric-decorated-trees}. Let us list the first graded components $\cTr[\FM]_i =: \cV_i$,  $-4\leq i\leq -1$:
\begin{align*}
    \cV_{-1} \coloneqq& \adjustbox{valign =c }{
  \scalebox{0.5}{   \begin{forest}
for tree = {grow' = 90}, nice empty nodes, for tree={ inner sep=5 pt, s sep= 5 pt, fit=band, 
},
[,
[
]
];
\end{forest}
}} \otimes_{\mathbb K} \FM_{-1} \\
 \cV_{-2} \coloneqq& \adjustbox{valign =c }{
  \scalebox{0.5}{   \begin{forest}
for tree = {grow' = 90}, nice empty nodes, for tree={ inner sep=5 pt, s sep= 5 pt, fit=band, 
},
[,
[
]
];
\end{forest}
}} \otimes_{\mathbb K} \FM_{-2} \\
\cV_{-3}\coloneqq& \adjustbox{valign =c }{
  \scalebox{0.5}{   \begin{forest}
for tree = {grow' = 90}, nice empty nodes, for tree={ inner sep=5 pt, s sep= 5 pt, fit=band, 
},
[,
[
]
];
\end{forest}
}} \otimes_{\mathbb K} \FM_{-3} \oplus \adjustbox{valign =c }{
  \scalebox{0.5}{   \begin{forest}
for tree = {grow' = 90}, nice empty nodes, for tree={ inner sep=5 pt, s sep= 5 pt, fit=band, 
},
[
    [, tier =1]
      [, tier =1]
]
\path[fill=black] (.parent anchor) circle[radius=4pt];
\end{forest}
}}\otimes_{\mathbb K} \FM_{-1}\odot \FM_{-1} \\
\cV_{-4}\coloneqq& \adjustbox{valign =c }{
  \scalebox{0.5}{   \begin{forest}
for tree = {grow' = 90}, nice empty nodes, for tree={ inner sep=5 pt, s sep= 5 pt, fit=band, 
},
[,
[
]
];
\end{forest}
}} \otimes_{\mathbb K} \FM_{-4} \oplus \adjustbox{valign =c }{
  \scalebox{0.5}{   \begin{forest}
for tree = {grow' = 90}, nice empty nodes, for tree={ inner sep=5 pt, s sep= 5 pt, fit=band, 
},
[
    [, tier =1]
      [, tier =1]
]
\path[fill=black] (.parent anchor) circle[radius=4pt];
\end{forest}
}}\otimes_{\mathbb K} \FM_{-1}\odot \FM_{-2} \oplus \adjustbox{valign =c }{
  \scalebox{0.5}{   \begin{forest}
for tree = {grow' = 90}, nice empty nodes, for tree={ inner sep=5 pt, s sep= 5 pt, fit=band, 
},
[
    [, tier =1]
      [, tier =1]
      [, tier =1]
]
\path[fill=black] (.parent anchor) circle[radius=4pt];
\end{forest}
}}\otimes_{\mathbb K} \FM_{-1}\odot \FM_{-1}\odot \FM_{-1}.
\end{align*}

The root map $\r$ is consistently defined on this quotient, providing an isomorphism $\r\colon S^{\geq 2}(\cTr[\FM]) \simeq \cTr^{\geq 2}[\FM]$. 
 
\begin{convention}
\label{conv:triv}
\normalfont
    We shall identify the submodule of trivial decorated trees $|\otimes_{\mathbb K}\FM$ with $\FM$. In particular, $a\in \FM$ is understood as a trivial tree decorated with $a$. We use the notation $\p, \pv$ and  $\ptr$ for the natural projections $\p\colon\, S(\cTr[\FM]) \longrightarrow S^{\geq 2}(\cTr[\FM])$, $\pv\colon\, S(\cTr[\FM]) \longrightarrow \cTr^{\geq 2}[\FM]$ and $\ptr \colon\, S(\cTr[\FM]) \longrightarrow \FM$.
\end{convention}
\begin{definition}
\label{def:arb.KT}Let $\psi\colon \cTr[\FM] \rightarrow \FM$ be an $\cO$-linear map  of degree $+1$. The \emph{arborescent pre-differential associated to $\psi$} is a derivation $\delta_\psi: S(\cTr(\FM)) \longrightarrow S(\cTr(\FM))$  of degree $+1$ that is defined by the following recursive formula:
    \begin{itemize}
        \item \emph{Initial conditions:}  if $t[a] \in \FM$, $ \delta_\psi (t[a]) \coloneqq d(a)$.
        \item \emph{Recursive relations:} For $t[a_1, \dots, a_m] \in \cTr[\FM]_{-n}$, $n\geq 3, m \geq 2$:
        \begin{equation}
            \label{eq:recursive.arb.der}
         \delta_{\psi}(t[a_1,\dots, a_m]) \coloneqq \r^{-1} (t[a_1,\dots, a_m]) - \r \circ \mathrm{p}^{\geq 2} \circ \delta_{\psi} \circ \r^{-1} (t[a_1,\dots, a_m]) - \psi(t [a_1, \dots, a_m]).
        \end{equation}
        
    \end{itemize}
\end{definition}
For $n=3, m =2$ the derivation $\delta_{\psi}$ is decoded as follows:
$$
\delta_{\psi} \left( \adjustbox{valign =c }{
  \scalebox{0.5}{   \begin{forest}
for tree = {grow' = 90}, nice empty nodes, for tree={ inner sep=5 pt, s sep= 5 pt, fit=band, 
},
[
    [\scalebox{2}{$a_1$}, tier =1]
      [\scalebox{2}{$a_2$}, tier =1]
]
\path[fill=black] (.parent anchor) circle[radius=4pt];
\end{forest}
}} \right) = a_1 \odot a_2 - \psi \left ( \adjustbox{valign =c }{
  \scalebox{0.5}{   \begin{forest}
for tree = {grow' = 90}, nice empty nodes, for tree={ inner sep=5 pt, s sep= 5 pt, fit=band, 
},
[
    [\scalebox{2}{$a_1$}, tier =1]
      [\scalebox{2}{$a_2$}, tier =1]
]
\path[fill=black] (.parent anchor) circle[radius=4pt];
\end{forest}
}} \right) 
$$
For $n=3$ there are no terms from the second summand in \eqref{eq:recursive.arb.der}. They do appear, for example, when the decoration $a_1$ has a degree lower than $-1$. In particular,
$$
\delta_{\psi} \left( \adjustbox{valign =c }{
  \scalebox{0.5}{   \begin{forest}
for tree = {grow' = 90}, nice empty nodes, for tree={ inner sep=5 pt, s sep= 5 pt, fit=band, 
},
[
    [\scalebox{2}{$a_1$}, tier =1]
      [\scalebox{2}{$a_2$}, tier =1]
]
\path[fill=black] (.parent anchor) circle[radius=4pt];
\end{forest}
}} \right) = a_1 \odot a_2 -\adjustbox{valign =c }{
  \scalebox{0.5}{   \begin{forest}
for tree = {grow' = 90}, nice empty nodes, for tree={ inner sep=5 pt, s sep= 5 pt, fit=band, 
},
[
    [\scalebox{2}{$da_1$}, tier =1]
      [\scalebox{2}{$a_2$}, tier =1]
]
\path[fill=black] (.parent anchor) circle[radius=4pt];
\end{forest}
}}- \psi \left ( \adjustbox{valign =c }{
  \scalebox{0.5}{   \begin{forest}
for tree = {grow' = 90}, nice empty nodes, for tree={ inner sep=5 pt, s sep= 5 pt, fit=band, 
},
[
    [\scalebox{2}{$a_1$}, tier =1]
      [\scalebox{2}{$a_2$}, tier =1]
]
\path[fill=black] (.parent anchor) circle[radius=4pt];
\end{forest}
}} \right), 
$$
if $|a_1| < -1, |a_2|=-1.$
The purpose of this derivation $\delta_{\psi}$ is to turn $S(\cTr[\FM])$ into a Koszul-Tate resolution, as stated in the following theorem:
\begin{theorem}[\cite{hancharuk:tel-04692988,hancharuk2024}]
\label{thm:arborescentKT}
    Let $\cO$ be a commutative $\mathbb K$-algebra, and  $\cI\subsetneq \cO$ an ideal of $\mathcal O$. There exists a map $\psi\colon \cTr[\FM] \rightarrow \FM$ of degree $+1$ such that $(S(\mathcal{T}ree[\FM]), \delta_{\psi})$ is a Koszul-Tate resolution of $\cO/\cI$. In addition,$(S(\cTr[\FM],\delta_{\psi})$ is a homotopy retract on $(\FM ,d)$
    $$  
\begin{tikzcd}[column sep = 4em]
(S(\cTr[\FM]), \delta_{\psi})  \arrow[r, bend left=20, "{{p}}", shift ={(0 ,1mm)}]  \arrow[out=225, in=135, looseness=8, loop, distance =20mm, " \r\circ \p ", shift = {(-5mm, 0)}] & (\FM,d) \arrow[l,  bend left=20, swap, "{\iota}", shift = {(0, -1mm)}]
\end{tikzcd}.
$$ 
with $\iota\colon \FM \to \FM$ being the identity map, ${p} = (\ptr + \psi\circ \r\circ \p)$, so that $p\circ \iota= \mathrm{Id}$ and
    \begin{equation}
    \label{eq:hom.retract0}
        \delta_{\psi}\circ (\r\circ \p) +  (\r\circ \p) \circ \delta_\psi = \mathrm{Id} - \iota\circ p
    \end{equation}
    The pair $(S(\mathcal{T}ree[\FM]), \delta_\psi)$ is called an arborescent Koszul-Tate resolution of $\mathcal{O}/\mathcal I$. In that case, $\psi$ is called the \emph{hook map}  of $(S(\mathcal{T}ree[\FM]), \delta_\psi)$.
\end{theorem}
\begin{remark}
\label{rem:psi.mult}
   \normalfont
    The defining condition on the existence of the map $\psi\colon \cTr[\FM] \rightarrow \FM$ in Theorem \ref{thm:arborescentKT} comes from the identity $\delta_{\psi}^2  = 0$. More precisely, the condition reads as
    \begin{equation}
        \label{eq:psi.recursion}
         d\psi = \ptr \circ \delta_\psi \circ \r^{-1} + \psi\circ \r \circ \p \circ \delta_\psi \circ \r^{-1}.\end{equation}
    This equation can be solved at each homological degree. The interpretation of the map $\psi$ is the following: on trees with two leaves $\psi$ induces a multiplication $\star$ on $\FM$, compatible with $d$. More precisely, for $a,b\in \FM$, $$a\star b \coloneqq \psi \left( \adjustbox{valign = c} {\scalebox{0.5}{  \begin{forest}
for tree = {grow' = 90}, nice empty nodes,
            [,
      [\scalebox{2}{$a$}, tier =1]
      [\scalebox{2}{$b$}, tier =1] 
         ]
\path[fill=black] (.parent anchor) circle[radius=4pt];
\end{forest}}} \right)$$ and the Equation \eqref{eq:psi.recursion} is tantamount to $d(a\star b) = da\star b + (-1)^{|a|}a\star db$, where $\star $ coincides with the multiplication in $\cO$ when $da$ or $db$ are in $\cO$. In general, such a $\psi\colon \cTr[\FM] \rightarrow \FM$  of Theorem \ref{thm:arborescentKT} induces a $C_{\infty}$-structure on $\FM$, see \cite{hancharuk2024}.

\end{remark}

\noindent
\textbf{A more explicit description of the arborescent Koszul-Tate resolution}. For practical purposes, it is useful to write the expression of the arborescent Koszul-Tate differential $\delta_{\psi}\colon S(\cTr[\FM]) \longrightarrow S(\cTr[\FM])$ explicitly when choosing a representative of a symmetric tree. Let $t[a_1, \dots, a_n]$ be such a representative. When acting on $t[a_1, \dots, a_n]$, the derivation $ \delta_\psi$ modifies the subtrees of $t$. In order to write this action explicitly, one introduces a weight defined for any vertex of $t$ as well as some technical conventions:
\begin{definition}
    For any vertex $A$ of a tree $t[a_1, \dots, a_n]$ we associate a \emph{weight} $W_A\in \mathbb N$ defined as follows: 
\begin{itemize}
    \item $W_R = 0$ for $R$ being the root of $t$. 
    \item 
    Write $s$ for the unique path from the root $R$ of $t$ to the vertex $A$. Let us view it as an ordered set of vertices $(s_0, s_1, \dots, s_k)$, where $s_0 = A$ and $s_k$ is the root $R$. Since we work with a representative tree $t$, at any vertex $s_i$ we can distinguish between the edges on the left and on the right to the path. We address them as left edges and right edges. Let us denote by $\lbrace \theta^{\alpha}_{i},\, \alpha \in J_{i}\rbrace$ a collection of subtrees of $t$, whose root is connected to $s_i$ by a left edge. Then the weight $W_A = -k + \sum_{i=0}^k\sum_{\alpha\in J_{i}} |\theta^{\alpha}_{i}|$. In other words, it is a sum of all degrees of all trees $\theta^{\alpha}_i$ minus the length of the path. 
    
    \item For a trivial tree $a$ the weight of the leaf $a$ is set to be $0$.
\end{itemize}
\end{definition}
\begin{example}
For a tree
$$
 \scalebox{0.5}{   \begin{forest}
for tree = {grow' = 90}, nice empty nodes, for tree={ inner sep=5 pt, s sep= 5 pt, fit=band, 
},
[, {label=[mystyle]{\scalebox{2}{$R\quad $}}}
    [\scalebox{2}{$a_1$}, tier =1]
      [\scalebox{2}{$a_2$}, tier =1]
    [, , {label=[mystyle]{\scalebox{2}{$A\quad $}}}
        [\scalebox{2}{$a_3$}, tier =1] 
        [\scalebox{2}{$a_4$}, tier = 1]
        [\scalebox{2}{$a_5$}, tier = 1]
    ]
]
\path[fill=black] (.parent anchor) circle[radius=4pt]
(!3.child anchor) circle[radius=4pt];
\end{forest}
}
$$
$W_{a_1} = -1$, $W_R = 0$, $W_A = -1 + |a_1| + |a_2|$, $W_{a_4} = -2 + |a_1| + |a_2| + |a_3|$.
\end{example}

\begin{convention}
\normalfont
  \label{conv:tree}$ $
    We allow the trees to have decorations in $\cO$ by means of the following relations:
  \begin{align*}
  \hbox {if $n\geq 2$: }\quad &t[a_1, \dots, a_i+ F_i, \dots, a_n] = t[a_1, \dots, a_i, \dots, a_n]+F_i\cdot \mathrm{pr}_{T} (t[a_1, \dots, \hat F_i, \dots, a_n]) ,\\ &\quad \hbox{for any $i,j,\quad$ $a_j \in  \FM$, $F_i \in \cO$},\\
\hbox{if $n =1$: }\quad & \hbox{replace $\FM$ by $\FM\oplus\cO$} 
  \end{align*}
  Here $\hat F$ means omitting the leaf decorated with $F$. The resulting tree might not be admissible, which is remedied by projection $\mathrm{pr}_T$:  $\mathrm{pr}_{T} (t[b_1 ,\dots, b_m]) = t[b_1,  \dots, b_m]$ if $t[b_1,\dots, b_m] \in \cTr[\FM]$ and zero otherwise.
  
\end{convention}
This technical convention allows to simplify the explicit formulas for the differential $\delta_{\psi}$, as it is explained below:
\begin{proposition}
\label{prop:KT.explicit.form}
    In view of Conventions \ref{conv:uptree}, \ref{conv:triv}, \ref{conv:tree} the derivation  $ \delta_\psi\colon S(\cTr[\FM]) \longrightarrow S(\cTr[\FM])$ has the following explicit form while acting on $\cTr[\FM]$:
    
        \begin{align*}&\hbox{for $n=1$: }\quad  \delta_{\psi}( a) = d(a),\\ \\
    &\hbox{for $n\geq 2$: }\quad \delta_\psi(t[a_1, \dots, a_n]) = \r^{-1}t[a_1,\dots, a_n]+ \sum_{A\in \mathrm{InVert}(t)}(-1)^{W_A}\partial_A t[a_1,\dots, a_n] \\ &+\sum_{A\in \mathrm{Leaves}(t)}(-1)^{W_A}t[a_1,\dots,d(a_A) ,\dots, a_n] - \sum_{A\in \mathrm{InVert}(t)\cup\mathrm{Root}}(-1)^{W_A}t_{\downarrow A}[a_1,\dots,\psi({t_{\uparrow A}}(a_A)) ,\dots, a_n].
    \end{align*} 
\end{proposition}
Here, $\mathrm{InVert}(t)$, $\mathrm{Leaves}(t)$ are the sets of vertices, inner vertices, and leaves of $t$ respectively. 
Let us illustrate the explicit form on a particular example:
\begin{align*}
 \delta_{\psi} \left( \adjustbox{valign = c} {\scalebox{0.5}{  \begin{forest}
for tree = {grow' = 90}, nice empty nodes,
[ ,
            [,
      [\scalebox{2}{$a$}, tier =1]
      [\scalebox{2}{$b$}, tier =1] 
         ]
         [\scalebox{2}{$c$}, tier = 1 ]
 ]
\path[fill=black] (.parent anchor) circle[radius=4pt]
                (!1.child anchor) circle[radius=4pt];
\end{forest}}}\right) = &
  \adjustbox{valign = c} {\scalebox{0.5}{  \begin{forest}
for tree = {grow' = 90}, nice empty nodes,
            [
      [\scalebox{2}{$a$}, tier =1]
      [\scalebox{2}{$b$}, tier =1] 
         ]
\path[fill=black] (.parent anchor) circle[radius=4pt];
\end{forest}}} \odot
  \adjustbox{valign = c} {\scalebox{0.5}{  \begin{forest}
for tree = {grow' = 90}, nice empty nodes,
[
         [\scalebox{2}{$c$}, tier = 1 ]
 ]
\end{forest}}} 
- 
  \adjustbox{valign = c} {\scalebox{0.5}{  \begin{forest}
for tree = {grow' = 90}, nice empty nodes,
[ ,
      [\scalebox{2}{$a$}, tier =1]
      [\scalebox{2}{$b$}, tier =1] 
         [\scalebox{2}{$c$}, tier = 1 ]
 ]
\path[fill=black] (.parent anchor) circle[radius=4pt];
\end{forest}}}
  + \adjustbox{valign = c} {\scalebox{0.5}{  \begin{forest}
for tree = {grow' = 90}, nice empty nodes,
[ ,
            [,
      [\scalebox{2}{$da$}, tier =1]
      [\scalebox{2}{$b$}, tier =1] 
         ]
         [\scalebox{2}{$c$}, tier = 1 ]
 ]
\path[fill=black] (.parent anchor) circle[radius=4pt]
                (!1.child anchor) circle[radius=4pt];
\end{forest}}}
+(-1)^{|a|}\adjustbox{valign = c} {\scalebox{0.5}{  \begin{forest}
for tree = {grow' = 90}, nice empty nodes,
[ ,
            [,
      [\scalebox{2}{$a$}, tier =1]
      [\scalebox{2}{$db$}, tier =1] 
         ]
         [\scalebox{2}{$c$}, tier = 1 ]
 ]
\path[fill=black] (.parent anchor) circle[radius=4pt]
                (!1.child anchor) circle[radius=4pt];
\end{forest}}}
\\ \\
&+(-1)^{|a|+|b|}\adjustbox{valign = c} {\scalebox{0.5}{  \begin{forest}
for tree = {grow' = 90}, nice empty nodes,
[ ,
            [,
      [\scalebox{2}{$a$}, tier =1]
      [\scalebox{2}{$b$}, tier =1] 
         ]
         [\scalebox{2}{$dc$}, tier = 1 ]
 ]
\path[fill=black] (.parent anchor) circle[radius=4pt]
                (!1.child anchor) circle[radius=4pt];
\end{forest}}}
+ \adjustbox{valign = c} {\scalebox{0.5}{  \begin{forest}
for tree = {grow' = 90}, nice empty nodes,
            [
      [\scalebox{2}{$a\star b$}, tier =1]
      [\scalebox{2}{$c$}, tier =1] 
         ]
\path[fill=black] (.parent anchor) circle[radius=4pt];
\end{forest}}} - \psi \left( \adjustbox{valign = c} {\scalebox{0.5}{  \begin{forest}
for tree = {grow' = 90}, nice empty nodes,
[ ,
            [,
      [\scalebox{2}{$a$}, tier =1]
      [\scalebox{2}{$b$}, tier =1] 
         ]
         [\scalebox{2}{$c$}, tier = 1 ]
 ]
\path[fill=black] (.parent anchor) circle[radius=4pt]
                (!1.child anchor) circle[radius=4pt];
\end{forest}}}\right),
\end{align*}
where $|a|, |b|, |c| <-1$ and $a\star b = \psi\left( \adjustbox{valign = c} {\scalebox{0.5}{  \begin{forest}
for tree = {grow' = 90}, nice empty nodes,
            [
      [\scalebox{2}{$a$}, tier =1]
      [\scalebox{2}{$b$}, tier =1] 
         ]
\path[fill=black] (.parent anchor) circle[radius=4pt];
\end{forest}}}\right)$.

\section{Main results}\label{sec:2}

In this section, we state the main results of the paper. Recall that $\mathcal{O}$ is an associative, commutative, unital algebra and $\cI\subset \cO$ is a proper ideal of $\cO$. The relevant geometric examples of $\mathcal{O}$ are an algebra of functions on a smooth manifold, a Stein manifold, or a polynomial ring $\mathbb{K}[x_1, \dots, x_n]$. In the latter case, the quotient algebra $\mathcal{O}/\mathcal{I}$ corresponds to a coordinate ring of an algebraic variety.
\medskip
This section investigates the extension problem for positively graded $Q$-varieties by the arborescent Koszul-Tate resolution of $\cO/\cI$ in two distinct settings:
\begin{enumerate}
    \item the case of a positively graded $Q$-variety over $\mathcal{O}/\mathcal{I}$, which is the subject of Theorem \ref{thm:gen.case.computations};
    \item the case of a positively graded $Q$-variety over $\mathcal{O}$ such that $Q$ preserves $\mathcal{I}$, which is delegated to Theorem \ref{thm:main2}.
\end{enumerate}
Although the second case is a particular instance of the first case (it yields a positively graded $Q$-variety over $\mathcal{O}/\mathcal{I}$), there are two distinct features appearing:
\begin{enumerate}
    \item The assumption imposed on $\mathcal{O}$ in Theorem~\ref{thm:gen.case.computations} can be avoided. Namely, the (smooth) K\"ahler module $\Omega_{\cO/\mathbb K}$ need not be projective.
    \item  The extension of Theorem \ref{thm:main2} is inherently equipped with a homotopy retract data coming from the arborescent Koszul-Tate resolution of $\mathcal{O}/\mathcal{I}$. This allows a more explicit description, as stated in Proposition \ref{prop:Q.explicit.form}.
\end{enumerate}
The extension problem is known in many applications, for instance in BV and BFV formalisms \cite{HT,stasheff_poisson, Barnich_1998, Felder-Kazhdan}. The main tool is a homological perturbation lemma, which provides a step-by-step algorithm and is based on the acyclicity of the $Q$-variety in negative degrees; see, e.g., Theorem 3.21 in \cite{KOTOV2023104908} in the context of our scope. In practice, finding such extensions explicitly can be a highly non-trivial task. The extension requires calculating a preimage of the Koszul-Tate resolution at every negative degree, and such calculations can be infinitely many. In this section, we present an enhanced algorithm that significantly simplifies the complexity in a class of examples. The cornerstone of this new technique is the use of arborescent Koszul-Tate resolutions as the negative part of $Q$.
Let us address the issue of complexity and calculations. In our work, we are dealing with two types of computational problems:
\begin{enumerate}
    \item  Finding a kernel of an $\cO$-linear map $\phi: A \rightarrow B$.
    \item Constructing an $\cO$-linear map $\gamma:A \rightarrow B$, $A$ being a projective module, in a diagram as depicted below :
    \begin{tikzcd} A \arrow[r, dotted, "\gamma"] \arrow[dr] & B \arrow[d,twoheadrightarrow]  \\ & C  \end{tikzcd}
\end{enumerate}
Although the difficulty of finding a map $h$ in the second computational problem is highly dependent on $C$ (if $C$ is zero, $\gamma = 0$ is an admissible choice), we say that the homological computations are \emph{restricted} to $(A, B)$. In case there is a collection of $\cO$-modules $ \lbrace \lbrace A_i, B_i\rbrace ,\, i\in I\rbrace$ for some indexed set $I$, we call the construction of appropriate morphisms $\gamma_i\colon A_i \rightarrow B_i$ as doing homological computations restricted to $\lbrace (A_i, B_i),\, i \in I\rbrace$. 
The classical Tate algorithm is essentially focused on problems of the first type, while the algorithm to find arborescent Koszul-Tate resolutions has both problems equally present. The first type is found in the construction of the resolution $(\FM, d)$ of $\cO/\cI$. The second type of computation problem is to find an appropriate hook  map $\psi\colon \mathcal{T}ree[\FM]\to \FM$ which is associated with the arborescent Koszul-Tate resolution $(S(\mathcal{T}ree[\FM]), \delta_\psi)$. The standard technique for extending $NQ$-varieties used in Theorem 3.21 in \cite{KOTOV2023104908} only has problems of the second type. Let us introduce our modified technique in detail.

\subsection{Arborescent extension of $\mathbb N$-graded $Q$-varieties over $\cO/\cI$}\label{sec:2.1}
 One of the crucial technical details in the proof of Theorem 3.21 \cite{KOTOV2023104908} is the lifting property of derivations of $\cO/\cI$ to $\cO$. We cover this case using the assumption that the K\"ahler module $\Omega_{\cO/\mathbb K}$ is projective, see Lemma \ref{lem:kahler}. {If $\cO$ is the ring of smooth functions $C^{\infty}(M)$ of a smooth manifold $M$, we substitute the K\"ahler module by the module of de Rham 1-forms, which is always projective, thereby recovering the smooth framework considered in \cite{KOTOV2023104908}, see Remark \ref{rem:smooth.kahler} and Convention \ref{convention-smooth}.} In this setting, we will need to lift the homological derivation of a positively graded  variety $(\cA^+, Q^+)$ over $\cO/\cI$  to a derivation over $\cO$.  

 {Throughout this section, following \cite{KOTOV2023104908}, we restrict our attention to positively graded varieties over $\cO/\cI$ whose homogeneous components are projective $\cO$-modules. More precisely, \begin{definition}
Let $\mathcal I\subset\mathcal O$ be an ideal. A \emph{positively graded variety over $\mathcal O/\mathcal I$} is a positively graded variety (Definition \ref{def:Z-variety}) of the form
\(
\cA^+ = S_{\cO/\cI}(\oplus_{i\geq 1} \mathcal V_{i} ),
\)
where each $\mathcal V_i$ is a projective $\mathcal O$-module.
\end{definition}}
\begin{remark}
{The above definition should be regarded as a special case of Definition~\ref{def:Z-variety}. Indeed, we only consider positively graded varieties over $\mathcal O/\mathcal I$ whose homogeneous components arise from projective
$\mathcal O$-modules\footnote{Recall that if $\mathcal V$ is a projective module over $\mathcal O$, then $\mathcal V\otimes_{\mathcal O}\mathcal O/\mathcal I\simeq \cV/\cI\cV$ is a projective module over $\mathcal O/\mathcal I$.}. Not every projective
$\mathcal O/\mathcal I$-module admits such a lift (see, e.g., \cite{rosenberg_algebraic_1994}). 
However, this class of varieties is sufficiently broad for our purposes and includes all the
examples considered in this article.}
\end{remark}

\begin{definition}\label{def:extensionQ^+}
	Let $(\cA^+ = S_{\cO/\cI}(\oplus_{i\geq 1} \mathcal V_{i}),\; Q^+)$ be a positively graded $Q$-variety over $\cO/\cI$. 
    An \emph{$\cO$-lift} of $(\cA^+, Q^+)$ over $\mathcal{O}$ is a pair $(\hat\cA^+\coloneqq S(\oplus_{i\geq 1} \mathcal V_{i} ),\, \hat Q^+)$ such that $\hat Q^+\in \mathrm{Der}(S(\oplus_{i\geq 1} \mathcal V_{i}))$ is a lift of $Q^+\in\mathrm{Der}( S_{\cO/\mathcal
    I}(\oplus_{i\geq 1} \mathcal V_{i} ))$ in the sense of Proposition \ref{prop:Qlift}. 
\end{definition}

We proceed with the following technical definition:
\begin{definition}
\label{def:arb.extension} Let $\cI\subset \cO$ be an ideal. Let $(\cA^+=S_{\cO/ \cI}(\oplus_{i\geq 1} \mathcal V_{i}), Q^+)$ be a positively graded $Q$-variety over $\cO/\cI$ as in Definition \ref{def:extensionQ^+}, and let $(\cA^-=S(\oplus_{i\geq 1} \mathcal V_{-i}), Q^-)$ be a  Koszul-Tate resolution of $\cO/\cI$. 

\begin{enumerate}
    \item A \emph{$\mathbb Z$-graded extension} of $(\cA^-,\cA^+)$ is a $\mathbb Z$-graded manifold/variety $$(\cA=\cA^-\bar \odot \hat\cA^+=\bar S(\oplus_{i\in \mathbb Z^\times} \mathcal V_{i}), Q)$$ over $\mathcal{O}$ such that its negative and positive components are given exactly by $(\cA^-, Q^-)$ and $(\cA^+, Q^+)$ in the sense of \S \ref{+and-components}. We shall say that $(\mathcal{A}, Q)$ is a \emph{$\mathbb Z$-graded extension} of $(\mathcal A^-, \mathcal A^+)$. 
    \item 
If $(\mathcal A^-, Q^-)$ is an arborescent Koszul-Tate resolution $(S(\cTr[\FM],\delta_{\psi})$ of $\cO/\cI$, then  an \emph{arborescent extension $(\mathcal A, Q, \alpha, \beta)$ of $(\mathcal A^-, \mathcal A^+)$} is a $\mathbb Z$-graded extension $(\mathcal A, Q)$ such that the homological vector field $Q$ decomposes as
\[
Q = \delta_\psi+\nabla_\alpha + {L_\beta},
\]
where:
\begin{enumerate}
\item $\delta_\psi$ is extended by $\hat\cA^+$-linearity on $\cA$;
    \item ${\nabla_\alpha}$ is a degree $+1$  derivation of $\mathcal A$ so that $({\nabla_\alpha})_{(0)}$ extends $\hat Q^{+}$ to $\cA^-$; 
    \item ${L_\beta} \colon \mathcal A \to \mathcal A$ is an $\cO$-linear derivation of degree $+1$ (referred to the \emph{total $\mathcal{O}$-linear part of $Q$}) so that {$({L_\beta})_{(0)}$ is $\hat\cA^+$-linear.}
\end{enumerate}

so that the negative degree $(Q_{(k)})_{k\geq -1}$ components of $Q$ are defined recursively through the operators $\delta_\psi, {\nabla_\alpha},{L_\beta}$ as follows:

\medskip

\begin{itemize}
    \item the negative degree  $-1$ component $Q_{(-1)}$  is $\delta_\psi$.
    \item \begin{equation}
    \label{eq:def.Q_0}
    Q_{(0)}(a) = \begin{cases}
        ({\nabla_\alpha})_{(0)}(a)=\hat Q^+(a), \hbox{ if $a\in \hat \cA^+$}; \\\\
        ({\nabla_\alpha})_{(0)}(a) -\underbrace{\r\circ \p\circ \left(\left[\delta_{\psi},({\nabla_\alpha})_{(0)}\right]+({L_\beta})_{(0)}\circ\delta_\psi\right)(a) - \beta_{(0)}(a)}_{({L_\beta})_{(0)}}, \hbox{ if $a \in \cTr[\FM]_{(\bullet)}$}.
    \end{cases}\end{equation}
    \item For $N\geq 0$,  the negative degree $N+1$ component $Q_{(N+1)}$ is given recursively as:
    \begin{equation}
    \label{eq:def.Q_N}
    \scalebox{0.9}{\hbox{$ {Q_{(N+1)}(a) = \begin{cases}
         ({\nabla_\alpha})_{(N+1)}(a)=- \sum\limits_{\substack{i+j=N \\ 0\leq i\leq j}}\r\circ \p \circ\left[Q_{(i)},  Q_{(j)}\right](a) - \alpha_{(N+1)}(a), \hbox{ if $a\in \cO$} \\ \\
        ({\nabla_\alpha})_{(N+1)}(a)-\underbrace{\r\circ \p \circ \left(\left[\delta_{\psi}, ({\nabla_\alpha})_{(N+1)}\right] +({L_\beta})_{(N+1)}\circ\delta_\psi +\sum\limits_{\substack{i+j=N \\ 0\leq i\leq j}}\left[Q_{(i)}, Q_{(j)}\right]\right)(a) - \beta_{(N+1)}(a)}_{({L_\beta})_{(N+1)}},\\\\ \hbox{ if $a\in \cTr[\FM]_{(\bullet)}\oplus\mathcal V_{\geq 1}$}.
    \end{cases}}$}}\end{equation}
\end{itemize}
Here, $\alpha$ and $\beta$ are respectively $\mathbb{K}$-linear and $\mathcal{O}$-linear maps of degree 
$+1$\begin{equation}
    \begin{cases}
    \alpha\colon \cO\to \FM\bar \odot \hat\cA^+\quad\hbox{with$\quad\alpha_{(0)}=0$}\\\\        \beta\colon \cTr[\FM]\oplus \cV_{\geq 1}\to \FM\bar \odot \hat\cA^+ \quad \text{with}\quad \beta_{(0)}|_{\cV_{\geq 1}}=0
    \end{cases}\end{equation}
which are referred to as {retraction residues} of $Q|_{\mathcal O}$ and $Q|_{\cTr[\FM]\oplus\mathcal V_{\geq 1}}$, respectively. Those maps introduced in \S \ref{sec:2}, namely,  $\psi\colon \cTr[\FM]\to \FM$; $\r \colon S^{\geq 2}(\cTr[\FM]) \cong \cTr^{\geq 2}[\FM]$; $\p\colon\, S(\cTr[\FM]) \longrightarrow S^{\geq 2}(\cTr[\FM])$; and $\ptr \colon\, S(\cTr[\FM]) \longrightarrow \FM$ are extended by $\hat\cA^+$-linearity to $\cA$.
\end{enumerate}
\end{definition}

\begin{remark}
Definition \ref{def:arb.extension}(2) is inspired by the proof of Theorem \ref{thm:gen.case.computations} below, which we establish in \S \ref{sec:result1}. In that proof, the maps $\alpha$ and $\beta$ are obtained after refining  $Q|_{\mathcal O}$ and $Q|_{\cTr[\FM]\oplus\mathcal V_{\geq 1}}$ while taking into account the homotopy retraction  of  $(S(\cTr[\FM],\delta_{\psi})$ on $(\FM ,d)$ \begin{equation}\label{retration0}
    \delta_{\psi}\circ \underbrace{(\r\circ \p)}_{h} +  (\r\circ \p) \circ \delta_\psi = \mathrm{Id} - \underbrace{(\ptr + \psi\circ \r\circ \p)}_{p}.
\end{equation} This motivates the terminology retraction residue. It is important to note that these maps are not unique, but are determined only up to $\delta_\psi$-boundaries.    
\end{remark}

In the smooth setting, the existence of $\mathbb Z$-graded extensions and the uniqueness of the holomogical vector field $Q$ are granted by \cite[Theorem 3.21]{KOTOV2023104908}. We have to mention that the extension result in \cite[Theorem 3.21]{KOTOV2023104908} does not take into account the arborescent structure of the Koszul--Tate resolution.


The definition of an arborescent extension above appears rather intricate and lengthy; however, it provides a more explicit description of the homological vector field $Q$, and the homological computations become significantly more constrained.

The following results constitute the main statements of the paper.

{\begin{theorem}
\label{thm:gen.case.computations} Let $\mathcal{O}$ be an associative commutative unital algebra such that the (smooth) K\"ahler module {$\Omega_{\cO/\mathbb K}$} is a projective $\mathcal{O}$-module. \begin{enumerate}
    \item Let $\cI\subset \cO$ be a proper ideal of $\cO$ and $(\cA^+, Q^+)$, $\cA^+ = S_{\cO/\cI}(\oplus_{i\geq 1}\cV_i)$  a positively-graded $Q$-variety over $\cO/\cI$ for a collection of projective $\cO$-modules $\lbrace\cV_i, i\geq 1\rbrace$;
    \item Let $(S(\cTr[\FM]), \delta_{\psi})$ be an arborescent Koszul-Tate resolution of $\cO/\cI$ with hook map $\psi\colon \mathcal{T}ree[\FM]\to \FM$.
\end{enumerate} Then there exist retraction residues \begin{equation}
    \begin{cases}
    \alpha\colon \cO\to \FM\bar \odot \hat\cA^+\\\\        \beta\colon \cTr[\FM]\oplus \cV_{\geq 1}\to \FM\bar \odot \hat\cA^+
    \end{cases}\end{equation} so that   $(\cA^+, Q^+)$ admits an arborescent extension $(\cA=\cA^-\bar \odot\hat\cA^+, Q, \alpha, \beta)$ with $(\cA^-, Q^-)=(S(\cTr[\FM]), \delta_{\psi})$. In particular,  
        the homological computations are restricted\footnote{See the introduction of \S \ref{sec:2} for the meaning of “restricted” here. } to the collection of $\cO$-modules \[\left\{(\cTr[\FM]_{(i)},\;\FM_{(i+j-1)}\bar \odot \hat\cA^+_{j} )\bigcup  (\cV_i, \FM_{(j)}\bar \odot\hat\cA^+_{i+j+1} )\bigcup(\Omega_{\cO/\mathbb K}, \FM_{(i)}\bar \odot \hat\cA^+_{i+1})| \; i,j \geq 1 \right\}.\]
Here, $\hat{\mathcal{A}}^+:=S(\oplus_{i\geq 1}\cV_i)$  and $$
\begin{tikzcd}
  \cdots \arrow[r,""]  & \mathfrak M_{-i}\arrow[r,""] & \cdots \arrow[r, ""] & \FM_{-1} \arrow[r, ""] & \cO/\mathcal I \arrow[r] & 0
\end{tikzcd}
 $$
 is a free/projective resolution of $\cO/\mathcal I$.
\end{theorem}

The homological computations needed to obtain an arborescent extensions in Theorem \ref{thm:gen.case.computations} are exactly those needed to find suitable retraction residues $\alpha$ and $\beta$.
\begin{corollary}In Theorem \ref{thm:gen.case.computations}, 
if $\oplus_{i\geq 1}\FM_{(i)}, \oplus_{j\geq 1}\cV_{j}$  are finite collections of finitely generated $\cO$-modules, then there are finitely many homological computations needed to obtain an arborescent extension $(\cA, Q, \alpha, \beta)$.  
    In particular, if $\cO$ is a polynomial ring in $n$ variables, a $\mathbb Z$-graded $Q$-variety over $\cO$ such that 
    \begin{itemize}
    	\item [$\bullet$] its negative part is an arborescent Koszul-Tate resolution of $\cO/\cI$
    	\item[$\bullet$] its positive part is a universal Lie $\infty$-algebroid associated to a Lie-Rinehart algebra over $\cO$
    \end{itemize} can be constructed in a finite number of homological computations.  
\end{corollary}
\begin{proof}
If $\oplus_{i\geq 1}\FM_{(i)}, \oplus_{j\geq 1}\cV_{j}$  are finite collections of finitely generated $\cO$-modules, then the claim is a direct consequence of the statement of Theorem $\ref{thm:gen.case.computations}$. If $\cO$ is a polynomial ring, using Hilbert Syzygy theorem, we can always obtain both a resolution of $\cO/\cI$, as well as a resolution of the subalgebra of vector fields on $\cO$ of finite length and finite rank at each degree. Also, the ranks of the $\mathcal V_j$'s  in the universal positively graded $Q$-variety of a Lie-Rinehart algebra of Proposition \ref{thm;Universal2} are finite at each degree, so is the number of homological computations needed to construct $Q^+$. From  \cite[Theorem 5.1]{hancharuk2024} the arborescent Koszul-Tate can be obtained in a finite number of homological computations. The claim follows.
\end{proof}

\subsection{ 
Proof of Theorem \ref{thm:gen.case.computations}}\label{sec:result1}
In this section, we prove Theorem \ref{thm:gen.case.computations}. We refine Theorem 3.21 of  \cite{KOTOV2023104908} by presenting a precise algorithm to extend a positively graded $Q$-manifold/variety over $\mathcal{O}/\mathcal{I}$ to a $\mathbb{Z}$-graded one, whose negative component corresponds to the arborescent Koszul-Tate resolution of $\mathcal{O}/\mathcal{I}$.  We show that this algorithm terminates in a finite number of steps when $\mathcal{O}/\mathcal{I}$ admits a projective resolution of finite length and of finite ranks. This enables us to compute explicit examples of such extensions, which are illustrated in \S \ref{sec:3}. In \S \ref{sec:arb,extension}, we further refine the extension technique for positively graded $Q$-manifolds over any unital commutative algebra $\cO$. In particular, for a universal $\mathbb{N}$-graded manifold of a Lie-Rinehart algebra that preserves $\mathcal{I}$.

We now prove the main theorem of the paper
\begin{proof}[Proof (of Theorem \ref{thm:gen.case.computations})] The strategy of the proof is to carefully look at the computational steps of the homological perturbation technique. The construction of the homological vector field $Q$ on $\mathcal A=\cA^-\bar \odot \hat\cA^+=\bar S(\oplus_{i\in \mathbb Z^\times} \cV_i)$ goes into three main repetitive stages. \begin{itemize}
    \item [(a)]The first one is a construction of a $\mathbb K$-linear derivation of $\cO$ valued in $\cA$ that coincides with $\hat Q^+|_{\mathcal O}$, and is extended to a {derivation} $ \nabla$ on $\cA$. 
    \item [(b)]The second step is the construction of an $\cO$-linear derivation (referred to as an \emph{$\mathcal{O}$-linear part} of $Q$)  $L\colon \cA \longrightarrow \cA$ such that $Q=\delta_\psi + \nabla + L$.

    \item[(c)] The third step consists in refining  $\nabla $ and $L$ by means of the homotopy retract data of $(S(\cTr[\FM]),\delta_{\psi})$. Recall that from  Theorem \ref{thm:arborescentKT}  that $(S(\cTr[\FM],\delta_{\psi})$ is a homotopy retract on $(\FM ,d)$ with
    \begin{equation}
    \label{eq:hom.retract2}
        \delta_{\psi}\circ (\r\circ \p) +  (\r\circ \p) \circ \delta_\psi = \mathrm{Id} - (\ptr + \psi\circ \r\circ \p)
    \end{equation}
This will prove in particular that its homological computations are restricted to the collection of modules \[\left\{(\cTr[\FM]_{(i)},\;\FM_{(i+j-1)}\bar \odot \hat\cA^+_{j} )\bigcup  (\cV_i, \FM_{(j)}\bar \odot\hat\cA^+_{i+j+1} )\bigcup (\Omega_{\cO/\mathbb K}, \FM_{(i)}\bar \odot \hat\cA^+_{i+1})| \; i,j \geq 1 \right\}.\]

To do that, we need to extend the maps $\psi\colon \cTr[\FM]\to \FM$; $\r \colon S^{\geq 2}(\cTr[\FM]) \cong \cTr^{\geq 2}[\FM]$; $\p\colon\, S(\cTr[\FM]) \longrightarrow S^{\geq 2}(\cTr[\FM])$; and $\ptr \colon\, S(\cTr[\FM]) \longrightarrow \FM$ by $\hat\cA^+$-linearity to $\cA$.
\end{itemize}
	
	Let us now go into details. 

    Notice that for \(N \in \mathbb{N}\cup\{-1\}\), the requirement that \(Q\) be a differential, i.e., \(Q^2 = 0\), imposes the following condition in negative degree \(N\):
\begin{equation}
\label{eq:general.case}
0=(Q^2)_{(N)} = \delta_\psi \circ Q_{(N+1)} + Q_{(N+1)} \circ \delta_\psi + \sum_{\substack{i+j=N \\ i,j \geq 0}} Q_{(i)} \circ Q_{(j)}, \quad N \geq -1.
\end{equation}
This means that $\left(\sum_{i\geq -1}^{N+1} Q_{(i)}\right)^2_{(N)}=0$, and the component \(Q_{(N+1)}\) must serve as a preimage of a certain $[\;\cdot\,,\delta_\psi]$-cycle, which depends on the lower negative degree components $Q_{(0)}, \ldots, Q_{(N)}$. In particular, for all $f\in\cO$ 

\begin{equation}
\label{eq:general.case2}
 \delta_\psi \circ Q_{(N+1)}(f) = - \sum_{\substack{i+j=N \\ i,j \geq 0}} Q_{(i)} \circ Q_{(j)}(f), \quad N \geq -1.
\end{equation}
This forces the r.h.s of Equation \eqref{eq:general.case2} to be a $\delta_\psi$-boundary, in particular a $\delta_\psi$-cycle. {In the construction of $Q$, we proceed by a double recursion: first on the negative-degree component $Q_{(N)}$ of $Q$, and the other on the negative degree of $\cTr[\mathfrak M]_{(\bullet)}$}.

\noindent 
\textbf{Step 1: Construction of $Q_{(-1)}$ and $ Q_{(0)}$}. We let $Q_{(-1)}:=\delta_\psi$ which is extended by $\hat{\cA}^+$-linearity on $\mathcal{A}=\mathcal{A}^-\bar \odot \hat\cA^+$. Notice that the extension $(\cA, \delta_{\psi})$ is acyclic on $\hat\cA_{(\geq 1)}$: This holds for the graded symmetric algebra 
$\mathcal A^- \odot \widehat{\mathcal A}^+ = S(\oplus_{i \ge 1} \mathcal V_{-i}) \odot S(\oplus_{i \ge 1} \mathcal V_{+i})$, 
as a consequence of the vanishing 
$\mathrm{Tor}(\mathcal O/\mathcal I, \widehat{\mathcal A}^+) = 0$, 
together with the fact that 
$\widehat{\mathcal A}^+ = S(\oplus_{i\geq 1} \mathcal V_{i})$ 
is projective. Using this observation, for every $\delta_\psi$-closed Cauchy sequence in the completion
\[
\mathcal A=\mathcal A^- \,\bar\odot\, \widehat{\mathcal A}^+=\bar S(\oplus_{i\in \mathbb Z^\times} \mathcal V_{i}),
\]
we construct a $\delta_\psi$-preimage, which is itself a Cauchy sequence. See also \cite[Lemma~3.7]{KOTOV2023104908}. 

\noindent
\textbf{Step 1(a): Construction of a $\mathbb K$-linear derivation ${\nabla}_{(0)}$.} Equation \eqref{eq:general.case} translates into $Q_{(0)}$ and $\delta_{\psi}$ to commute. We can choose $Q_{(0)}|_{\hat \cA^+} := \hat Q^+$ on $\hat \cA^+=S(\oplus_{i\geq 1} \mathcal V_{i} )$, where $\hat Q^+$ is an extension of $Q^+$ as in Definition \ref{def:extensionQ^+}. Clearly, $\delta_{\psi}\circ \hat Q^+ + \hat Q^+\circ \delta_{\psi} = 0$ on $\hat\cA^+$ by $\hat \cA^+$-linearity of $\delta_\psi$. Now we need to extend $Q_{(0)}|_{\hat \cA^+}$ to $\cTr[\FM]$: to do so, we introduce a derivation $\nabla_{(0)}\colon \cA \to \cA$ 
which extends the derivation $\hat Q^+|_{\hat \cA^+}$ to $\mathcal{A}^{-}=S(\mathcal{T}ree[\FM])$. We choose $\nabla_{(0)}$ to be an $\cA$-valued  derivation of $\mathcal A^-$ of degree $+1$ and negative degree $0$; it  preserves $\mathcal{I}$ by construction.

\noindent
\textbf{Step 1(b): A construction of an $\cO$-linear part $L_{(0)}$}. One searches for $Q_{(0)} = \nabla_{(0)} + L_{(0)}$, for some $\hat \cA^+$-linear derivation $L_{(0)}$ of $\cA$.
 
Since $Q_{(0)}$ and $\delta_{\psi}$ need to commute, the derivation $L_{(0)}$  must satisfy the following equation
 \begin{equation}
 	\label{eq:gen.case.A}
 	\delta_{\psi}L_{(0)} (a) = -(\delta_{\psi}\circ  \nabla_{(0)} + \nabla_{(0)}\circ \delta_{\psi} + L_{(0)}\circ \delta_{\psi}) (a)
 \end{equation}
 for $a \in \cTr[\FM]_{(i)}$ for all $i \geq 1$. Now we employ the standard homological perturbation technique to construct a particular $L_{(0)}$. The construction is done by induction on negative degrees $i\in \mathbb N$ of $\cTr[\FM]$.
 \begin{itemize}
 	\item[$\bullet$] Let us fix $i=1$.

    We first define $L_{(0)}$ on $\cTr[\FM]_{(1)}$. Let $a\in\cTr[\FM]_{(1)} = \FM_{(1)} $. The fact that $\nabla_{(0)}$ preserves $\cI$ yields $(\delta_{\psi}\circ \nabla_{(0)} + \hat\nabla_{(0)}\circ \delta_{\psi}) (a) \in \cI\hat \cA^+$. Since $\cI\hat \cA^+ = \delta_{\psi}(\FM_{(1)}\bar \odot \hat\cA^+)$, the latter is $\delta_{\psi}$-exact. By projectivity of $\cTr[\FM]_{(1)}$ there exists an $\cO$-linear map $L_{(0)}$ that makes the following diagram commute:
 	\[
 	\begin{tikzcd}[row sep=large, column sep=large]
 		& \cA_{(1)} \arrow[two heads, d, "\delta_{\psi}"']\\
 		\cTr[\FM]_{(1)} \arrow[dashed, ur, "\exists L_{(0)}"] \arrow[r, "rhs" above] & \cI\hat\cA^+\subset \cA_{(0)}=\hat\cA^+ 
 	\end{tikzcd}
 	\]
 	where $rhs$ is $-(\delta_{\psi}\circ  \nabla_{(0)} + \nabla_{(0)}\circ \delta_{\psi})$.
 	\item[$\bullet$]We assume that we have constructed $L_{(0)}$ solving Equation \eqref{eq:gen.case.A} on $\cTr[\FM]_{(<i)}$ for $i>1$. Let $a\in\cTr[\FM]_{(i)}$.

    Let us write $R_{\leq(0)}$ the sum of  derivations $Q_{(-1)}$ and $R_{(0)}$, where $R_{(0)}$ is an $\cA$-valued derivation of $\cA$ of degree $+1$, negative degree $0$, and coincides with $Q_{(0)}$ on $\cTr[\FM]_{(<i)}$ and $\hat \cA^+$. The trivial identity $0 = R_{\leq(0)}^2 \circ R_{\leq(0)} - R_{\leq(0)} \circ R_{\leq(0)}^2$ has the following component of negative degree $-2$:
    $$
    (R_{\leq(0)}^2 \circ R_{\leq(0)} - R_{\leq(0)} \circ R_{\leq(0)}^2)_{(-2)} = (R_{\leq(0)}^2)_{(-1)}\circ \delta_{\psi} - \delta_{\psi}\circ (R_{\leq(0)}^2)_{(-1)} = 0.
    $$
    Now, let us evaluate this expression on an element $a\in \cTr[\FM]_{(i)}$. By definition of $R_{\leq(0)}^2$, the first summand is equal to $(Q_{(0)}\circ\delta_{\psi} + \delta_{\psi}\circ Q_{(0)})\circ \delta_{\psi}(a)$ and it vanishes due to \eqref{eq:gen.case.A}. The second summand $\delta_{\psi}\circ (R_{\leq(0)}^2)_{(-1)}(a) = 0$ translates into $\delta_{\psi}\circ R_{(0)} \circ \delta_{\psi} (a) =\delta_{\psi}\circ Q_{(0)} \circ \delta_{\psi} (a) = 0$. The latter equality can be rewritten as $( \nabla_{(0)} + L_{(0)})\circ \delta_{\psi}(a)$ is $\delta_{\psi}$ closed, which easily translates into the rhs of Equation \eqref{eq:gen.case.A} being $\delta_{\psi}$-closed. Thus, the rhs of Equation \eqref{eq:gen.case.A} is exact, so for all $a\in \cTr[\FM]$ there exists a $\delta_\psi$-preimage of $-(\delta_{\psi}\circ  \nabla_{(0)} + \nabla_{(0)}\circ \delta_{\psi} + L_{(0)}\circ \delta_{\psi}) (a)$. Therefore, the map $$-(\delta_{\psi}\circ  \nabla_{(0)} + \nabla_{(0)}\circ \delta_{\psi} + L_{(0)}\circ \delta_{\psi})\colon \cTr[\FM]_{(i)}\to \cA_{(i-1)}$$ is $\delta_\psi(\cA_{(i)})$-valued, moreover, it is easily checked to be $\cO$-linear on $\cTr[\FM]_{(i)}$. By projectivity of $\cTr[\FM]_{(i)}$ there exists an $\cO$-linear map  $L_{(0)}$ such that the following  diagram commutes:
 	
 	\[
 	\begin{tikzcd}[row sep=large, column sep=large]
 		& \cA_{(i)} \arrow[two heads, d, "\delta_{\psi}"']\\
 		\cTr[\FM]_{(i)} \arrow[dashed, ur, "\exists L_{(0)}"] \arrow[r, "rhs" above] & \delta_{\psi}(\cA_{(i)})\subset\cA_{(i-1)} 
 	\end{tikzcd}
 	\]
 	where $rhs$ is $-(\delta_{\psi}\circ  \nabla_{(0)} + \nabla_{(0)}\circ \delta_{\psi} + L_{(0)}\circ \delta_{\psi})$.
 \end{itemize} 
 
 \noindent
 \textbf{Step 1(c): A more refined description of $L_{(0)}$ via the homotopy retract data}. We can further refine $L_{(0)}\colon \cA_{}\to \mathcal A_{}$ to include homotopy retract data \eqref{eq:hom.retract2}. Indeed, the identity \begin{equation}\label{identity1}
     \delta_\psi\circ L_{(0)}=-(\delta_{\psi}\circ \nabla_{(0)} + \underbrace{\nabla_{(0)}\circ \delta_{\psi} + L_{(0)}\circ \delta_{\psi}}_{Q_{(0)}\circ \delta_\psi})=rhs
 \end{equation} 
 and the homotopy retract
\begin{equation}\label{retration}
    \delta_{\psi}\circ \underbrace{(\r\circ \p)}_{h} +  (\r\circ \p) \circ \delta_\psi = \mathrm{Id} - \underbrace{(\ptr + \psi\circ \r\circ \p)}_{p}
\end{equation}
 imply that 
\begin{equation}\label{eq:L_0}
   \delta_\psi(L_{(0)}(a))=\delta_\psi\circ h\circ rhs(a)+ p\circ rhs(a),\quad \text{for all}\; a\in \cTr[\FM]_{(i)}, 
\end{equation}
where $rhs$ is $-(\delta_{\psi}\circ  \nabla_{(0)} + \nabla_{(0)}\circ \delta_{\psi} + L_{(0)}\circ \delta_{\psi})$. In particular, the summand $p\circ rhs(a)$ is a $\delta_\psi$-cycle for all $a\in \cTr[\FM]_{(i)}$, therefore, it is $\delta_\psi$-exact. {This allows to search $L_{(0)}$ in the following form: 
\begin{equation}
\label{eq:L0-refined}
L_{(0)} = h\circ rhs - \beta_{(0)},
\end{equation}
where $\beta_{(0)}$ is a $\FM\bar \odot \hat\cA^+_{1}$-valued $\mathcal O$-linear map. The existence of such a map is due to the projectivity of $\cTr[\FM]$, as is in the following diagram below:} 
   	\[
   \begin{tikzcd}[row sep=large, column sep=large]
   	& \FM_{(i)}\bar \odot \hat \cA^+_{1} \arrow[two heads, d, "\delta_{\psi}"']\\
   	\cTr[\FM]_{(i)} \arrow[dashed, ur, "\exists -\beta_{(0)}"] \arrow[r, "p\circ rhs" above] & \delta_{\psi}(\FM_{(i)})\bar \odot \hat \cA^{+}_{1} 
   \end{tikzcd}
   \]
  
In conclusion, for $N=0$ the homological computations are restricted to $(\cTr[\FM]_{(i)}, \FM_{(i)}\bar \odot \hat \cA^+_{1})$ for all $i\geq 1$. This greatly reduces the number of homological computations.\\

\noindent
\textbf{Step 2: Recursive construction of $Q_{(k)}$, $k\geq 1$.} Assume we have constructed $Q=(Q_{(k)})_{0\leq k\leq N}$ such that $(Q^2)_{(\leq N-1)}=0$ for some $N\geq 0$.

The construction of $Q_{(N+1)}$ such that $(Q^2)_{(\leq N)} =0$ goes, in part, along the similar lines as for $Q_{(0)}$.
  
\noindent
\textbf{Step 2(a): Construction of a $\mathbb K$-linear derivation ${\nabla}_{(N+1)}$.} This construction is not as straightforward as in Step 1 (a). One notable difference is that one has to construct an $\mathcal A$-valued $\mathbb K$-linear  derivation of $\cO$ of negative degree $N+1$ which we do not have for free, in contrast to the previous.  It is carried out in the following stages.

\noindent
    \textbf{A construction of $Q_{(N+1)}|_{\cO}$}: From Equation \eqref{eq:general.case2} the rhs of 
   \begin{equation}
   	\label{eq:aux.eq.thm}
   \delta_{\psi}\circ Q_{(N+1)}(f) = - \sum\limits_{\substack{i+j=N \\ i,j\geq 0}}Q_{(i)}\circ Q_{(j)}(f),\quad f \in \cO.
   \end{equation}
   must be $\delta_\psi$-closed. Indeed,
  as previously, one can show that there is no homological obstruction to find such $Q_{(N+1)}$:
  Write $R_{<(N+1)}$ for the sum of derivations $\delta_{\psi} +\sum_{i=0}^N Q_{(i)}$.   Due to the recursive assumption, the negative degree $N-1$ component of the trivial identity $0 = R_{<(N+1)} \circ R_{<(N+1)}^2 - R_{<(N+1)}^2 \circ R_{<(N+1)}$ reads as
  $$
  0=(R^2_{<(N+1)} \circ R_{<(N+1)} - R_{<(N+1)} \circ R^2_{<(N+1)})_{(N-1)} = (R^2_{<(N+1)})_{(N)}\circ \delta_{\psi} - \delta_{\psi} \circ (R^2_{<(N+1)})_{(N)}
  $$
  Upon evalution on $f\in \cO$, the term $(R^2_{<(N+1)})_{(N)}\circ \delta_{\psi}$ vanishes. On the other hand, $(R^2_{<(N+1)})_{(N)}(f) = \sum\limits_{\substack{i+j=N \\ i,j\geq 0}}Q_{(i)}\circ Q_{(j)}(f) $ is $\delta_\psi$-closed, thus exact. It is clear that the map \[\sum\limits_{\substack{i+j=N \\ i,j\geq 0}}Q_{(i)}\circ Q_{(j)}=\sum\limits_{\substack{i+j=N \\ 0\leq i\leq j}}[Q_{(i)}, Q_{(j)}]\colon \mathcal O\to \cA\] is a derivation. By projectivity of  $\Omega_{\cO/\mathbb K}$, the construction of a $\mathbb K$-linear derivation $Q_{(N+1)}$ can be described as follows:
   
   \[
    \begin{tikzcd}[row sep=large, column sep=large]
   	\Omega_{\cO/\mathbb K}\arrow[dashed, rr, "\exists \tilde q_{(N+1)}"] \arrow[drr, "D_{(N)}"{pos=0.4, sloped}]& & \cA_{(N+1)} \arrow[two heads, d, "\delta_{\psi}"']\\
   	\cO \arrow[u, "D"]  \arrow[rr, "rhs" above] \arrow[dotted,urr, ""]& & \delta_{\psi}(\cA_{(N+1)})\subset\cA_{(N)}
   \end{tikzcd}
   \]Here $rhs$ is $\sum\limits_{\substack{i+j=N \\ i,j\geq 0}}Q_{(i)}\circ Q_{(j)}$, and $D_{(N)}$ is the unique $\cO$-linear map determining $rhs$ as a derivation. Then, we define $Q_{(N+1)}|_{\cO}$ as the composition of the universal derivation $D$ with $\tilde q_{(N+1)}$, that is, $Q_{(N+1)}|_{\cO} := \tilde q_{(N+1)}\circ D$. The latter satisfies Equation \eqref{eq:aux.eq.thm}, by construction.

    \noindent
 \textbf{Step 2(c$'$): A refinement $Q_{(N+1)}|_{\cO}$}:
In order to reduce the homological computations, it is possible to refine the construction of the derivation $Q_{(N+1)}|_{\cO}$ by taking into account the homotopy retract data \eqref{eq:hom.retract2} of $(S(\cTr[\FM]), \delta_{\psi})$ on $(\FM, d)$. This is done along the same lines as for Step 1(c). Namely, we obtain from \eqref{eq:hom.retract2} and \eqref{eq:aux.eq.thm}
   \begin{equation}\label{eq:der_N+1}
       \delta_{\psi}(Q_{(N+1)}(f)) = \delta_{\psi}\circ h\circ rhs(f) + p\circ rhs(f), 
   \end{equation}
   where $rhs= -\sum\limits_{\substack{i+j=N \\ i,j\geq 0}}Q_{(i)}\circ Q_{(j)}$. This allows to choose $Q_{(N+1)}|_{\cO}$ in the following form: \begin{equation}\label{eq:Q_{(1)}f}
       Q_{(N+1)}(f) = - \r\circ \p \circ \sum\limits_{\substack{i+j=N\\i,j\geq 0}}Q_{(i)}\circ Q_{(j)} (f) - \alpha_{(N+1)}(f),
   \end{equation}
   for some {$\FM_{(N+1)}\bar \odot\hat\cA^+_{N+2}$}-valued derivation $\alpha_{(N+1)}$ of $\cO$ such that \begin{equation}
   \label{eq:Delta.N}
   	\delta_{\psi}(\alpha_{(N+1)}(f)) =- p\circ rhs(f), \quad \forall f \in \cO.
   \end{equation} 
   By projectivity of $\Omega_{\cO/\mathbb K}$, there exists an $\mathcal O$-linear map $\tilde\alpha_{(N+1)}\colon\Omega_{\cO/\mathbb K}\to \FM_{(N+1)}\bar \odot \hat \cA^+_{N+2}$ so that the following diagram commutes:
 \[
 \begin{tikzcd}[row sep=large, column sep=large]
 	\Omega_{\cO/\mathbb K}\arrow[dashed, rr, "\exists -\tilde \alpha_{(N+1)}"] \arrow[drr, "\tilde D_{(N)}"{pos=0.4, sloped}]& & \FM_{(N+1)}\bar \odot \hat \cA^+_{N+2} \arrow[two heads, d, "\delta_{\psi}"']\\
 	\cO \arrow[u, "D"] \arrow[dotted, urr, ""] \arrow[rr, "p\circ rhs" above] & & \delta_{\psi}(\FM_{(N+1)})\bar \odot \hat \cA^+_{N+2} 
 \end{tikzcd}
 \]
 where $\tilde D_{(N)}$ is the $\cO$-linear map uniquely determined by derivation $p\circ  rhs$. We define $\alpha_{(N+1)} := \tilde \alpha_{(N+1)}\circ D$. The homological computation done through this homotopy retract procedure are restricted to $(\Omega_{\cO/\mathbb K},\, \FM_{(N+1)}\bar \odot \hat \cA^+_{N+2})$, $N\geq 0$.

 Having obtained $Q_{(N+1)}|_{\cO}$ we extend it arbitrary to a derivation $\nabla_{(N+1)}\colon \cA\to \cA$ of degree $+1$ and negative degree $N+1$.

 \noindent
 \textbf{Step 2(b): A construction of an $\cO$-linear part $L_{(N+1)}$}. The $\cO$-linear part of $Q_{(N+1)}$ is obtained similarly to Step 1(b).  We search for $Q_{(N+1)}$ in the form $\nabla_{(N+1)} + L_{(N+1)}$ so that the $\mathcal O$-linear  derivation  $L_{(N+1)}\colon \cA\to \cA$ satisfies Equation \eqref{eq:general.case} on $\cTr[\FM]_{(i)}\oplus\mathcal V_j$, $i,j
 \geq 1$, that is, 
 \begin{equation}\label{eq:N.gen.case.A}
 	\delta_{\psi}( L_{(N+1)}(a)) =  -\left(\delta_{\psi}\circ \nabla_{(N+1)} + \nabla_{(N+1)}\circ \delta_{\psi} + L_{(N+1)}\circ \delta_{\psi} + \sum\limits_{\substack{i+j=N \\ i,j\geq 0}}Q_{(i)}\circ Q_{(j)}\right) (a)
 \end{equation}
 for any $a\in \cTr[\FM]_{(i)}\oplus\mathcal V_j$, $i,j
 \geq 1$. 
 The construction of $L_{(N+1)}$, as well as the consistency check can be done as follows :
 \begin{itemize}
 	\item [$\bullet$] For $a\in \mathcal V_j$, $j\geq 1$, the right hand side of Equation \eqref{eq:N.gen.case.A} is $\sum\limits_{\substack{i+j=N \\ i,j\geq 0}}Q_{(i)}\circ Q_{(j)}(a)$. We proceed then exactly as in Step 2(a) to show that the right hand side of \eqref{eq:N.gen.case.A} is $\delta_{\psi}$-closed\footnote{Replace $f\in\cO$ by an $a\in \cV_j$, and apply the argument of Step 2(a).}, thus exact. So there are no obstructions to find an appropriate  $\cO$-linear $L_{(N+1)}|_{\cV_j}$:
 		\[
 	\begin{tikzcd}[row sep=large, column sep=large]
 		& \cA_{(N+1)} \arrow[two heads, d, "\delta_{\psi}"']\\
 		\mathcal V_j \arrow[dashed, ur, "\exists L_{(N+1)}"] \arrow[r, "rhs" above] & \delta_\psi(\cA_{(N+1)}) 
 	\end{tikzcd}
 	\]
 	\item[$\bullet$] Let us now extend $L_{(N+1)}$ to $\cTr[\FM]$. Let us fix $i \geq 1$. The case $i=1$, i.e., for $a\in \cA_{(0)}$, is done above. Let us write $R_{\leq(N+1)}\coloneqq \delta_{\psi} + \sum_{j=0}^N Q_{(j)} + R_{(N+1)}$, where $R_{(N+1)}$ is an $\cO$-linear derivation $\cA \to \cA$ which coincides with $Q_{(N+1)}$ on $\cA_{(<i)}$. Considering the negative degree $N-1$ component of the trivial identity $R_{\leq(N+1)}^2 \circ R_{\leq(N+1)} - R_{\leq(N+1)}\circ R_{\leq(N+1)}^2 = 0$ and applying it to an element $a\in \cA_{(i)}$ gives:
    $$
    (R_{\leq(N+1)}^2 \circ R_{\leq(N+1)} - R_{\leq(N+1)}\circ R_{\leq(N+1)}^2)_{(N-1)}(a) = ((R_{\leq(N+1)}^2)_{(N)}\circ \delta_{\psi} - \delta_{\psi}\circ (R_{\leq(N+1)}^2)_{(N)})(a) = 0.
    $$
    The first summand $(R_{\leq(N+1)}^2)_{(N)}\circ \delta_{\psi}(a)= (\sum_{j=-1}^{N+1} Q_{(j)})^2_{(N)}(\delta_\psi a) = 0$, by the induction assumption $(\sum_{j=-1}^{N+1} Q_{(j)})^2_{(N)}(a)=~0$ for all $a\in \cA_{(< i)}$. The second summand can be easily checked to be $$- \delta_{\psi} \circ (R_{\leq(N+1)}^2)_{(N)}(a) = -\delta_{\psi}\circ \left( (\sum\limits_{\substack{i+j=N \\ i,j\geq 0}}Q_{(i)}\circ Q_{(j)})(a) + Q_{(N+1)}\circ \delta_\psi\right)(a) = 0.$$ From here we see directly that the rhs of Equation \eqref{eq:N.gen.case.A} is $\delta_{\psi}$-closed, thus exact.  By projectivity of $\cTr[\FM]_{(i)}$, there exists an $\cO$-linear map $L_{(N+1)}$ such that
    \[
 	\begin{tikzcd}[row sep=large, column sep=large]
 		& \cA_{(N+i+1)} \arrow[two heads, d, "\delta_{\psi}"']\\
 		 \cTr[\FM]_{(i)} \arrow[dashed, ur, "\exists L_{(N+1)}"] \arrow[r, "rhs" above] & \delta_{\psi}(\cA_{(N+i+1)}) 
 	\end{tikzcd}
 	\]
 	where $rhs = -(\delta_{\psi}\circ \nabla_{(N+1)} + \nabla_{(N+1)}\circ \delta_{\psi} + L_{(N+1)}\circ \delta_{\psi} + \sum\limits_{\substack{i+j=N \\ i,j\geq 0}}Q_{(i)}\circ Q_{(j)})$. 
    
    \end{itemize}
    
    We the extend $L_{(N+1)}$ by derivation on $\cA_{(i)}$.

\noindent
 \textbf{Step 2(c): A refined description of $L_{(N+1)}$ via the homotopy retract data}. As in the case for $L_{(0)}$, it is possible to refine the required computations for $L_{(N+1)}$. It follows the exact same lines as in Step 1(c). From Equation \eqref{retration}, projectivvity of $\cTr[\FM]$ and $\cV_{\geq 1}$, as well as the exactness of $\delta_
 \psi$-cycles, there exists an $\mathcal{O}$-linear map $$\beta_{(N+1)}\colon \cTr[\FM]\oplus \cV_{\geq 1}\to \FM\bar \odot \hat\cA^+$$ of negative degree $N+1$ and of degree $+1$ so that
 
 \begin{equation}
 	\label{eq:gen.A.final.form}
 L_{(N+1)}(a) = h \circ rhs(a) - \beta_{(N+1)}(a),
 \end{equation}
 where $rhs = -(\delta_{\psi}\circ \nabla_{(N+1)} + \nabla_{(N+1)}\circ \delta_{\psi} + L_{(N+1)}\circ \delta_{\psi} + \sum\limits_{\substack{i+j=N \\ i,j\geq 0}}Q_{(i)}\circ Q_{(j)})$.
 The map $\beta_{(N+1)}$ is obtained  by projectivity of $\cTr[\FM]_{(i)}$ and $\mathcal V_j,\,\, \forall i,j\geq 0$ as in the following diagrams:
 \begin{equation}\label{eq:beta_{N+1}}
 \begin{tikzcd}[row sep=large, column sep=large]
 	& \FM_{(N+1)}\bar \odot \hat \cA^+_{N+j+1} \arrow[two heads, d, "\delta_{\psi}"']  & & & \FM_{(i+N+1)}\bar \odot \hat \cA^+_{i+N+2} \arrow[two heads, d, "\delta_{\psi}"']\\
 	\mathcal V_j \arrow[dashed, ur, "\exists -\beta_{(N+1)}"] \arrow[r, "p\circ rhs" above] & \delta_{\psi}(\FM_{(N+1)}\bar \odot \hat \cA^+_{N+j+1}) & & 	\cTr[\FM]_{(i)} \arrow[dashed, ur, "\exists -\beta_{(N+1)}"] \arrow[r, "p\circ rhs" above] & \delta_{\psi}(\FM_{(i+N+1)}\bar \odot \hat \cA^+_{i+N+2})  
 \end{tikzcd}\end{equation}
 The homological computations here are restricted to $(\mathcal V_j, \FM_{(N+1)}\bar \odot \hat \cA^+_{N+j+1} )\bigcup(	\cTr[\FM]_{(i)}, \FM_{(i+N+1)}\bar \odot \hat \cA^+_{i+N+2})$, $i,j \geq 1$.

The maps $\nabla, L$ correspond to $\nabla_\alpha, L_{\beta}$ in Definition \ref{def:arb.extension} (2), see Equations \eqref{eq:L0-refined}, \eqref{eq:Q_{(1)}f}, \eqref{eq:gen.A.final.form}. This ends the proof.
\end{proof}

\subsection{Arborescent extension of $\mathbb N$-graded $Q$-varieties over $\cO$}
\label{sec:arb,extension}

In \S~\ref{sec:2.1}, we examined the general extension problem for positively graded \( Q \)-varieties over \( \mathcal{O}/\mathcal{I} \). {Here, in contrast with the previous section, we begin with a positively graded \( Q \)-variety $(\hat\cA^+ = S(\oplus_{i\geq 1} \mathcal V_{i} ),\; \hat Q^+)$ over \( \mathcal{O} \) (that is, \( ({\hat{Q}^+})^2 = 0 \) on $\cO$)  such that 
\[
\hat Q^+(\mathcal{I}) \subseteq \mathcal{I}\hat\cA^+.
\]
Of course, the latter induces a positively graded \( Q \)-variety $(\cA^+ = S_{\cO/\cI}(\oplus_{i\geq 1} \mathcal V_{i} ),\; Q^+)$  over \( \mathcal{O}/\mathcal{I} \).  The difference in this section is that we do not need Proposition \ref{prop:Qlift} needed to ensure the existence of an $\cO$-lift as in Definition \ref{def:extensionQ^+}, since the $\cO$-lift is already provided.}

In several important examples (see \S \ref{sec:3}), we deal with positively graded varieties over \( \mathcal{O} \) that preserve the ideal \( \mathcal{I} \). In particular,  
\begin{enumerate}
    \item the universal Lie~\(\infty\)-algebroid of a singular foliation described in~\cite{LLS}, and  
    \item more generally, the universal Lie~\(\infty\)-algebroid associated with a Lie–Rinehart algebra over~\( \mathcal{O} \) \cite{CLRL},
\end{enumerate}  
provide instances of such positively graded \( Q \)-varieties. One might reasonably expect that the number of homological computations required in this simplified setting is smaller than in the general case \S \ref{sec:2.1}.

In this section, we retain the notations introduced in \S \ref{sec:2.1}. The main result of this section is the following
\begin{theorem}\label{thm:main2}

Let $\mathcal{O}$ be an associative commutative unital algebra. \begin{enumerate}
    \item Let $\cI\subset \cO$ be a proper ideal of $\cO$ and $(\hat \cA^+, \hat Q^+)$, $\hat \cA^+ = S(\oplus_{i\geq 1}\cV_i)$  a positively-graded $Q$-variety over $\cO$ such that $Q^+(\cI)\subset \cI\hat\cA^+$;
    \item and $(S(\cTr[\FM]), \delta_{\psi})$ an arborescent Koszul-Tate resolution of $\cO/\cI$ with hook map $\psi\colon \mathcal{T}ree[\FM]\to \FM$.
\end{enumerate} 
There exists an arborescent extension $(\cA, Q, {\alpha=0, \beta})$ of the following form:
\begin{itemize}
    \item[$i.$] The restriction of $Q$ to $\hat \cA^+$ is $\hat Q^+$.
    \item[$ii.$] For all $a \in \FM$, $Q(a)$ is valued in $\FM\bar \odot \hat \cA^+\oplus \hat \cA^+$. In more details,
    $$
Q_{(i)}(a) = \begin{cases} d(a), \hbox{ if $i = -1$} \\
\left({(\nabla_\alpha)}_{(0)} - \beta_{(0)}\right)(a), \hbox{ if $i=0$}\\
-\beta_{(i)}, \hbox{ if $i \geq 1$}
\end{cases}
    $$
    \item[$iii.$] On non-trivial trees $\cTr^{\geq 2}[\FM]$ $Q$ is given recursively by
    $$
Q = \r^{-1}-\r \circ \p \circ Q \circ \r^{-1} - \chi,
$$
where $\chi\colon \cTr^{\geq 2}[\FM]\to \FM\bar \odot\hat \cA^+$ is an $\cO$-linear map, whose negative degree $-1$ component is $\psi$ and  $\chi_{(\geq 0)} = \beta|_{\cTr^{\geq 2}[\FM]}$.
\end{itemize}
We shall call the extension $(\cA, Q, \alpha=0, \beta)$  described above an \emph{explicit arborescent extension} of $(\cA^-, \hat \cA^+)$ with \emph{hook} $\chi$, and denote it by $(\cA, Q_{\chi})$.
\end{theorem}

The proof is delegated to \S \ref{sec:proof.main2}.
The adjective “explicit” in Theorem \ref{thm:main2} is clarified in the following proposition.
\begin{proposition}
\label{prop:Q.explicit.form}
An explicit arborescent extension of Theorem \ref{thm:main2} has the following form:
     \begin{align*}
    \hbox{On $\hat\cA^+\colon$ }&\quad Q_\chi = \hat Q^+. \\
     \hbox{On $\FM\colon$ }&\quad Q_\chi = \delta_\psi + (\nabla_\alpha)_{(0)}- \beta. \\
    \hbox{On $\cTr^{\geq 2}[\FM]\colon$ }&\quad Q_\chi(t[a_1, \dots, a_n]) = \r^{-1}t[a_1,\dots, a_n]+ \sum_{A\in \mathrm{InVert}(t)}(-1)^{W_A}\partial_A t[a_1,\dots, a_n] \\ +\sum_{A\in \mathrm{Leaves}(t)}&(-1)^{W_A}t[a_1,\dots,Q(a_A) ,\dots, a_n] - \sum_{A\in \mathrm{InVert}(t)\cup\mathrm{Root}}(-1)^{W_A}t_{\downarrow A}[a_1,\dots,\chi({t_{\uparrow A}}(a_A)) ,\dots, a_n]. 
    \end{align*} 
    Here, we used notations from Convention \ref{conv:uptree} and Proposition \ref{prop:KT.explicit.form}.
\end{proposition}
\begin{proof}
    The proof is a consequence of the Proposition \ref{prop:KT.explicit.form}, which follows from  \cite[Proposition 2.22]{hancharuk2024}.
\end{proof}

Notice the similarity of the explicit arborescent extension and the arborescent Koszul-Tate resolution in Definition \ref{def:arb.KT}. It turns out that there is an analogous statement on the existence of homotopy retract data as in  \cite[Proposition 3.1]{hancharuk2024}. More precisely,

\begin{proposition}

\label{prop:Homotopy}
Let $(\cA, Q_\chi)$ be an explicit arborescent extension  of a positively graded $Q$-variety $(\hat\cA^+, \hat Q^+)$ over $\cO$. Consider the complex $(\FM\bar \odot \hat \cA^+\oplus \hat \cA^+, \hat Q'_\chi)$ where $ Q'_{\chi}$ is a restriction $(Q_\chi)|_{\FM \bar \odot \hat \cA^+\oplus \hat \cA^+}$. The $\cO $-linear maps
$$
\mathrm{Incl}\colon (\FM\bar \odot \hat \cA^+) \oplus \hat \cA^+\longrightarrow \cA, \quad (a, b) \mapsto a+b
$$
and
$$\mathrm{Proj}\colon \cA \longrightarrow (\FM\bar \odot \hat \cA^+)\oplus \hat \cA^+, \quad a \mapsto \chi\circ \r\circ \p(a) + \mathrm{p}_{\FM,\hat \cA^+}(a), $$
are chain maps. Moreover, they are homotopy inverse one to the other. More precisely,
 $$ \left\{ \begin{array}{rcl} {\mathrm{Proj}} \circ {\mathrm{Incl}} &=& {\mathrm{Id}} \\ 
{\mathrm{Incl}} \circ {\mathrm{Proj}}  &= & {\mathrm {Id}} - \left(h \circ Q_\chi + Q_\chi \circ h \right)  \end{array}\right.$$
where the homotopy $h$ is given by $ h:= \r \circ \p 
$. Here, $\mathrm{p}_{\FM, \hat\cA^+}$ is the projection of $\cA$ onto $(\FM \bar \odot\hat \cA^+)\oplus \hat \cA^+$. Also, 
 \begin{equation}\label{eq:additional} h^2 =0 \, , \, h \circ {\mathrm{Incl}} =0 \, , \, {\mathrm{Proj}} \circ h =0. \end{equation}
Conditions \eqref{eq:additional} are called \emph{side relations} in \cite{zbMATH03080175}.
 \end{proposition}
 \begin{proof}
 The proof is done by a direct verification. 
 The restriction map $Q'_{\chi}\colon \,\,\FM\bar \odot \hat \cA^+ \oplus \hat \cA^+ \to \FM\bar \odot \hat \cA^+ \oplus \hat \cA^+$ is well-defined, by Theorem \ref{thm:main2}$(i-ii)$. The inclusion map is clearly a chain map, that is \[\mathrm{Incl}\circ Q'_{\chi} = Q_\chi \circ \mathrm{Incl}.\] 

 The side relations \eqref{eq:additional} are clearly satisfied by the definition of the maps $\mathrm{Incl}$, $\mathrm{Proj}$, $h$. We have  $\mathrm{Proj}\circ \mathrm{Incl} = \mathrm{p}_{\FM, \hat\cA^+}\circ  \mathrm{Incl} = \mathrm{Id}$. Let us consider the negative degree expansion of the composition $\mathrm{Incl}\circ \mathrm{Proj}$: the negative degree zero component reads $$(\mathrm{Incl}\circ \mathrm{Proj})_{(0)} = \mathrm{Incl}\circ \mathrm{p}_{\FM, \cA^+} +\mathrm{Incl}\circ \psi\circ \r\circ \p .$$  By Equation \eqref{eq:hom.retract0}, this expression is equal to $\mathrm{Id} - (h\circ \delta_{\psi} + \delta_{\psi} \circ h)$ with $h,\delta_{\psi}$ being $\hat \cA^+$-linear. Since the components $(Q_\chi)_{(\geq 0)}$ map $\cTr[\FM]$ to $\cTr[\FM]\bar \odot\hat \cA^+$, we therefore deduce the following:
 \begin{align*}
     h\circ (Q_\chi)_{(\geq 0)} + (Q_\chi)_{(\geq 0)}\circ h &= h\circ (Q_\chi)_{(\geq 0)}\circ (\mathrm{Id}- \p + \p) -\underbrace{r\circ \p}_{h} \circ (Q_
 \chi)_{(\geq 0)}\circ \p - \chi_{(\geq 0)}\circ \r\circ \p \\&= - \chi_{(\geq 0)}\circ \r\circ \p.
 \end{align*}
 Combining the expressions in negative degree $0$ and $\geq 1$ we obtain the desired relation \[{\mathrm{Incl}} \circ {\mathrm{Proj}}  = {\mathrm {Id}} - \left(h \circ Q_\chi + Q_\chi \circ h \right)\] It remains to check that $\mathrm{Proj}$ is a chain map. This follows directly from $\mathrm{Incl}$ being an injective chain map, as well as the established identities:
 \begin{align*}
\mathrm{Incl}\circ (Q'_\chi \circ \mathrm{Proj} - \mathrm{Proj}\circ Q_\chi)& = Q_\chi\circ \mathrm{Incl}\circ \mathrm{Proj} - \mathrm{Incl}\circ \mathrm{Proj} \circ Q_\chi    \\
&= Q_\chi\circ (\mathrm{Id} - \left(h \circ Q_\chi + Q_\chi \circ h \right)) - (\mathrm{Id} - \left(h \circ Q_\chi + Q_\chi \circ h \right))\circ Q_\chi = 0.
 \end{align*}
 \end{proof}

\subsubsection{Applications and results in the geometric setting} Theorem \ref{thm:main2} together with those of Laurent-Gengoux, Lavau, and Strobl \cite{LLS}, as well as \cite{CLRL}, yields the following.

\begin{theorem}\label{main:smooth}
The universal positively graded manifold (or variety) $(E= (E_{-i})_{i\in \mathbb N}, \hat Q^+)$ associated with a singular foliation or Lie--Rinehart algebra preserving the ideal $\mathcal{I}$, as constructed in \cite{LLS} and \cite{CLRL}, extends to a $\mathbb{Z}$-graded manifold (or variety) $((E_{i})_{i\in \mathbb{Z}^{\times}}, {Q})$ over $\mathcal{O}$.\begin{enumerate}
    \item This extension can be chosen to be an explicit arborescent extension.
    \item  Moreover, the homological vector field ${Q}$ is unique up to diffeomorphism.
\end{enumerate}
\end{theorem}

\begin{proof}
Let  $(\mathfrak A, [\cdot, \cdot]_\mathfrak A, \rho_\mathfrak A)$ be a Lie-Rinehart algebra over $\mathcal O$ that preserves an ideal $\cI\subset \cO$ . Assume there exists a projective resolution of $\mathfrak A$ by sections of vector bundles \begin{equation}\label{eq:geo_res}
    \cdots \stackrel{\dd} \longrightarrow\mathcal P_{-3} \stackrel{\dd}{\longrightarrow} \mathcal P_{-2} \stackrel{\dd}{\longrightarrow} \mathcal P_{-1} \stackrel{\pi}{\longrightarrow} \mathfrak A
\end{equation} that is, $\mathcal{P}_{-i}\simeq \Gamma(E_{-i})$ for some vector bundle $E_{-i}\to M$ for $i\geq 1$.  By Thereom \ref{thm;Universal}, the resolution \eqref{eq:geo_res} is endowed with a unique (up to homotopy) Lie $\infty$-algebroid over $\mathcal{O}$, whose $1$-ary bracket is $\dd$ and the anchor map is $\rho_\mathfrak{A}\circ\pi$.  The latter corresponds to a positively graded $Q$-variety $(S(\oplus_{i\geq 1} \cV_i), \hat Q^+)$ over $\mathcal{O}$ that preserves $\mathcal I$. Here, $\cV_i = \Gamma(E^*_{-i})$. The result follows from Theorem \ref{thm:main2}.
\end{proof}

In particular,  when $\mathcal{O}$ is a Noetherian ring, \eqref{eq:geo_res} exists, e.g.,  $\mathcal{O}=\mathbb C[x_1,\ldots, x_d]$.  We have the following result
\begin{corollary}\label{cor:affine} Let $\mathcal{O}=\mathbb C[x_1,\ldots,x_n]$ and $W\subset \mathbb C^n$ be an affine variety given by an ideal $\mathcal{I}_W\subset \mathcal{O}$. Let $\mathfrak X_W\subset \mathrm{Der}(\mathcal O)$ be the Lie-Rinehart algebra generated by derivations of $\mathcal{O}$  that preserve $\mathcal{I}_W$. There exists a $\mathbb{Z}$-graded  variety over $\mathcal{O}$ which is an explicit arborescent extension of the restriction of a universal $NQ$-variety of the Lie-Rinehart algebra $\mathfrak X_W$ on $W$.

\end{corollary}

\begin{remark}
Notice that in Corollary \ref{cor:affine}, one can consider a Lie–Rinehart algebra over the quotient $\mathbb C[x_1,\ldots,x_n]/\mathcal I_W$. Moreover, Theorem \ref{thm;Universal} still applies, and we obtain a positively graded $Q$-variety over 
$\mathbb C[x_1,\ldots,x_n]/\mathcal I_W$. In this case, we have an arborescent extension $(\cA, Q, \alpha, \beta)$ as in Theorem \ref{thm:gen.case.computations}, not necessarily an explicit one with $\alpha=0$.
\end{remark}

\subsection{Proof of Theorem \ref{thm:main2}}
\label{sec:proof.main2}
We start with the following lemma
\begin{lemma}
	\label{prop:arb.ext.exists}
	Let $(\cA^-, Q^-)$ be an arborescent Koszul-Tate resolution $(S(\cTr[\FM]), \delta_{\psi})$ of $\cO/\cI$, and let $(\hat \cA^+, \hat Q^+)$ be a positively graded $Q$-variety over $\cO$. \begin{enumerate}
	    \item There exists an arborescent extension of $(\cA^-, \hat \cA^+)$ of the form $(\cA, Q, \alpha = 0,\beta)$, with $\beta|_{\cV_{\geq 1}}=0$.
    \item {The negative degree components $(Q_{(i)})_{i\neq 0}$ are $\hat \cA^+$-linear.}
    \item Moreover, the homological computations are restricted to
	$\lbrace(\cTr[\FM]_{(i)},\;\FM_{(i+j-1)}\bar \odot \hat \cA^+_{j} )| \; i,j \geq 1 \rbrace$.
	\end{enumerate}
\end{lemma}
\begin{proof}We shall use the notations introduced in the proof of Theorem \ref{thm:gen.case.computations}. The argument proceeds by revisiting the steps of that proof, with particular attention to Equations \eqref{retration}-\eqref{eq:beta_{N+1}}, 
which are used to construct the retraction residues $\alpha$ and $\beta$. We show that $\alpha, \beta$ of Theorem \ref{thm:gen.case.computations} can be chosen such that $\alpha_{(i)} = 0$ and $\beta_{(i)}|_{\cV_{\geq 1}} = 0$ for all $i\in \mathbb N$. Then the result follows.

By Definition \ref{def:arb.extension}, $\alpha_{(0)}$ and $\beta_{(0)}=0$  and $\beta|_{\cV_{\geq 1}}=0$. Let $i=1$. Since for all $f\in \cO$, we have $Q_{(0)}^2(f)=(\hat Q^+)^2(f)=0$, Equation \eqref{eq:Delta.N} reads 
$$
	\delta_{\psi}(\alpha_{(1)}(f)) = -p\circ Q_{(0)}^2(f) = 0, \quad \forall f \in \cO.
$$ 
We choose $\alpha_{(1)} = 0$. Moreover, Equation \eqref{eq:Q_{(1)}f} yields \[Q_{(1)}|_{\cO} =  -\r\circ \p \circ Q^2_{(0)}|_{\cO} - \alpha_{(1)} = 0\] Therefore, we can choose $(\nabla_{\alpha})_{(1)}$, an extension of $Q_{(1)}|_{\cO}$ to $\cA$, to be zero. The condition  $Q_{(0)}^2 = 0$ on $\hat \cA^+$ and  the homotopy retract equation \eqref{retration} imply that the $\cO$-linear part $(L_{\beta})_{(1)}$ given in  Equation \eqref{eq:gen.A.final.form} takes a simplified form upon restriction to $\mathcal V_j$, $j\in \mathbb N$:
$$
(L_\beta)_{(1)}(a) = \underbrace{h \circ rhs(a)}_{=0} - \beta_{(1)}(a)= -\beta_{(1)}(a),\; a\in \cV_{\geq 1}.
$$
Moreover, by Equation \eqref{eq:beta_{N+1}}, the term $\beta_{(1)}$ satisfies \[\delta_{\psi}(\beta_{(1)}(a))=p \circ rhs(a)  =0\quad \text{for all} \quad a \in \cV_{\geq 1}.\] Thus, as an admissible choice, we choose $\beta_{(1)} = 0$ on $\hat\cA^+$, as well as $(L_{\beta})_{(1)} =0$ on $\hat\cA^+$. Therefore, the negative degree $+1$ of $Q$ is \[Q_{(1)} = (\nabla_{\alpha})_{(1)} + (L_\beta)_{(1)} = 0,\quad \text{on} \quad \hat\cA^+.\] We proceed by induction: Assume that we have constructed  $Q_{(i)},\, 1\leq i \leq N$ to be  $\hat\cA^+$-linear. Then from Equation \eqref{eq:Delta.N} evaluated on elements $f\in\cO$ we obtain:
$$
\delta_{\psi}\alpha_{(N+1)}(f) = p\circ\sum\limits_{\substack{i+j=N \\ i,j\geq 0}}Q_{(i)}\circ Q_{(j)}(f) = p\circ Q_{(N)}(Q_{(0)}(f)) = 0.
$$
Therefore, we choose $\alpha_{(N+1)} = 0$ on $\cO$. It follows that $Q_{(N+1)}|_{\cO}=0$ and  $(\nabla_{\alpha})_{(N+1)}$ is chosen to be equal to $0$ on $\cA$. In this case, the linear part $(L_\beta)_{(N+1)}$ satisfies \[\delta_{\psi}(L_{\beta})_{(N+1)}(a)  =-\beta_{(N+1)}\quad \text{on}\;\cV_{\geq 1}\] by Equation \eqref{eq:gen.A.final.form}. Moreover, Equation \eqref{eq:beta_{N+1}} reads as $\beta_{(N+1)}$ is $\delta_{\psi}$-closed on $\cV_{\geq 1}$. Therefore,  we choose $\beta_{(N+1)}$ to be zero. Thus, we obtain an $\hat \cA^+$-linear $Q_{(N+1)}$. 
\end{proof}

\begin{remark}
    We showed that homological computations in the simplified setting are reduced to studying the extension on the arborescent Koszul-Tate resolution $(S(\cTr[\FM]), \delta_{\psi})$ part. Note that no conditions are needed on the (smooth) K\"ahler module $\Omega_{\cO/\mathbb K}$. This extension can be further improved by taking a specific extension $(\nabla_{\alpha = 0})_{(0)}$, as explained in the proof of Theorem \ref{thm:main2} below.
\end{remark}

\begin{proof}[Proof (of Theorem \ref{thm:main2})]
The proof is a prolongation of Lemma \ref{prop:arb.ext.exists}, we further impose the following conditions on the choice of the derivation $(\nabla_\alpha)_{(0)}$:
\begin{enumerate}
    \item $ (\nabla_\alpha)_{(0)}(\FM) \subset \FM \bar \odot \cA^+$.
    \item  $ (\nabla_\alpha)_{(0)} \circ \r\circ \p + \r\circ \p \circ  (\nabla_\alpha)_{(0)}  = 0$.
\end{enumerate}
The first condition can always be satisfied without any restrictions. The second condition can be interpreted as a recursive formula for $(\nabla_\alpha)_{(0)}$ for non-trivial trees (i.e., trees with several leaves): for any tree $a = \r(b)$, where $b \in S^{\geq 2}(\cTr[\FM])$ define  
\begin{equation}
\label{eq:aux.nabla}
    (\nabla_\alpha)_{(0)} (a) = -\r\circ (\nabla_\alpha)_{(0)}(b) 
\end{equation}
so that Item 2 is satisfied.

Let us show that the arborescent extension of Lemma \ref{prop:arb.ext.exists} together with the condition on $(\nabla_\alpha)_{(0)}$ from above satisfies items $i-iii$ of Theorem \ref{thm:main2}.

\begin{itemize}
    \item [$\bullet$] on $\hat \cA^+$:  $Q$ coincides with $\hat Q^+$. This is item $i$ of Theorem \ref{thm:main2}.
    \item[$\bullet$] Let $N\geq -1$. The imposed conditions on $(\nabla_\alpha)_{(0)}$ imply that $(\nabla_\alpha)_{(0)}$ is an arity $0$ operation on $\cTr[\FM]$, that is,  $(\nabla_\alpha)_{(0)}(\cTr[\FM]) \subset \cTr[\FM]\bar \odot \hat\cA^+$. Together with the Definition \ref{def:arb.extension}(2),  this implies that all the components of negative degree $\geq 0$ of $Q$ are of arity $0$ on $\cTr[\FM]$, i.e.,  $Q_{(\geq 0)}(\cTr[\FM]) \subset\cTr[\FM]\bar \odot \cA^+$. In particular,  the linear part $L_\beta$ restricted to $\FM$ is $L_\beta = -\beta$,  since  \begin{align*}
        &\left[\delta_{\psi}, ({\nabla_\alpha})_{(N+1)}\right] +({L_\beta})_{(N+1)}\circ\delta_\psi +\sum\limits_{\substack{i+j=N \\ 0\leq i\leq j}}\left[Q_{(i)}, Q_{(j)}\right] \\ \\
        &= \delta_\psi \circ (\nabla_{\alpha})_{(N+1)} + Q_{(N+1)}\circ \delta_{\psi} + \sum\limits_{\substack{i+j=N \\ 0\leq i\leq j}}\left[Q_{(i)}, Q_{(j)}\right]
    \end{align*}is in the kernel of $\r\circ \p $. Using Equation \eqref{eq:def.Q_N} we simplify further the expression for $Q_{(N+1)}$:
    $$
    Q_{(N+1)}(a) = (\nabla_\alpha)_{(N+1)}(a) - \beta_{(N+1)}(a)= \begin{cases}
        (\nabla_{(0)} - \beta_{(0)})(a), \hbox{if $N=-1$} \\
        -\beta_{(N+1)}, \hbox{if $N\geq 0$}
    \end{cases}, \quad \hbox{for $a\in \FM$.}
    $$
    This, together with $Q_{(-1)} = \delta_{\psi}$, implies item $ii$ of Theorem \ref{thm:main2}.
\item[$\bullet$] In a similar fashion we simplify $Q_{(N+1)}$ upon evaluation on $\cTr^{\geq 2}[\FM]$: for $a\in \cTr^{\geq 2}[\FM]$
\begin{align*}
Q_{(N+1)}(a)& = ((\nabla_\alpha)_{(N+1)}-\r\circ \p \circ (\delta_{\psi}\circ (\nabla_\alpha)_{(N+1)} + Q_{(N+1)}\circ \delta_{\psi} + \sum\limits_{\substack{i+j=N \\ i,j\geq 0}}Q_{(i)}\circ Q_{(j)}) - \beta_{(N+1)})(a)  \\
&=((\nabla_\alpha)_{(N+1)}-\r\circ \p \circ (\delta_{\psi}\circ (\nabla_\alpha)_{(N+1)} + Q_{(N+1)}\circ \r^{-1} ) - \beta_{(N+1)})(a)  \\
&=(-\r\circ \p \circ  Q_{(N+1)}\circ \r^{-1}  - \beta_{(N+1)})(a),\quad\text{by Equation \eqref{eq:aux.nabla}}
\end{align*}
 Here, the transition from the first line to the second one is due to $Q_{(i)}, \, i\geq 0$ mapping $\cTr[\FM]\bar \odot \hat \cA^+$ to $\cTr[\FM]\bar \odot \hat\cA^+$ and the kernel of $\p$ being $\hat\cA^+\oplus\cTr[\FM]\bar \odot \hat\cA^+$. So that $$\p \circ (Q_{(N+1)}\circ \delta_{\psi} + \sum\limits_{\substack{i+j=N \\ i,j\geq 0}}Q_{(i)}\circ Q_{(j)})(a) = \p\circ Q_{(N+1)} \circ \r^{-1}(a).$$
Together with $Q_{(-1)} = \delta_{\psi}$ on $\cTr^{\geq 2}[\FM]$,  we recover item $iii$ of Theorem \ref{thm:main2}. This ends the proof.
\end{itemize}
\end{proof}

\section{Examples}\label{sec:3}
In this section, we present explicit examples that illustrate the main results discussed in \S \ref{sec:2}. 

\subsection{Arborescent extension of an $\cI$-Lie algebroid}

It is natural to begin with an example in which the positive part is a Lie algebroid, see Example  \ref{ex:Lie-alg}. Let $(A, \lb_A, \rho)$ be a Lie algebroid over a manifold $M$, such that the anchor map preserves an ideal $\cI\subset C^{\infty}(M)$, that is, $\rho(\Gamma (A))(\cI)\subset \cI$. Here, $A$ shall be called an \emph{$\cI$-Lie algebroid}. Such a Lie algebroid $A$ corresponds to a positively graded $Q$-variety $(S(\cV_1), \hat Q^+)$ that preserves $\cI$, where $\cV_1 \coloneqq \Gamma(A[1]^*)$. In local coordinates, the homological vector fields  $\hat Q^+$ takes the following form: $$
\hat Q^+ = \rho^i_a\xi^a\frac{\partial}{\partial x^i} - \frac{1}{2}C_{ab}^c\xi^a\xi^b \frac{\partial}{\partial \xi^c},
$$
where $x^i, \xi^a$ are local coordinates on $A[1]^*$, $\rho^i_a$, $C^c_{ab}$ are local coordinates for the anchor map and the bracket $\lb_A$. The summation of the repeated indices is assumed. 

We associate to the ideal $\cI$ a projective resolution $(\FM, d)$ of $\cO/\cI$, upon which we construct an arborescent Koszul–Tate resolution $(S(\cTr[\FM], \delta_\psi)$ of $\cO/\cI$. By Theorem \ref{thm:main2}, it is possible to choose an explicit arborescent extension the Lie algebroid $A$. Moreover, such an extension induces a specific structure on $\FM$,  as explained in the following:
\begin{proposition}
    Let $(\cA^-, Q^-)\!\coloneqq (S(\cTr[\FM]), \delta_\psi)$ be an arborescent Koszul-Tate resolution of $\cO/\cI$, and let $(\hat \cA^+, \hat Q^+)\!\coloneqq (S(\cV_1), \hat Q^+)$ be a positively graded $Q$-variety associated to an $\cI$-Lie algebroid $(A, \lb_A, \rho)$ of rank $k$. \begin{enumerate}
        \item An explicit arborescent extension $(\cA, Q_{\chi})$ of $(\cA^-, \hat \cA^+)$ has its homological calculations restricted to $(\cTr[\FM]_{(\bullet)}, \FM_{(\bullet+j-1)}\odot S^j(\cV_1))$, $1\leq j\leq k$. \item Moreover, it induces an $\mathcal O$-linear multiplication $\star_{(\mathrm{top})}\colon \FM \times\FM \to \FM\odot \Omega^{\mathrm{top}}(A)$ of negative degree \footnote{Here, $\Omega^{\mathrm{top}}(A)$ is a space of sections of the line bundle over the base manifold $M$ of top-forms on $A$.} $k=\mathrm{top}$.
           \end{enumerate}
\end{proposition}
\begin{proof}
The degrees of elements in $S(\cV_1)$ are in $\{0, 1, 2, \ldots, k\}$, since $\cV_1$ is concentrated in degree $+1$. Therefore, the claim on homological computations for the explicit extension $Q_\chi$ follows from Theorem \ref{thm:gen.case.computations} as well as Lemma \ref{prop:arb.ext.exists}. This proves item 1.  In order to prove the second claim, observe that in item 1 the negative degree $k-1$ component of the hook map $\chi$ of $(\cA, Q_\chi)$ induces an $\mathcal O$-linear product $\star_{(k)}\colon \FM \times \FM \to \FM\odot S^{\mathrm{top}}(\cV_1)$ by 
    $$
    a\star_{(k)} b\coloneqq  \chi_{(k-1)}\left( \adjustbox{valign = c} {\scalebox{0.5}{  \begin{forest}
for tree = {grow' = 90}, nice empty nodes,
            [,
      [\scalebox{2}{$a$}, tier =1]
      [\scalebox{2}{$b$}, tier =1] 
         ]
\path[fill=black] (.parent anchor) circle[radius=4pt];
\end{forest}}} \right)
    $$
\end{proof}
The product $\star_{(k)}$ measures the violation of Leibniz rule of a tower of products $\star_{(i)}$, and derivations $(Q_\chi)_{(j)}$ restricted to $S(\FM)$, for all  $i,j\leq k-1$. The latter is  explained in full details in \S \ref{sec:h.mult}, where an example of degree $+1$ product $\star_{(1)}$ is also provided.
\subsection{Vector fields vanishing on a subspace of $\mathbb K^2 $}
Let $\cI$ be an ideal of a ring $\cO = \mathbb K[x_1, x_2]$, so that $\cI$ is generated by a regular sequence $\varphi_1, \varphi_2$. Let $\mathfrak  A_\mathcal{I} = \cI\,\mathrm{Der}(\cO)$ be the Lie-Reinhart algebra generated by $\varphi_1 \frac{\partial}{\partial x_1}, \varphi_2\frac{\partial}{\partial x_1}, \varphi_1 \frac{\partial}{\partial x_2}, \varphi_2\frac{\partial}{\partial x_2}$. 

In this example, we associate a $\mathbb Z$-graded $Q$-variety to the Lie-Rinehart algebra $\mathfrak A_{\cI}$ such that its negative {part} is a Koszul-Tate resolution of $\cO/\cI$, and its positive part is the universal $Q$-variety of $\mathfrak A_{\cI}$. We do the extension in two distinct cases: \begin{enumerate}
    \item First, we provide an extension when the negative part is given by a Koszul complex (which is a Koszul-Tate resolution in this case). \item Second, we give an explicit arborescent extension, with the negative part being an arborescent Koszul-Tate resolution. The general discussion of such extensions is delegated to Appendix \ref{sec:appendix}.
\end{enumerate}

\noindent
 \textbf{The positive graded part}: There is a positively graded $Q$-variety $(\hat \cA^+, \hat{Q}^+)$ associated with $\mathfrak A_\cI$, which is described as follows:

    There is free resolution of $\mathfrak A_{\cI}$ in the following form:
    \begin{equation}\label{eq:res_IDer}
\begin{tikzcd}
  0 \arrow[r,] & P_{-2} \arrow[r, "\ell"] & P_{-1} \arrow[r, "\rho"] & \mathfrak A_\cI
\end{tikzcd}
    \end{equation}
where $P_{-1}=\cO^{\times4}$ and $P_{-2}=\cO^{\times 2}$. The anchor map $\rho$ is defined on the basis $(e^1_1, e^1_2, e^2_1, e^2_2)$ of $P_{-1}$ as follows:
$$
\rho(e^1_1) = \varphi_1 \frac{\partial}{\partial x_1}, \quad  \rho(e^1_2) = \varphi_2 \frac{\partial}{\partial x_1}, \quad\rho(e^2_1) = \varphi_1 \frac{\partial}{\partial x_2}, \quad\rho(e^2_2) = \varphi_2 \frac{\partial}{\partial x_2}.
$$
While the map $\ell$ describes the relations between the generators. It is given on the basis  elements $\{u, v\}$ of $P_{-2}$ by
$$
\ell(u) =\varphi_1 e^1_2 - \varphi_2 e^1_1, \quad  \ell(v) =\varphi_1 e^2_2 - \varphi_2 e^2_1.
$$
Equation \eqref{eq:res_IDer} comes equipped with an explicit Lie $2$-algebroid structure (\cite[Proposition 3.17]{CLRL}}) whose dual is the positively graded $Q$-variety $(S(\cV_1\oplus \cV_2),\, \hat Q^+)$ over $\cO$, where $\cV_1\coloneqq P^*_{-1}$, $\cV_2 \coloneqq P^*_{-2}$ . For  $\xi^1_1, \xi^1_2,  \xi^2_1, \xi^2_2$ and $\xi^1, \xi^2$, a  basis of $\cV_1\oplus\cV_2$, the differential $\hat Q^+$ is given by:
\begin{align*}
   \hat Q^+(x_1) &= \xi^1_1\varphi_1 + \xi^1_2\varphi_2,\\
    \hat Q^+(x_2) &= \xi^2_1\varphi_1 + \xi^2_2\varphi_2, \\
    \hat Q^+(\xi^1_1) &= -\xi^1\varphi_2 - \xi^1_1\xi^1_2\frac{\partial\varphi_2}{\partial x_1} +\xi^1_1 \xi^2_1 \frac{\partial\varphi_1}{\partial x_2}+ \xi^1_2 \xi^2_1 \frac{\partial\varphi_2}{\partial x_2},\\
    \hat Q^+(\xi^1_2) &= \xi^1\varphi_1 + \xi^1_1\xi^1_2\frac{\partial\varphi_1}{\partial x_1} +\xi^1_1 \xi^2_2 \frac{\partial\varphi_1}{\partial x_2}+ \xi^1_2 \xi^2_2 \frac{\partial\varphi_2}{\partial x_2},\\ 
    \hat Q^+(\xi^2_1) &= -\xi^2\varphi_2 - \xi^1_1\xi^2_1\frac{\partial\varphi_1}{\partial x_1} -\xi^1_1 \xi^2_2 \frac{\partial\varphi_2}{\partial x_1}- \xi^2_1 \xi^2_2 \frac{\partial\varphi_2}{\partial x_2},\\
    \hat Q^+(\xi^2_2) &= +\xi^2\varphi_1 - \xi^1_2\xi^2_1\frac{\partial\varphi_1}{\partial x_1} +\xi^2_1 \xi^2_2 \frac{\partial\varphi_1}{\partial x_2}- \xi^1_2 \xi^2_2 \frac{\partial\varphi_2}{\partial x_1},\\
    \hat Q^+(\xi^1) &= -\xi^2_1\xi^1 \frac{\partial\varphi_1}{\partial x_2} + \xi^2_2\xi^1 \frac{\partial\varphi_2}{\partial x_2} + \xi^1_1\xi^2 \frac{\partial\varphi_1}{\partial x_2} + \xi^1_2\xi^2 \frac{\partial\varphi_2}{\partial x_2} - \xi^1_1\xi^1_2\xi^2_2\frac{\partial^2\varphi_2}{\partial x_1 \partial x_2} -  \xi^1_1\xi^1_2\xi^2_1\frac{\partial^2\varphi_1}{\partial x_1 \partial x_2}  \\&+\xi^1_1\xi^2_1\xi^2_2\frac{\partial^2\varphi_1}{\partial x_2 \partial x_2} + \xi^1_2\xi^2_1\xi^2_2\frac{\partial^2\varphi_2}{\partial x_2 \partial x_2},\\
    \hat Q^+(\xi^2) &= -\xi^1_1\xi^2 \frac{\partial\varphi_1}{\partial x_1} - \xi^1_2\xi^2 \frac{\partial\varphi_2}{\partial x_1} + \xi^2_1\xi^1 \frac{\partial\varphi_1}{\partial x_1} - \xi^2_2\xi^1 \frac{\partial\varphi_2}{\partial x_1} + \xi^1_1\xi^1_2\xi^2_1\frac{\partial^2\varphi_1}{\partial x_1 \partial x_1} -  \xi^1_1\xi^2_1\xi^2_2\frac{\partial^2\varphi_1}{\partial x_1 \partial x_2}  \\&+\xi^1_1\xi^1_2\xi^2_2\frac{\partial^2\varphi_2}{\partial x_1 \partial x_1} - \xi^1_2\xi^2_1\xi^2_2\frac{\partial^2\varphi_2}{\partial x_1 \partial x_2}.
\end{align*}

\noindent
    \textbf{The negative graded part}:
    
    \textbf{Case 1: Koszul resolution.} Since the ideal $\cI=(\varphi_1, \varphi_2)$ is generated by a regular sequence, the Koszul complex $(S(\cV_{-1}), \delta)$, where $\cV_{-1} = \cO^{\times2}$, and the differential $\delta$ is defined on the basis of $\cV_{-1}$ as follows:
    $$
\delta \pi_1 = \varphi_1, \quad \delta \pi_2 = \varphi_2.
    $$ is a Koszul-Tate resolution of $\cO/\cI$. Therefore, it can be used as the negative part $(\cA^-, Q^-)$ of  a $\mathbb Z$-graded $Q$-variety extension $(\cA, Q)$ of $(\hat \cA^+, \hat{Q}^+)$. 
    
    \textbf{Case 2: Arborescent Koszul-Tate resolution.} As an alternative to the Koszul complex, we can use an arborescent Koszul-Tate resolution as the negative part $(\cA^-, Q^-)$. It is build on top of the Koszul complex, where we "forget" the multiplicative structure. In more details, we start with the Koszul complex $(\FM, d)$:
    $$
    \begin{tikzcd}
  0 \arrow[r,] & \FM_{-2} \arrow[r, "d"] & \FM_{-1} \arrow[r, "d"] & \cO,
\end{tikzcd}
    $$
    where $\FM_{-1}=\cV_{-1}$, $\FM_{-2}= \FM_{-1}\odot\FM_{-1}\simeq \cO$. The differential $d$ is given by
    $$
    d(\pi_1) = \varphi_1, \quad d(\pi_2) = \varphi_2, \quad d(\pi) = \varphi_1\pi_2 - \varphi_2\pi_1,
    $$
    where $\pi_1, \pi_2$ is the basis of $\FM_{-1}$, and $\pi$ is the basis element of $\FM_{-2}$. An arborescent Koszul-Tate resolution is completely specified by a hook map $\psi$, which is even unique in this example, and is given by
$$\psi\left( \adjustbox{valign = c} {\scalebox{0.5}{  \begin{forest}
for tree = {grow' = 90}, nice empty nodes,
            [,
      [\scalebox{2}{$\pi_1$}, tier =1]
      [\scalebox{2}{$\pi_2$}, tier =1] 
         ]
\path[fill=black] (.parent anchor) circle[radius=4pt];
\end{forest}}} \right) = \pi.
$$
The differential $\delta_{\psi}$ is illustrated on a tree with three leaves as follows:
 \begin{equation}
 \delta_{\psi} \left( \adjustbox{valign = c} {\scalebox{0.5}{  \begin{forest}
for tree = {grow' = 90}, nice empty nodes,
[ ,
            [,
      [\scalebox{2}{$\pi_1$}, tier =1]
      [\scalebox{2}{$\pi_2$}, tier =1] 
         ]
         [\scalebox{2}{$\pi$}, tier = 1 ]
 ]
\path[fill=black] (.parent anchor) circle[radius=4pt]
                (!1.child anchor) circle[radius=4pt];
\end{forest}}}\right) = 
  \adjustbox{valign = c} {\scalebox{0.5}{  \begin{forest}
for tree = {grow' = 90}, nice empty nodes,
            [
      [\scalebox{2}{$\pi_1$}, tier =1]
      [\scalebox{2}{$\pi_2$}, tier =1] 
         ]
\path[fill=black] (.parent anchor) circle[radius=4pt];
\end{forest}}} \odot
  \adjustbox{valign = c} {\scalebox{0.5}{  \begin{forest}
for tree = {grow' = 90}, nice empty nodes,
[
         [\scalebox{2}{$\pi$}, tier = 1 ]
 ]
\end{forest}}} 
- 
  \adjustbox{valign = c} {\scalebox{0.5}{  \begin{forest}
for tree = {grow' = 90}, nice empty nodes,
[ ,
      [\scalebox{2}{$\pi_1$}, tier =1]
      [\scalebox{2}{$\pi_2$}, tier =1] 
         [\scalebox{2}{$\pi$}, tier = 1 ]
 ]
\path[fill=black] (.parent anchor) circle[radius=4pt];
\end{forest}}}
  +\varphi_1 \adjustbox{valign = c} {\scalebox{0.5}{  \begin{forest}
for tree = {grow' = 90}, nice empty nodes,
[ ,
            [,
      [\scalebox{2}{$\pi_1$}, tier =1]
      [\scalebox{2}{$\pi_2$}, tier =1] 
         ]
         [\scalebox{2}{$\pi_2$}, tier = 1 ]
 ]
\path[fill=black] (.parent anchor) circle[radius=4pt]
                (!1.child anchor) circle[radius=4pt];
\end{forest}}}
- 
  \varphi_2\adjustbox{valign = c} {\scalebox{0.5}{  \begin{forest}
for tree = {grow' = 90}, nice empty nodes,
[ ,
            [,
      [\scalebox{2}{$\pi_1$}, tier =1]
      [\scalebox{2}{$\pi_2$}, tier =1] 
         ]
         [\scalebox{2}{$\pi_1$}, tier = 1 ]
 ]
\path[fill=black] (.parent anchor) circle[radius=4pt]
                (!1.child anchor) circle[radius=4pt];
\end{forest}}}
+ \adjustbox{valign = c} {\scalebox{0.5}{  \begin{forest}
for tree = {grow' = 90}, nice empty nodes,
            [
      [\scalebox{2}{$\pi$}, tier =1]
      [\scalebox{2}{$\pi$}, tier =1] 
         ]
\path[fill=black] (.parent anchor) circle[radius=4pt];
\end{forest}}}
\end{equation}

 \noindent
\textbf{The total $\mathbb Z$-graded description}:

    \textbf{Case 1: Extension with the Koszul resolution.} The total differential $Q$ on $\cA= S(\cV_{-1}\oplus \cV_{1}\oplus \cV_{2})$ is described as follows:
    On $\cV_{i\geq 1}$ it coincides with $\hat Q^+$. On $\cV_{-1}$ it consists of components $Q_{(-1)}, Q_{(0)}, Q_{(1)}$, $Q = Q_{(-1)}+ Q_{(0)}+ Q_{(1)}$. The component $Q_{(-1)}$ is given by $\delta$ which is extended by $0$ on $\hat \cA^+=S( \cV_{1}\oplus \cV_{2})$. The component $Q_{(0)}$ is defined as:
    $$
    Q_{(0)}(\pi_i) =-\pi_1\left(\xi^1_1\frac{\partial\varphi_i}{\partial x_1}+ \xi^2_1\frac{\partial\varphi_i}{\partial x_2}\right) -\pi_2\left(\xi^1_2\frac{\partial\varphi_i}{\partial x_1}+ \xi^2_2\frac{\partial\varphi_i}{\partial x_2}\right),
    $$
for $i=1,2$, while the component $Q_{(1)}$ is described as follows:
\begin{align*}
    Q_{(1)}(\pi_i) = &-\pi_1\pi_2\xi^1\frac{\partial \varphi_i}{\partial x_1} -\pi_1\pi_2\xi^2\frac{\partial \varphi_i}{\partial x_2} -\pi_1\pi_2\xi^1_1\xi^1_2\frac{\partial^2 \varphi_i}{\partial x_1\partial x_1}
-\pi_1\pi_2\xi^1_1\xi^2_2\frac{\partial^2 \varphi_i}{\partial x_1\partial x_2}
+\pi_1\pi_2\xi^1_2\xi^2_1\frac{\partial^2 \varphi_i}{\partial x_1\partial x_2}\\\\
&-\pi_1\pi_2\xi^2_1\xi^2_2\frac{\partial^2 \varphi_i}{\partial x_2\partial x_2}
\end{align*}

    \textbf{Case 2: Extending with an arborescent Koszul-Tate resolution.}
    Here we provide an explicit arborescent extension $(\cA,Q, \alpha=0, \beta)$ of $(\cA^-, \hat\cA^+)$ by applying the explicit formula of Proposition \ref{prop:Q.explicit.form}. This is completely specified by a retraction residue map $\beta$ and a derivation $(\nabla_{\alpha=0})_{(0)}$ restricted to $\FM$. Since the resolution $(\FM,d)$ is a resolution of free modules, we put without any restriction $(\nabla_{\alpha=0})_{(0)} = 0$ on the basis elements $\xi^1_1, \xi^1_2, \xi^2_1, \xi^2_2, \xi^1, \xi^2$. The map $\beta $ contains components $\beta_{(0)}$, $\beta_{(1)}$, $\beta = \beta_{(0)} + \beta_{(1)}$, defined as follows:
   $$
    \beta_{(0)}(\pi_i) =\pi_1\left(\xi^1_1\frac{\partial\varphi_i}{\partial x_1}+ \xi^2_1\frac{\partial\varphi_i}{\partial x_2}\right) +\pi_2\left(\xi^1_2\frac{\partial\varphi_i}{\partial x_1}+ \xi^2_2\frac{\partial\varphi_i}{\partial x_2}\right),
    $$
    $$
    \beta_{(0)}(\pi) = -\pi\left(\xi^1_1\frac{\partial\varphi_1}{\partial x_1}+ \xi^1_2\frac{\partial\varphi_2}{\partial x_1}+ \xi^2_1\frac{\partial\varphi_1}{\partial x_2}+ \xi^2_2\frac{\partial\varphi_2}{\partial x_2}\right)
    $$
for $i=1,2$, while the component $\beta_{(1)}$ is given by:
$$
\beta_{(1)}(\pi_i) = +\pi\xi^1\frac{\partial \varphi_i}{\partial x_1} +\pi\xi^2\frac{\partial \varphi_i}{\partial x_2} +\pi\xi^1_1\xi^1_2\frac{\partial^2 \varphi_i}{\partial x_1\partial x_1}
+\pi\xi^1_1\xi^2_2\frac{\partial^2 \varphi_i}{\partial x_1\partial x_2}
-\pi\xi^1_2\xi^2_1\frac{\partial^2 \varphi_i}{\partial x_1\partial x_2}
+\pi\xi^2_1\xi^2_2\frac{\partial^2 \varphi_i}{\partial x_2\partial x_2}
$$
so that $Q = -\beta$ on the basis elements $\pi_1,\pi_2,\pi$. For degree reasons, the restriction $\beta|_{\cTr^{\geq 2}[\FM]}$ is identically zero. The negative degree components of the homological vector field $Q$ are then $Q_{(-1)},Q_{(0)}, Q_{(1)}$. So that  $Q=Q_{(-1)}+Q_{(0)}+ Q_{(1)}$. The component $Q_{(-1)}$ is given by $\delta_{\psi}$ on $\cA^-$ and $0$ on $\hat \cA^+$. The component $Q_{(0)}$:
$$
Q_{(0)}(a) = \begin{cases}
    \hat Q^+(a), \quad \hbox{if $a\in \hat\cA^+$,} \\
    \hat Q^+(c_1)\pi_1+ \hat Q^+(c_2)\pi_2+ \hat Q^+(c)\pi-\beta_{(0)}(a), \quad \hbox{if $a=c_1\pi_1+ c_2\pi_2 +c\pi $,  for any $c_1,c_2,c \in \cO$},\\
    -\r\circ Q_{(0)}\circ \r^{-1}(a), \quad \hbox{if $a\in \cTr^{\geq2}[\FM]$},
\end{cases}
$$
while the component $Q_{(1)}$ is
$$
Q_{(1)}(a) = \begin{cases}
    -\alpha_{(1)}, \quad \hbox{if $a\in \FM$},\\
    -\r\circ Q_{(1)}\circ \r^{-1}(a), \quad \hbox{if $a\in \cTr^{\geq2}[\FM]$},
    \end{cases}
$$
The recursion for the components $Q_{(0)}, Q_{(1)}$ tantamount to applying $Q_{(0)}$, $Q_{(1)}$ to the decoration of trees with an appropriate sign. For example,
$$
Q_{(0)}\left( \adjustbox{valign = c} {\scalebox{0.5}{  \begin{forest}
for tree = {grow' = 90}, nice empty nodes,
[ ,
            [,
      [\scalebox{2}{$a$}, tier =1]
      [\scalebox{2}{$b$}, tier =1] 
         ]
         [\scalebox{2}{$c$}, tier = 1 ]
 ]
\path[fill=black] (.parent anchor) circle[radius=4pt]
                (!1.child anchor) circle[radius=4pt];
\end{forest}}}\right) = \adjustbox{valign = c} {\scalebox{0.5}{  \begin{forest}
for tree = {grow' = 90}, nice empty nodes,
[ ,
            [,
      [\scalebox{2}{$Q_{(0)}(a)$}, tier =1]
      [\scalebox{2}{$b$}, tier =1] 
         ]
         [\scalebox{2}{$c$}, tier = 1 ]
 ]
\path[fill=black] (.parent anchor) circle[radius=4pt]
                (!1.child anchor) circle[radius=4pt];
\end{forest}}} +(-1)^{|a|}
\adjustbox{valign = c} {\scalebox{0.5}{  \begin{forest}
for tree = {grow' = 90}, nice empty nodes,
[ ,
            [,
      [\scalebox{2}{$a$}, tier =1]
      [\scalebox{2}{$Q_{(0)}(b)$}, tier =1] 
         ]
         [\scalebox{2}{$c$}, tier = 1 ]
 ]
\path[fill=black] (.parent anchor) circle[radius=4pt]
                (!1.child anchor) circle[radius=4pt];
\end{forest}}} + (-1)^{|a|+|b|}
\adjustbox{valign = c} {\scalebox{0.5}{  \begin{forest}
for tree = {grow' = 90}, nice empty nodes,
[ ,
            [,
      [\scalebox{2}{$a$}, tier =1]
      [\scalebox{2}{$b$}, tier =1] 
         ]
         [\scalebox{2}{$Q_{(0)}(c)$}, tier = 1 ]
 ]
\path[fill=black] (.parent anchor) circle[radius=4pt]
                (!1.child anchor) circle[radius=4pt];
\end{forest}}}
$$
for any homogeneous $a,b,c\in \FM$.

\subsection{Vector fields preserving the ideal of quadratic functions}
Let $\cO =\mathbb K[x,y]$ and let $\cI = \langle x^2,xy, y^2\rangle$. Let $\mathfrak A_\cI\subset \mathrm{Der}(\cO)$ be the Lie-Rinehart algebra made of derivations $X$ of $\cO$ that preserve $\cI$, i.e., $X[\cI]\subset \cI$. The Lie-Rinehart algebra  $\mathfrak A_\cI$  is spanned by
    $\langle x\frac{\partial}{\partial x}, y\frac{\partial}{\partial x}, x\frac{\partial}{\partial y}, y\frac{\partial}{\partial y} \rangle$. We compute  a $\mathbb Z$-graded $Q$-variety that is associated with $\mathfrak A_{\cI}$ as follows:\\

\noindent
\textbf{The positive graded part}. There is a positively graded $Q$-variety $(\hat \cA^+, \hat{Q}^+)$ associated with $\mathfrak A_\cI$, which is constructed out of a free resolution of $\mathfrak A_\cI$. It is  described  as follows:

 A free $\cO$-module resolution of $\mathfrak A_\cI$ can be chosen as follows:

\begin{equation}\label{eq:ex:preserveI}
    \begin{tikzcd}
  0 \arrow[r,] & P_{-2} \arrow[r, "\ell"] & P_{-1} \arrow[r, "\rho"] & \mathfrak A_\cI
\end{tikzcd}
\end{equation}
with $P_{-1} \simeq \cO^4$, $P_{-2} \simeq \cO^2$. The anchor map $\rho$ is given on the basis of $P_{-1}$ as follows:
$$
\rho(e_1) = x\frac{\partial}{\partial x}, \quad
\rho(e_2) = y\frac{\partial}{\partial x}, \quad
\rho(e_3) = x\frac{\partial}{\partial y}, \quad \rho(e_4) = y\frac{\partial}{\partial y}.
$$
The map $\ell$ encodes the two relations of generators of $\cI$:
$$
\ell(u) =xe_2 - ye_1, \quad \ell(v) = xe_4 - ye_3.
$$
where $u,v$ are basis elements of $P_{-2}$. Write $\xi^{1}, \xi^{2}, \xi^{3}$ for the basis of $\cV_{1}:=P_{-1}^*$ and $\eta^1, \eta^2$ for the basis of $\cV_2:=P_{-2}^*$. Then there is a  Lie $2$-algebroid structure on \eqref{eq:ex:preserveI} \cite{CLRL}. The latter corresponds to a positively graded $Q$-variety   $(S(\cV_1\oplus \cV_2), \hat Q^+)$ over $\cO$ which we  describe explicitly  as follows:
\begin{equation*}
    \begin{split}
        \hat Q^+ =\enspace &\xi^1 x\frac{\partial}{\partial x}  +
        \xi^2 y\frac{\partial}{\partial x}+ \xi^3 x\frac{\partial}{\partial y} + \xi^4 y\frac{\partial}{\partial y} + \eta^1  (x \frac{\partial}{\partial \xi^2} - y \frac{\partial}{\partial \xi^1}) \\
        &+\eta^2  (x \frac{\partial}{\partial \xi^4} - y \frac{\partial}{\partial \xi^3})
        +(\xi^2\xi^3) \frac{\partial}{\partial \xi^1} +(\xi^1 \xi^2 + \xi^2\xi^4) \frac{\partial}{\partial \xi^2} -(\xi^1 \xi^3 + \xi^3\xi^4) \frac{\partial}{\partial \xi^3}\\
& - (\xi^2\xi^3) \frac{\partial}{\partial \xi^4} 
        +(\xi^2 \eta^2- \xi^4\eta^1) \frac{\partial}{\partial \eta^1} + (-\xi^1 \eta^2 + \xi^3\eta^1) \frac{\partial}{\partial \eta^2}.
    \end{split}
\end{equation*}

\noindent
 \textbf{The negative graded  part}: Here $(\cA^-, Q^- )$ is the arborescent Koszul-Tate resolution, which is computed out of a free resolution $(\mathfrak M, d)$ of $\cO/\cI$:
$$
\begin{tikzcd}
  0 \arrow[r,] & \mathfrak M_{-2} \arrow[r, "d"] & \FM_{-1} \arrow[r, "d"] & \cO \arrow[r] & 0.
\end{tikzcd}
$$
Here $\FM_{-1} = \cO^{\times 3}$, $\FM_{-2} = \cO^{\times 2}$. The differential $d$ is defined on the basis $\lbrace \pi_1, \pi_2, \pi_3\rbrace$ of $\FM_{-1}$ and $\lbrace \pi, \bar \pi\rbrace$ of $\FM_{-2}$ as 
$$
d(\pi_1) = x^2, \quad d(\pi_2) = xy, \quad d(\pi_3) = y^2,
$$
$$
d(\pi) = x  \pi_2 - y \pi_1, \quad d(\bar \pi) = x\pi_3 - y\pi_2.
$$
The arborescent Koszul-Tate resolution $(S(\cTr[\FM]),\, \delta_{\psi})$ is specified by a choice of the hook map $\psi$, which takes the following form:
$$\psi\left( \adjustbox{valign = c} {\scalebox{0.5}{  \begin{forest}
for tree = {grow' = 90}, nice empty nodes,
            [,
      [\scalebox{2}{$\pi_1$}, tier =1]
      [\scalebox{2}{$\pi_2$}, tier =1] 
         ]
\path[fill=black] (.parent anchor) circle[radius=4pt];
\end{forest}}} \right) = x\pi, \quad 
\psi\left( \adjustbox{valign = c} {\scalebox{0.5}{  \begin{forest}
for tree = {grow' = 90}, nice empty nodes,
            [,
      [\scalebox{2}{$\pi_2$}, tier =1]
      [\scalebox{2}{$\pi_3$}, tier =1] 
         ]
\path[fill=black] (.parent anchor) circle[radius=4pt];
\end{forest}}} \right) = y\bar \pi,\quad \psi\left( \adjustbox{valign = c} {\scalebox{0.5}{  \begin{forest}
for tree = {grow' = 90}, nice empty nodes,
            [,
      [\scalebox{2}{$\pi_1$}, tier =1]
      [\scalebox{2}{$\pi_3$}, tier =1] 
         ]
\path[fill=black] (.parent anchor) circle[radius=4pt];
\end{forest}}} \right) = y\pi+ x\bar \pi, \, \hbox{ and 0 in all other cases}.
$$

\noindent
\textbf{The total $\mathbb Z$-graded description}: By  the formula of Proposition \ref{prop:Q.explicit.form}, an explicit arborescent extension $(\cA, Q_
\chi)$  is determined by a retraction residue $\beta$ and a derivation $(\nabla_{\alpha = 0})_{(0)}$. Since $\FM$ consists of free modules, we choose $(\nabla_{\alpha=0}) =0$ on the basis elements $\pi_1, \pi_2, \pi_3, \pi, \bar\pi$. For degree reasons, the hook map $\chi_{(\geq 0)} = \beta_{(\geq 0)}|_{\cTr^{\geq 2}[\FM]}$ is zero (the length of $\FM$ should be  at least 3 for such $\chi_{(i)}$ to appear). Therefore, we only need to describe $\beta$ on $\FM$, whose non-zero components are only $\beta_{(0)}, \beta_{(1)}$. On basis elements of $\FM$, we have $(Q_\chi)_{(0)} = -\beta_{(0)}$ and it is given by:
\begin{align*}
    (Q_\chi)_{(0)}(\pi_1) &= 2\xi^1\pi_1 + 2\xi^2\pi_2,\\  (Q_\chi)_{(0)}(\pi_2) &=\xi^1\pi_2 + \xi^2\pi_3 + \xi^3\pi_1+\xi^4\pi_4, \\ (Q_\chi)_{(0)}(\pi_3) &= 2\xi^3\pi_2 +2\xi^4\pi_3, \\
    (Q_\chi)_{(0)}(\pi) &= 2\xi^1\pi + \xi^2\bar \pi + \xi^4\pi\quad \text{and}\quad (Q_\chi)_{(0)}(\bar \pi) = \xi^1\bar \pi + \xi^3 \pi + 2\xi^4\bar \pi,
\end{align*}
while $(Q_\chi)_{(1)} = -\beta_{(1)}$ on $\FM$ is given by
$$
(Q_\chi)_{(1)}(\pi_1) = -2\eta^1\pi, \quad (Q_\chi)_{(1)}(\pi_2) = -\eta^1\bar \pi - \eta^2\pi, \quad (Q_\chi)_{(1)}(\pi_3) = -2\eta^2\bar \pi,
$$
with a vanishing contribution on $\FM_{-2}$ for degree reasons. Let us illustrate how the data above is encoded in the total differential $Q_\chi$:
$$
 Q_\chi\left(
  \adjustbox{valign = c} {\scalebox{0.5}{  \begin{forest}
for tree = {grow' = 90}, nice empty nodes,
[ ,
            [,
      [\scalebox{2}{$\pi_1$}, tier =1]
      [\scalebox{2}{$\pi_3$}, tier =1] 
         ]
         [\scalebox{2}{$\pi$}, tier = 1 ]
 ]
\path[fill=black] (.parent anchor) circle[radius=4pt]
                (!1.child anchor) circle[radius=4pt];
\end{forest}}} \right) = (A) + (B) + (C) +(D) + (E), 
$$
where
\begin{enumerate}
    \item[(A)] is of the form
    $$
  \adjustbox{valign = c} {\scalebox{0.5}{  \begin{forest}
for tree = {grow' = 90}, nice empty nodes,
            [
      [\scalebox{2}{$\pi_1$}, tier =1]
      [\scalebox{2}{$\pi_3$}, tier =1] 
         ]
\path[fill=black] (.parent anchor) circle[radius=4pt];
\end{forest}}}\,\, \pi
- 
  \adjustbox{valign = c} {\scalebox{0.5}{  \begin{forest}
for tree = {grow' = 90}, nice empty nodes,
[ ,
      [\scalebox{2}{$\pi_1$}, tier =1]
      [\scalebox{2}{$\pi_3$}, tier =1] 
         [\scalebox{2}{$\pi$}, tier = 1 ]
 ]
\path[fill=black] (.parent anchor) circle[radius=4pt];
\end{forest}}}
$$
 and is obtained from the tree operations $\r^{-1}$ and $\partial$. \item[(B)] is the contribution from $d$ acting by derivation on leaves of the tree:
    $$  
  x\adjustbox{valign = c} {\scalebox{0.5}{  \begin{forest}
for tree = {grow' = 90}, nice empty nodes,
[ ,
            [,
      [\scalebox{2}{$\pi_1$}, tier =1]
      [\scalebox{2}{$\pi_3$}, tier =1] 
         ]
         [\scalebox{2}{$\pi_2$}, tier = 1 ]
 ]
\path[fill=black] (.parent anchor) circle[radius=4pt]
                (!1.child anchor) circle[radius=4pt];
\end{forest}}}
- 
  y\adjustbox{valign = c} {\scalebox{0.5}{  \begin{forest}
for tree = {grow' = 90}, nice empty nodes,
[ ,
            [,
      [\scalebox{2}{$\pi_1$}, tier =1]
      [\scalebox{2}{$\pi_3$}, tier =1] 
         ]
         [\scalebox{2}{$\pi_1$}, tier = 1 ]
 ]
\path[fill=black] (.parent anchor) circle[radius=4pt]
                (!1.child anchor) circle[radius=4pt];
\end{forest}}}
    $$
    Note that terms where $d$ acts on leaves $\pi_1$ or $\pi_3$ are zero. This follows from the solution of the recursion in Proposition \ref{prop:KT.explicit.form}. 
    \item[(C)] is the action of the hook map  $\psi$: 
    $$ y
    \adjustbox{valign = c} {\scalebox{0.5}{  \begin{forest}
for tree = {grow' = 90}, nice empty nodes,
            [
      [\scalebox{2}{$\pi$}, tier =1]
      [\scalebox{2}{$\pi$}, tier =1] 
         ]
\path[fill=black] (.parent anchor) circle[radius=4pt];
\end{forest}}} 
+
x \adjustbox{valign = c} {\scalebox{0.5}{  \begin{forest}
for tree = {grow' = 90}, nice empty nodes,
            [
      [\scalebox{2}{$\bar \pi$}, tier =1]
      [\scalebox{2}{$\pi$}, tier =1] 
         ]
\path[fill=black] (.parent anchor) circle[radius=4pt];
\end{forest}}} 
    $$
\item[(D)] is the term obtained from the action of  $(Q_\chi)_{(0)}$ on leaves:
$$
\adjustbox{valign = c} {\scalebox{0.45}{  \begin{forest}
for tree = {grow' = 90}, nice empty nodes,
[ ,
            [,
      [\scalebox{2}{$2\xi^1\pi_1 $}, tier =1]
      [\scalebox{2}{$\pi_3$}, tier =1] 
         ]
         [\scalebox{2}{$\pi$}, tier = 1 ]
 ]
\path[fill=black] (.parent anchor) circle[radius=4pt]
                (!1.child anchor) circle[radius=4pt];
\end{forest}}}+ 
\adjustbox{valign = c} {\scalebox{0.45}{  \begin{forest}
for tree = {grow' = 90}, nice empty nodes,
[ ,
            [,
      [\scalebox{2}{$ 2\xi^2\pi_2$}, tier =1]
      [\scalebox{2}{$\pi_3$}, tier =1] 
         ]
         [\scalebox{2}{$\pi$}, tier = 1 ]
 ]
\path[fill=black] (.parent anchor) circle[radius=4pt]
                (!1.child anchor) circle[radius=4pt];
\end{forest}}} -
\adjustbox{valign = c} {\scalebox{0.45}{  \begin{forest}
for tree = {grow' = 90}, nice empty nodes,
[ ,
            [,
      [\scalebox{2}{$\pi_1$}, tier =1]
      [\scalebox{2}{$2\xi^3\pi_2$}, tier =1] 
         ]
         [\scalebox{2}{$\pi$}, tier = 1 ]
 ]
\path[fill=black] (.parent anchor) circle[radius=4pt]
                (!1.child anchor) circle[radius=4pt];
\end{forest}}}-
\adjustbox{valign = c} {\scalebox{0.45}{  \begin{forest}
for tree = {grow' = 90}, nice empty nodes,
[ ,
            [,
      [\scalebox{2}{$\pi_1$}, tier =1]
      [\scalebox{2}{$2\xi^4\pi_3$}, tier =1] 
         ]
         [\scalebox{2}{$\pi$}, tier = 1 ]
 ]
\path[fill=black] (.parent anchor) circle[radius=4pt]
                (!1.child anchor) circle[radius=4pt];
\end{forest}}}+
\adjustbox{valign = c} {{\scalebox{0.45}{  \begin{forest}
for tree = {grow' = 90}, nice empty nodes,
[ ,
            [,
      [\scalebox{2}{$\pi_1$}, tier =1]
      [\scalebox{2}{$\pi_3$}, tier =1] 
         ]
         [\scalebox{2}{$(2\xi^1+\xi^4)\pi$}, tier = 1 ]
 ]
\path[fill=black] (.parent anchor) circle[radius=4pt]
                (!1.child anchor) circle[radius=4pt];
\end{forest}}}} +
\adjustbox{valign = c} {\scalebox{0.45}{  \begin{forest}
for tree = {grow' = 90}, nice empty nodes,
[ ,
            [,
      [\scalebox{2}{$\pi_1$}, tier =1]
      [\scalebox{2}{$\pi_3$}, tier =1] 
         ]
         [\scalebox{2}{$\xi^2\bar \pi$}, tier = 1 ]
 ]
\path[fill=black] (.parent anchor) circle[radius=4pt]
                (!1.child anchor) circle[radius=4pt];
\end{forest}}} = 
$$

$$
2\xi^1\adjustbox{valign = c} {\scalebox{0.45}{  \begin{forest}
for tree = {grow' = 90}, nice empty nodes,
[ ,
            [,
      [\scalebox{2}{$\pi_1 $}, tier =1]
      [\scalebox{2}{$\pi_3$}, tier =1] 
         ]
         [\scalebox{2}{$\pi$}, tier = 1 ]
 ]
\path[fill=black] (.parent anchor) circle[radius=4pt]
                (!1.child anchor) circle[radius=4pt];
\end{forest}}}+ 2\xi^2
\adjustbox{valign = c} {\scalebox{0.45}{  \begin{forest}
for tree = {grow' = 90}, nice empty nodes,
[ ,
            [,
      [\scalebox{2}{$ \pi_2$}, tier =1]
      [\scalebox{2}{$\pi_3$}, tier =1] 
         ]
         [\scalebox{2}{$\pi$}, tier = 1 ]
 ]
\path[fill=black] (.parent anchor) circle[radius=4pt]
                (!1.child anchor) circle[radius=4pt];
\end{forest}}} +2\xi^3
\adjustbox{valign = c} {\scalebox{0.45}{  \begin{forest}
for tree = {grow' = 90}, nice empty nodes,
[ ,
            [,
      [\scalebox{2}{$\pi_1$}, tier =1]
      [\scalebox{2}{$\pi_2$}, tier =1] 
         ]
         [\scalebox{2}{$\pi$}, tier = 1 ]
 ]
\path[fill=black] (.parent anchor) circle[radius=4pt]
                (!1.child anchor) circle[radius=4pt];
\end{forest}}} +2\xi^4
\adjustbox{valign = c} {\scalebox{0.45}{  \begin{forest}
for tree = {grow' = 90}, nice empty nodes,
[ ,
            [,
      [\scalebox{2}{$\pi_1$}, tier =1]
      [\scalebox{2}{$\pi_3$}, tier =1] 
         ]
         [\scalebox{2}{$\pi$}, tier = 1 ]
 ]
\path[fill=black] (.parent anchor) circle[radius=4pt]
                (!1.child anchor) circle[radius=4pt];
\end{forest}}}+ (2\xi^1+\xi^4)
\adjustbox{valign = c} {{\scalebox{0.45}{  \begin{forest}
for tree = {grow' = 90}, nice empty nodes,
[ ,
            [,
      [\scalebox{2}{$\pi_1$}, tier =1]
      [\scalebox{2}{$\pi_3$}, tier =1] 
         ]
         [\scalebox{2}{$\pi$}, tier = 1 ]
 ]
\path[fill=black] (.parent anchor) circle[radius=4pt]
                (!1.child anchor) circle[radius=4pt];
\end{forest}}}} + \xi^2
\adjustbox{valign = c} {\scalebox{0.45}{  \begin{forest}
for tree = {grow' = 90}, nice empty nodes,
[ ,
            [,
      [\scalebox{2}{$\pi_1$}, tier =1]
      [\scalebox{2}{$\pi_3$}, tier =1] 
         ]
         [\scalebox{2}{$\bar \pi$}, tier = 1 ]
 ]
\path[fill=black] (.parent anchor) circle[radius=4pt]
                (!1.child anchor) circle[radius=4pt];
\end{forest}}}.
$$
\item[(E)] is the  term obtained from the action of $(Q_\chi)_{(1)}$ on leaves:
$$
\adjustbox{valign = c} {\scalebox{0.45}{  \begin{forest}
for tree = {grow' = 90}, nice empty nodes,
[ ,
            [,
      [\scalebox{2}{$-2\eta^1\pi $}, tier =1]
      [\scalebox{2}{$\pi_3$}, tier =1] 
         ]
         [\scalebox{2}{$\pi$}, tier = 1 ]
 ]
\path[fill=black] (.parent anchor) circle[radius=4pt]
                (!1.child anchor) circle[radius=4pt];
\end{forest}}}- 
\adjustbox{valign = c} {\scalebox{0.45}{  \begin{forest}
for tree = {grow' = 90}, nice empty nodes,
[ ,
            [,
      [\scalebox{2}{$\pi_1$}, tier =1]
      [\scalebox{2}{$-2\eta^2\bar \pi$}, tier =1] 
         ]
         [\scalebox{2}{$\pi$}, tier = 1 ]
 ]
\path[fill=black] (.parent anchor) circle[radius=4pt]
                (!1.child anchor) circle[radius=4pt];
\end{forest}}} = -2\eta^1
\adjustbox{valign = c} {\scalebox{0.45}{  \begin{forest}
for tree = {grow' = 90}, nice empty nodes,
[ ,
            [,
      [\scalebox{2}{$\pi$}, tier =1]
      [\scalebox{2}{$\pi_3$}, tier =1] 
         ]
         [\scalebox{2}{$\pi$}, tier = 1 ]
 ]
\path[fill=black] (.parent anchor) circle[radius=4pt]
                (!1.child anchor) circle[radius=4pt];
\end{forest}}}+2\eta^2
\adjustbox{valign = c} {\scalebox{0.45}{  \begin{forest}
for tree = {grow' = 90}, nice empty nodes,
[ ,
            [,
      [\scalebox{2}{$\pi_1$}, tier =1]
      [\scalebox{2}{$\bar\pi$}, tier =1] 
         ]
         [\scalebox{2}{$\pi$}, tier = 1 ]
 ]
\path[fill=black] (.parent anchor) circle[radius=4pt]
                (!1.child anchor) circle[radius=4pt];
\end{forest}}}
$$
\end{enumerate}

\subsection{Symmetries of  a Koszul function $\varphi$}
Let $V$ be a vector space of dimension $m$ and $\mathcal O$ be the algebra smooth functions on $V$. A function $\varphi \in \mathcal O$ 
is said to be a \emph{Koszul}, if the following \emph{complex} 
\begin{equation}
    \label{eq:KoszulComplex}
\cdots\xrightarrow{\iota_{\varphi}}\X^3 (V) \xrightarrow{\iota_{\varphi}}\X^2 (V)\xrightarrow{\iota_{\varphi}}\X(V) \xrightarrow{\iota_{\varphi}}\mathcal O \longrightarrow 0\end{equation}
is exact in all degrees, except for degree $0$. Here, $\iota_\varphi$ stands for the contration by $d\varphi$. By virtue of a theorem of Koszul \cite{Eisenbud1995}, see \cite{Matsumura} Theorem 16.5 $(i)$, $\varphi$ is \emph{Koszul} if $\left(\frac{\partial\varphi}{\partial x_1} ,\cdots,\frac{\partial\varphi}{\partial x_m}\right) $ is a regular sequence. 

From now on, we choose $\varphi $ a Koszul function, and consider the Lie-Rinehart algebra \begin{equation}\label{varphi}
\mathfrak{A}_\varphi:=\{X\in\mathfrak{X}(V):X[\varphi]=0\} = {\mathrm{Ker}}( \iota_{\varphi} )\colon \X(V) \xrightarrow{\iota_{\varphi}}\mathcal O 
.\end{equation} 
 The Koszul complex \eqref{eq:KoszulComplex} truncated of its degree $0$ term, namely, \begin{equation}
    \label{eq:KoszulComplex2}
\cdots\xrightarrow{\ell=\iota_{\varphi}}\X^3 (V) \xrightarrow{\ell=\iota_{\varphi}}\X^2 (V)\xrightarrow{\rho=\iota_{\varphi}}\X(V) \end{equation}  gives  a free resolution  
 of $\mathfrak{A}_\varphi$. Notice that the exactness of the Koszul complex implies in particular that $\mathfrak{A}_\varphi$ is generated by the vector fields 
: \begin{equation}\left\lbrace\frac{\partial\varphi}{\partial x_i}\frac{\partial}{\partial x_j}-\frac{\partial\varphi}{\partial x_j}\frac{\partial}{\partial x_i},\mid 1\leq i<j\leq m\right\rbrace.
 \end{equation}

 We associate a $\mathbb Z$-graded $Q$-variety with $\mathfrak A_\varphi$ as follows:

 \noindent
  \textbf{The positive graded part}. There is a explicit  positively graded $Q$-variety $(\hat \cA^+, \hat{Q}^+)$ over $\cO$ associated with $\mathfrak A_\cI$, which is constructed out of the free resolution \eqref{eq:KoszulComplex2} \cite{LLS, CLRL}. Here $\hat\cA^+:= S(\cV)$ with $\cV_{i}:=\Omega^{i+1}(V)$, $i\geq 1$.  The homological vector field $\hat Q^+$ is given as follows \begin{equation*}
         \hat{Q}^+=\sum_{i=1}^{m}\sum_{j=1}^md_{\{i,j\}}\frac{\partial\varphi}{\partial x_j}\frac{\partial}{\partial x_i}+\sum_{I_1,\ldots,I_n}\sum_{i_{1}\in I_{1},\ldots,i_{n}\in I_{n}}\epsilon(i_{1},\ldots,i_{n})\frac{\partial^{n}\varphi}{\partial x_{i_1}\cdots\partial x_{i_n}}d_{{I_{1}\bullet\cdots\bullet I_{n}}}\dfrac{\partial}{\partial d_{I_{1}^{i_{1}}\bullet\cdots\bullet I_{n}^{i_{n}}}}.
     \end{equation*}
Notice that $\hat{Q}_+$ satisfies $\hat{Q}^+(\varphi)=0$.

Let us explain the notations: for every multi-index $J=\left\lbrace j_1,\ldots ,j_n\right\rbrace\subseteq\left\lbrace 1,\ldots,d\right\rbrace$ of length $n\geq 2$, $d_J$ stands for the $n$-form $dx_{j_1}\wedge\cdots\wedge dx_{j_n}$. 
  Also, $I_{1}\bullet\cdots\bullet I_{n}$ is a multi-index obtained by concatenation of $n$ multi-indices $I_{1},\ldots,I_{n}$. 
  For every $i_1\in I_1,\ldots,i_n\in I_n$,\;$\epsilon(i_1,\ldots,i_n)$ is the signature of the permutation which brings $i_1,\ldots,i_n$ to the first $n$ slots of $I_{1}\bullet\cdots\bullet I_{n}$. Last, for $i_k\in I_k$, we define $I_{k}^{i_k}:=I_k\setminus \{i_k\}$.\\

\noindent
    \textbf{The Koszul-Tate part}. $\varphi$ is not a zero divisor in $\mathcal{O}$, thus $\cV_{-1}\stackrel{\delta}{\rightarrow} \mathcal{O}, \; e\mapsto \varphi$, where $\cV_{-1} \cong \cO$ with a basis element $e$, is a Koszul-Tate resolution of the ideal $\langle \varphi\rangle$ generated by $\varphi$. The latter coincides with the arborescent Koszul-Tate resolution.\\

        \noindent
        \textbf{The total $\mathbb Z$-graded description}: 
         We obtain a $\mathbb Z$-graded variety $(\cA, Q)$ over $\cO$ whose sheaf of graded functions is $$\cA= \cV_{-1}\oplus \mathcal{O}\oplus \left(\cV_{-1}\odot \hat\cA^+\right)\oplus \hat \cA^+.$$ 
         The total homological vector field  on $\cA$ is simply
             \begin{equation}
                 \label{eq:simple-total-Q}Q:={\delta} + \hat{Q}^+\end{equation} where $\delta $ is extended  by zero on $\hat \cA^+$, and $\hat{Q}^+$ is extended on the negatively graded generators by zero. We have $Q^2=0$, since $(\hat{Q}^+)^2=0$ and $\hat{Q}^+(\varphi)=0$, also \([\delta, \hat Q^+]=0\).


\appendix

\section{Extensions with Koszul complex vs. arborescent Koszul-Tate resolution}\label{sec:appendix}

As an important particular case, it is interesting to compare two $\mathbb Z$-graded extensions, one with the Koszul complex serving as the Koszul-Tate resolution in the negative degree part, while the other utilizes the Koszul complex as the building block of the corresponding arborescent Koszul-Tate resolution. The comparison between these two resolutions was studied in \cite{hancharuk2024, hancharuk:tel-04692988}. Let us recall first the main ingredients:

Let $\cI\subsetneq\cO$ be an ideal generated by a regular sequence $\varphi_1, \dots, \varphi_n$. 
Let \(\cV_{-1}\) be a free \(\mathcal O\)-module of rank \(n\) with basis 
\(\{e_1,\dots,e_n\}\), and consider the symmetric algebra $ S(\cV_{-1})$.
Define an \(\mathcal O\)-linear derivation 
\[
\delta : S(\cV_{-1}) \longrightarrow S(\cV_{-1})
\]
of degree \(+1\) by 
\[
\delta(e_i) = \varphi_i, \qquad i = 1,\dots,n,
\]
and extend \(\delta\) uniquely to all of \(S(\cV_{-1})\) by Leibniz rule. The differential graded algebra \((S(\cV_{-1}),\delta)\)
is the Koszul complex associated with the sequence
\((\varphi_1,\dots,\varphi_n)\).  
Since the sequence is regular, this complex is exact in all negative degrees, and its 
degree-zero cohomology satisfies
\[
H^0(S(\cV_{-1}),\delta) \cong \mathcal O / \mathcal I.
\]

Therefore, \((S(\cV_{-1}),\delta)\) is a Koszul--Tate resolution of 
\(\mathcal O / \mathcal I\). Alternatively, one can use an arborescent Koszul-Tate resolution $(S(\cTr[\FM]), \delta_\psi)$ built from a free resolution $(\FM, d)$ of $\cO/\cI$. We choose $(\FM, d)$ to be isomorphic to the Koszul complex. Let us denote the isomorphism $\phi\colon (\FM,d) \rightarrow (S(\cV_1),\delta)$. It can be easily checked that $\phi$ is a dg-algebra morphism, when one equips $(\FM, d)$ with a product $\cdot$:
\begin{equation}
\label{eq:mult.koszul}
a\cdot b \coloneqq \phi^{-1}(\phi(a)\phi(b)),\quad \forall\, a,b \in \FM.
\end{equation}
The product $\cdot$ is associative, and is compatible with $d$, i.e., $d(a\cdot b) = d(a) \cdot b + (-1)^{|a|}a\cdot d(b)$ for all homogeneous $a,b \in \FM$. As expected, the hook  map $\psi$ is completely specified by the product $\cdot$, as explained in the following proposition.
\begin{proposition}[\cite{hancharuk:tel-04692988,hancharuk2024}]\label{prop:KoszulHook}
    Let  $S(\cTr[\FM], \delta_{\psi})$ be an arborescent Koszul-Tate resolution of $\cO/\cI$ built from a Koszul complex $ (\FM, d)$. Then, the hook map $\psi$ can be chosen as follows: \begin{equation}
    \label{eq:psi.koszul}
\psi \left (  
  \adjustbox{valign = c} {\scalebox{0.5}{  \begin{forest}
for tree = {grow' = 90}, nice empty nodes,
[ ,
      [\scalebox{2}{$a$}, tier =1]
      [ \scalebox{2}{$\dots$}, edge=dotted, tier =1]
         [\scalebox{2}{$b$}, tier = 1 ]
 ]
\path[fill=black] (.parent anchor) circle[radius=4pt];
\end{forest}}} \right) = a\cdot  \dots \cdot b,
    \end{equation}
and $\psi$ vanishes on trees with $\geq 1$ inner vertices. Here "$\cdots$" is either empty (so the tree is binary), or it is a collection of  leaves decorated by $\FM$.
\end{proposition}
Proposition \ref{prop:KoszulHook} shows that $\psi$ contains no new information but the product $\star$, which is determined by the multiplication in $S(\cV_{-1})$. We explore such similarities further, in particular the task of extending positive graded $Q$-varieties over $\cO$ either with an exact Koszul complex, or with the arborescent Koszul-Tate resolution.
\begin{proposition}\label{prop:Kos-arb}
    Let $(\hat \cA^+, \hat Q^+)$ be a positive graded $Q$-variety over $\cO$ preserving $\cI\subsetneq \cO$. Let $(\FM,d) \cong(S(\cV_{-1}), \delta)$ be the exact Koszul complex resolving $\cO/\cI$, and let $(S(\cV_{-1})\bar \odot\hat \cA^+,  Q^K)$ be a $\mathbb Z$-graded extension of  $(S(\cV_{-1}), \delta)$ and $(\hat\cA^+, \hat Q^+)$. There exists an explicit arborescent extension $(\cA, Q_\chi)$ of $S(\cTr[\FM], \delta_{\psi})$, $(\hat \cA^+, \hat Q^+)$  such that 
    \begin{itemize}
        \item [$\bullet$] The hook map $\chi_{(-1)} = \psi$ is as in \eqref{eq:psi.koszul}. The components of $\chi$ of negative degree $\geq 0$ vanish, i.e. $\chi_{(\geq0)} =0$.
        \item[$\bullet$] $Q_\chi|_{\FM\bar \odot \hat \cA^+\oplus \hat \cA^+}$ is given by $Q^K$.
    \end{itemize}
\end{proposition}
\begin{proof} The proof is obtained through a direct computation. The task is to check that the derivation $Q_\chi$ defined in Proposition \ref{prop:Kos-arb} squares to zero. 
    We show this by induction on the number of inner vertices. For trivial trees, i.e., for all elements $a \in \FM$, we have by construction $Q_\chi(a) = Q^K(a) \in \FM\bar \odot \hat \cA^+$. Then $(Q_\chi)^2(a) = (Q^K)^2 (a) = 0 $. On trees with a root and $N\geq 0$ inner vertices,  $Q_\chi|_{\cTr^{\geq 2}[\FM]}$ reads $Q_\chi = \delta_{\psi} - \r\circ (Q_\chi)_{(\geq 0)} \circ \r^{-1}$, by Theorem \ref{thm:main2}(iii). Therefore, upon evaluation on any element $a \in \cTr^{\geq 2}[\FM]$:
    \begin{align}
    \label{eq:Qsquare.Koszul}
        (Q_{\chi})^2 &= (\delta_{\psi} + (Q_{\chi})_{(\geq 0)})(\delta_{\psi} - \r\circ (Q_{\chi})_{(\geq 0)}\circ \r^{-1}) \\ &=(Q_{\chi})_{(\geq 0)}\circ \delta_{\psi} + \r\circ (Q_{\chi})_{(\geq 0)}^2\circ \r^{-1} -\delta_{\psi}\circ \r\circ (Q_{\chi})_{(\geq 0)}\circ \r^{-1} \nonumber
    \end{align}
    The first summand on the r.h.s. of \eqref{eq:Qsquare.Koszul} can be written as:
    \begin{align*}
    (Q_{\chi})_{(\geq 0)}\circ \delta_{\psi}(a) &= \left( (Q_{\chi})_{(\geq 0)} \circ \r^{-1} - (Q_{\chi})_{(\geq 0)}\circ \r \circ \p \circ \delta_{\psi}\circ \r^{-1} - (Q_{\chi})_{(\geq 0)}\circ \psi \right) (a)  \\
    &= \left((Q_{\chi})_{(\geq 0)} \circ \r^{-1} + \r \circ (Q_{\chi})_{(\geq 0)} \circ \p \circ \delta_{\psi}\circ \r^{-1} - (Q_{\chi})_{(\geq 0)}\circ \psi \right) (a)\\
    &= \left ((Q_{\chi})_{(\geq 0)} \circ \r^{-1} + \r \circ \p\circ (Q_{\chi})_{(\geq 0)}  \circ \delta_{\psi}\circ \r^{-1} - (Q_{\chi})_{(\geq 0)}\circ \psi\right) (a).
    \end{align*}
    The transition from the second line to the third is due to $(Q_\chi)_{(\geq 0)}$ mapping $\cTr[\FM]\to \cTr[\FM]\bar \odot\hat\cA^+$. For this reasons, we can rewrite the second term on the r.h.s of \eqref{eq:Qsquare.Koszul} as
    $$
\r\circ (Q_{\chi})_{(\geq 0)}^2\circ \r^{-1}(a) = \r\circ \p\circ (Q_{\chi})_{(\geq 0)}^2\circ \r^{-1}(a).
    $$
    The third summand in \eqref{eq:Qsquare.Koszul} we rewrite in the following way:
    $$
    -\delta_{\psi}\circ \r\circ (Q_{\chi})_{(\geq 0)}\circ \r^{-1}(a) = \left( - (Q_{\chi})_{(\geq 0)}\circ \r^{-1} + \r \circ \p \circ \delta_{\psi}\circ (Q_{\chi})_{(\geq 0)}\circ \r^{-1} + \psi \circ \r\circ (Q_{\chi})_{(\geq 0)}\circ \r^{-1} \right)(a)
    $$
    All terms combined together can be rewritten as follows:
    \begin{equation}
        \label{eq:Qsquare.Koszul2}
        (Q_{\chi})^2(a) = \left(\r \circ \p \circ (Q_{\chi})^2\circ \r^{-1} + \psi \circ \r \circ (Q_{\chi})_{(\geq 0)}\circ \r^{-1} - (Q_{\chi})_{(\geq 0)}\circ \psi \right)(a).
    \end{equation}
    From Equation \eqref{eq:psi.koszul} it follows that there is a distinguished case when $a$ contains a root and no inner vertices, otherwise the last two summands in \eqref{eq:Qsquare.Koszul2} vanish. Let $a\in \cTr^{\geq 2}[\FM]$ be a tree obtained by rooting elements $b_1,\dots, b_k \in \FM$. Then $\psi(a) = b_1\cdot \ldots \cdot b_k$, which is just $b_1\odot \ldots \odot b_k$ where each $b_i$ is viewed as an element of $S(\cV_{-1})$. Upon this identification, the second and the third term constitute the Leibniz rule for $Q^K$, that is:
    \begin{align*}
    &(\psi \circ \r \circ (Q_{\chi})_{(\geq 0)}\circ \r^{-1} - (Q_{\chi})_{(\geq 0)}\circ \psi )(a) =\\ &\left( \sum_{j=1}^k (-1)^{|b_1| +\cdots +|b_{j-1}|}b_1\odot b_{j-1}\odot Q^K(b_j)\odot b_{j+1}\odot \ldots \odot b_k\right) - Q^k(b_1\odot \ldots \odot b_k) = 0
    \end{align*}
    Now it is clear that for the number of inner vertices $N= 0$, \[(Q_{\chi})^2(a) = \r \circ \p \circ (Q_\chi)^2\circ \r^{-1}(a) = 0.\] By a simple recursion, $(Q_{\chi})^2 = 0$ on trees with an arbitrary number of inner vertices $N$.
\end{proof}
\section{Higher order multiplications on projective resolutions}
\label{sec:h.mult}

 Theorem \ref{thm:main2} gave an explicit arborescent extension $(\cA, Q_\chi)$ of any positively graded $Q$-variety over $\cO$ that preserves an ideal $\cI\subset \cO$.  The latter extension comes equipped with what we called here a hook map $\chi\colon \cTr^{\geq 2}[\FM]\to \FM\bar \odot\hat \cA^+$. By analogy to the hook map  of an arborescent Koszul-Tate resolution of $\cO/\cI$,  the  map $\chi$ induces in particular a $\FM\bar \odot\hat \cA^+$-valued product on $\FM$ which we now describe.

To start with, the negative degree $0$ component $\chi_{(0)}$ of $\chi$ induces an $\cO$-linear multiplication $\star_{(1)}$ of negative degree $+1$ on $\FM$  as follows:
    \normalfont 
    $$
a\star_{(1)} b\coloneqq  \chi_{(0)}\left( \adjustbox{valign = c} {\scalebox{0.5}{  \begin{forest}
for tree = {grow' = 90}, nice empty nodes,
            [,
      [\scalebox{2}{$a$}, tier =1]
      [\scalebox{2}{$b$}, tier =1] 
         ]
\path[fill=black] (.parent anchor) circle[radius=4pt];
\end{forest}}} \right),
    $$
    for all $a,b \in \FM$. Applying $\mathrm{ad_{\delta_{\psi}}}$ to the product, we obtain the following relations, for $a\in \FM_{(i)},\, b \in \FM_{(j)},\,\, i,j \geq 1$:
   \begin{equation}
       \begin{cases}
           \delta_{\psi}(a\star_{(1)}b) - \delta_{\psi} (a) \star_{(1)} b- (-1)^{i}a\star_{(1)}\delta_\psi (b) =- (Q_\chi)_{(0)} (a\star b) +(Q_\chi)_{(0)}(a)\star b +(-1)^{i}a\star (Q_\chi)_{(0)}(b),\enspace \\\hbox{if $i,j \geq 2$,} \\ \\
           \delta_\psi(a\star_{(1)}b) +a\star_{(1)}\delta_\psi (b) =- (Q_\chi)_{(0)} (a\star b) +(Q_\chi)_{(0)}(a)\star b -a\star (Q_\chi)_{(0)}(b) ,\enspace\hbox{if $i =1, j\geq 2$,} \\ \\
           \delta_\psi (a\star_{(1)}b)  =- (Q_\chi)_{(0)} (a\star b) +(Q_\chi)_{(0)}(a)\star b -a\star (Q_\chi)_{(0)}(b) ,\enspace\hbox{if $i,j =1$.} \\
       \end{cases}
   \end{equation} 
   In other words, the product $\star_{(1)}$ measures the violation of the Leibniz rule by $(Q_\chi)_{(0)}$ on $\FM$. 
   
   In general, for $k\geq 2$,  we define a product $\star_{(k)}$:
       $$
a\star_{(k)} b\coloneqq  \chi_{(k-1)}\left( \adjustbox{valign = c} {\scalebox{0.5}{  \begin{forest}
for tree = {grow' = 90}, nice empty nodes,
            [,
      [\scalebox{2}{$a$}, tier =1]
      [\scalebox{2}{$b$}, tier =1] 
         ]
\path[fill=black] (.parent anchor) circle[radius=4pt];
\end{forest}}} \right),
    $$
    that measures the violation of Leibniz rule of a family of products $\{\star_{(i)}\}_{0\leq i\leq k-1}$ and derivations $(Q_{\chi})_{(j)}$, $0\leq j\leq k-1$, with the convention that "$\star_{(0)} := \star$" is the product induced by the hook map $\psi$, see Remark \ref{rem:psi.mult}. In more details,
  \begin{equation}
       \begin{cases}
           \delta_{\psi}(a\star_{(k)}b) - \delta_{\psi} (a) \star_{(k)} b- (-1)^{i}a\star_{(k)}\delta_\psi (b) =\\ \\ \qquad \qquad \sum\limits_{\substack{n+m=k-1 \\ n,m\geq 0}}- (Q_\chi)_{(m)} (a\star_{(n)} b) +(Q_\chi)_{(m)}(a)\star_{(n)} b +(-1)^{i}a\star_{(n)} (Q_\chi)_{(m)}(b),\enspace \hbox{if $i,j \geq 2$,} \\  \\
           \delta_\psi(a\star_{(k)}b) +a\star_{(k)}\delta_\psi (b)= \\ \\ \qquad \qquad \sum\limits_{\substack{n+m=k-1 \\ n,m\geq 0}}- (Q_\chi)_{(m)} (a\star_{(n)} b) +(Q_\chi)_{(m)}(a)\star_{(n)} b -a\star_{(n)} (Q_\chi)_{(m)}(b) ,\enspace\hbox{if $i =1, j\geq 2$,} \\ \\
           \delta_\psi (a\star_{(k)}b)  =\sum\limits_{\substack{n+m=k-1 \\ n,m\geq 0}}-(Q_\chi)_{(m)} (a\star_{(n)} b) +(Q_\chi)_{(m)}(a)\star_{(n)} b -a\star_{(n)} (Q_\chi)_{(m)}(b) ,\enspace\hbox{if $i,j =1$,} \\
       \end{cases}
   \end{equation}
   where $a\in \FM_{(i)},\,\, b\in \FM_{(j)}$, $i,j \geq 1$.

 We discuss an example of such a multiplication below.

\begin{example}
    
Let $\cI=\langle x^2, yz, xz, xy\rangle$ be an ideal of $\cO=\mathbb K[x,y,z]$, and let $\mathfrak A$ be a Lie-Rinehart algebra over $\cO$ generated by a single derivation $x\frac{\partial}{\partial y}$ that preserves $\cI$. \\

\noindent
 \textbf{The positive graded  part}:  The Lie-Rinehart algebra $\mathfrak A\subset \mathrm{Der}(\cO)$ is a Lie algebroid. The corresponding positively graded variety  $ (\hat \cA^+, \hat Q^+)$ takes a particularly simple form, namely $\hat \cA^+=S(\cV_1)$ is generated by a single  element $\xi\in \cV_1$ of degree $+1$, and $\hat Q^+ = \xi x\frac{\partial}{\partial y}$.\\
 
\noindent
 \textbf{The negative graded  part}: We use the arborescent Koszul-Tate resolution, obtained from a free resolution $(\mathfrak M, d)$ of $\cO/\cI$:
$$
\begin{tikzcd}
  0 \arrow[r,] & \mathfrak M_{-3} \arrow[r, "d"] & \mathfrak M_{-2} \arrow[r, "d"] & \FM_{-1} \arrow[r, "d"] & \cO \arrow[r] & 0.
\end{tikzcd}
$$
The ranks of $\FM_i$ are $4,4,1$ in degrees $-1,-2,-3$ respectively. The differential $d$ is defined on the basis of $\FM$ as follows:
\begin{align*}
    d(e_1) = x^2,\quad d(e_2) = yz,&\quad d(e_3) = xz,\quad d(e_4) = xy, \\
    d(e_{13})  =xe_3 - ze_1,\quad d(e_{14}) = xe_4 - ye_1,& \quad d(e_{24}) = ze_4 - xe_2,\quad d(e_{34}) = ze_4 -ye_3, \\
    d(e_{134}) = xe_{34}& - ze_{14}+ ye_{13.}
\end{align*}
Here, $(e_i)_{i=1}^4$ are the basis elements of $\FM_{-1}$; $e_{13}, e_{14}, e_{24}, e_{34}$ are basis elements of $\FM_{-2}$, and $e_{134}$ is the basis element of $\mathfrak M_{-3}$. The resolution $(S(\cTr[\FM]), \delta_{\psi})$ is determined by the hook map $\psi$ that we now describe: the list of non-zero contributions of $\psi$ on the basis of $\FM$ is as follows
\begin{align*} &\psi\left( \adjustbox{valign = c} {\scalebox{0.5}{  \begin{forest}
for tree = {grow' = 90}, nice empty nodes,
            [,
      [\scalebox{2}{$e_1$}, tier =1]
      [\scalebox{2}{$e_2$}, tier =1] 
         ]
\path[fill=black] (.parent anchor) circle[radius=4pt];
\end{forest}}} \right) = -xe_{24}+ ze_{14}, \quad \psi\left( \adjustbox{valign = c} {\scalebox{0.5}{  \begin{forest}
for tree = {grow' = 90}, nice empty nodes,
            [,
      [\scalebox{2}{$e_1$}, tier =1]
      [\scalebox{2}{$e_3$}, tier =1] 
         ]
\path[fill=black] (.parent anchor) circle[radius=4pt];
\end{forest}}} \right) = xe_{13},  \quad \psi\left( \adjustbox{valign = c} {\scalebox{0.5}{  \begin{forest}
for tree = {grow' = 90}, nice empty nodes,
            [,
      [\scalebox{2}{$e_1$}, tier =1]
      [\scalebox{2}{$e_4$}, tier =1] 
         ]
\path[fill=black] (.parent anchor) circle[radius=4pt];
\end{forest}}} \right) = xe_{14},
 \\  \\ &\psi\left( \adjustbox{valign = c} {\scalebox{0.5}{  \begin{forest}
for tree = {grow' = 90}, nice empty nodes,
            [,
      [\scalebox{2}{$e_2$}, tier =1]
      [\scalebox{2}{$e_3$}, tier =1] 
         ]
\path[fill=black] (.parent anchor) circle[radius=4pt];
\end{forest}}} \right) = +ze_{24}- ze_{34}, \quad \psi\left( \adjustbox{valign = c} {\scalebox{0.5}{  \begin{forest}
for tree = {grow' = 90}, nice empty nodes,
            [,
      [\scalebox{2}{$e_2$}, tier =1]
      [\scalebox{2}{$e_4$}, tier =1] 
         ]
\path[fill=black] (.parent anchor) circle[radius=4pt];
\end{forest}}} \right) = ye_{24},  \quad \psi\left( \adjustbox{valign = c} {\scalebox{0.5}{  \begin{forest}
for tree = {grow' = 90}, nice empty nodes,
            [,
      [\scalebox{2}{$e_3$}, tier =1]
      [\scalebox{2}{$e_4$}, tier =1] 
         ]
\path[fill=black] (.parent anchor) circle[radius=4pt];
\end{forest}}} \right) = xe_{34}, \\ \\
&\psi\left( \adjustbox{valign = c} {\scalebox{0.5}{  \begin{forest}
for tree = {grow' = 90}, nice empty nodes,
            [,
      [\scalebox{2}{$e_1$}, tier =1]
      [\scalebox{2}{$e_{34}$}, tier =1] 
         ]
\path[fill=black] (.parent anchor) circle[radius=4pt];
\end{forest}}} \right) = xe_{134}, \quad 
\psi\left( \adjustbox{valign = c} {\scalebox{0.5}{  \begin{forest}
for tree = {grow' = 90}, nice empty nodes,
            [,
      [\scalebox{2}{$e_2$}, tier =1]
      [\scalebox{2}{$e_{13}$}, tier =1] 
         ]
\path[fill=black] (.parent anchor) circle[radius=4pt];
\end{forest}}} \right) = ze_{134}, \quad \psi\left( \adjustbox{valign = c} {\scalebox{0.5}{  \begin{forest}
for tree = {grow' = 90}, nice empty nodes,
            [,
      [\scalebox{2}{$e_3$}, tier =1]
      [\scalebox{2}{$e_{14}$}, tier =1] 
         ]
\path[fill=black] (.parent anchor) circle[radius=4pt];
\end{forest}}} \right) = -xe_{134}, \\ \\  & \psi\left( \adjustbox{valign = c} {\scalebox{0.5}{  \begin{forest}
for tree = {grow' = 90}, nice empty nodes,
            [,
      [\scalebox{2}{$e_4$}, tier =1]
      [\scalebox{2}{$e_{13}$}, tier =1] 
         ]
\path[fill=black] (.parent anchor) circle[radius=4pt];
\end{forest}}} \right) = xe_{134}, \quad 
\psi\left( \adjustbox{valign = c} {\scalebox{0.5}{  \begin{forest}
for tree = {grow' = 90}, nice empty nodes,
            [,
    [\scalebox{2}{$e_1$}, tier =1]
      [\scalebox{2}{$e_2$}, tier =1]
      [\scalebox{2}{$e_3$}, tier =1] 
         ]
\path[fill=black] (.parent anchor) circle[radius=4pt];
\end{forest}}} \right) = - xze_{134}, \quad \psi\left( \adjustbox{valign = c} {\scalebox{0.5}{  \begin{forest}
for tree = {grow' = 90}, nice empty nodes,
            [,
            [\scalebox{2}{$e_1$}, tier =1]
      [\scalebox{2}{$e_3$}, tier =1]
      [\scalebox{2}{$e_4$}, tier =1] 
         ]
\path[fill=black] (.parent anchor) circle[radius=4pt];
\end{forest}}} \right) = x^2e_{134}.
\end{align*}

As mentioned in Remark \ref{rem:psi.mult}, the restriction of $\psi$  to trees of the form $ \adjustbox{valign = c} {\scalebox{0.5}{  \begin{forest}
for tree = {grow' = 90}, nice empty nodes,
            [,
            [\scalebox{2}{$a$}, tier =1]
      [\scalebox{2}{$b$}, tier =1] 
         ]
\path[fill=black] (.parent anchor) circle[radius=4pt];
\end{forest}}}$ induces a multiplication $\star$ compatible with $d$. In this example for degree reasons the multiplication is associative. In addition, the associativity plus Equation \eqref{eq:psi.recursion}  imply that $\psi\left(\adjustbox{valign=c}{\scalebox{0.5}{{\begin{forest}
for tree = {grow' = 90}, nice empty nodes,
            [,
            [\scalebox{2}{$a$}, tier =1]
      [\scalebox{2}{$b$}, tier =1]
      [\scalebox{2}{$c$}, tier =1] 
         ]
\path[fill=black] (.parent anchor) circle[radius=4pt];
\end{forest}}}}\right) =a\star (b\star c) $ for all $a,b,c\in \FM$.

\noindent
\textbf{The total $\mathbb Z$-graded description}: The explicit arborescent extension $(\cA, Q_\chi)$ is characterized by the derivation $(\nabla_{\alpha = 0})_{(0)}$ and a retraction residue  $\beta$. As in previous examples, we choose $(\nabla_{\alpha = 0}) = 0$ to vanish on the basis elements of $\FM$. The list of all non-zero contributions of $Q_\chi$ on the basis of $\FM$ is as follows:
$$
(Q_\chi)_{(0)}(e_2) = \xi  e_3, \quad (Q_\chi)_{(0)}(e_4) = \xi e_1,\quad (Q_\chi)_{(0)}(e_{24}) = (Q_\chi)_{(0)}(e_{34}) = - \xi e_{13},
$$
$$
(Q_\chi)_{(0)}\left(\adjustbox{valign=c}{\scalebox{0.5}{{\begin{forest}
for tree = {grow' = 90}, nice empty nodes,
            [,
           [\scalebox{2}{$e_2$}, tier =1]
      [\scalebox{2}{$e_4$}, tier =1] 
         ]
\path[fill=black] (.parent anchor) circle[radius=4pt];
\end{forest}}}}\right) = \adjustbox{valign=c}{\scalebox{0.5}{{\begin{forest}
for tree = {grow' = 90}, nice empty nodes,
            [,
           [\scalebox{2}{$\xi e_3$}, tier =1]
      [\scalebox{2}{$e_4$}, tier =1] 
         ]
\path[fill=black] (.parent anchor) circle[radius=4pt];
\end{forest}}}} + \adjustbox{valign=c}{\scalebox{0.5}{{\begin{forest}
for tree = {grow' = 90}, nice empty nodes,
            [,
           [\scalebox{2}{$e_2$}, tier =1]
      [\scalebox{2}{$\xi e_1$}, tier =1] 
         ]
\path[fill=black] (.parent anchor) circle[radius=4pt];
\end{forest}}}}  + \xi e_{134}.
$$
where the last summand comes from negative degree zero component $\chi_{(0)}$ of  the hook map $\chi$
$$
\chi_{(0)} \left(\adjustbox{valign=c}{\scalebox{0.5}{{\begin{forest}
for tree = {grow' = 90}, nice empty nodes,
            [,
           [\scalebox{2}{$e_2$}, tier =1]
      [\scalebox{2}{$e_4$}, tier =1] 
         ]
\path[fill=black] (.parent anchor) circle[radius=4pt];
\end{forest}}}}\right) = -\xi e_{134} 
$$
In particular, $\chi_{(0)}$ induces a product on $\FM$, which measures the failure of Leibniz rule between the product $\star$ on $\FM$ and the derivation $\xi x\frac{\partial}{\partial y}$:
$$e_2\star_{(1)} e_4 \coloneqq  \chi_{(0)} \left(\adjustbox{valign=c}{\scalebox{0.5}{{\begin{forest}
for tree = {grow' = 90}, nice empty nodes,
            [,
           [\scalebox{2}{$e_2$}, tier =1]
      [\scalebox{2}{$e_4$}, tier =1] 
         ]
\path[fill=black] (.parent anchor) circle[radius=4pt];
\end{forest}}}}\right), \quad d e_2\star_{(1)} e_4 =  \left( -(Q_\chi)_{(0)}(e_2\star e_4) + (Q_\chi)_{(0)}e_2\star e_4 - e_2\star (Q_\chi)_{(0)}e_4\right).
$$
\end{example}

\section{Lifting derivations}
\label{Lifting.derivations}
We end with some technical assertions that are used in the main results \S \ref{sec:2}. 
\noindent
When studying quotient constructions, a natural question that arises is the existence of the lifting property. For the problem under our scope, we can ask if it is possible to lift a derivation $Q^+$ defined on some $S_{\cO/\cI}(\oplus_{i\geq 1}\mathcal W_i)$ to a derivation $\hat Q^+$ on $S(\oplus_{i\geq 1}\mathcal W_i)$ for a collection of $\cO$-modules $(\mathcal  W_i)_{i\geq 1}$. The answer is positive for certain choices of $\cO$.
    \begin{lemma}
    \label{lem:kahler}
        Let $\eta\colon \cO \twoheadrightarrow \mathcal O'$ be a surjective homomoprhism of $\mathbb K$-algebras. Let $\mathcal S$ be an $\cO$-module, $\mathcal T$ be an $\mathcal O'$-module such that $\varphi\colon \mathcal S\twoheadrightarrow \mathcal T$ as $\cO$-modules. If the K\"ahler module $\Omega_{\cO/\mathbb K}$ is projective\footnote{The condition $\Omega_{\cO/\mathbb K}$ is projective holds, for example, for a polynomial ring in a finite number of variables, where $\Omega_{ \cO/\mathbb K}$ is free. Also, when $\mathcal O$ is algebra of functions of a manifold, see Remark \ref{rem:smooth.kahler} and Convention \ref{convention-smooth}.}, then any derivation $q'\in \mathrm{Der}_{\mathbb K}(\mathcal O', \mathcal T)$ can be lifted to a derivation $q\in \mathrm{Der}_{\mathbb K}(\cO, \mathcal S)$. \end{lemma}
        \begin{proof}
         By the universal property of $\Omega_{\mathcal O'/\mathbb K}$, $q$ is determined uniquely by the composition $h'\circ D'$ for some $h'\in \mathrm{Hom}_{\mathcal O'}(\Omega_{\mathcal O'/\mathbb K}, \mathcal T)$. Moreover, $\Omega_{\cO/\mathbb K}$ projects onto $\Omega_{\mathcal O'/\mathbb K}$ \cite{Hartshorne}. Therefore, there exists $h\in \mathrm{Hom}_{\mathcal O'}(\Omega_{\mathcal O/\mathbb K}, \mathcal T)$ such that the following diagram commutes: 
\begin{equation}
\label{eq:diag.kahler}
\begin{tikzcd}[row sep=large, column sep=large]
\cO\arrow[two heads,d ] \arrow[r, "D"] & \Omega_{\cO/\mathbb K} \arrow[dashed, r, "\exists h"] \arrow[two heads, d] &  \mathcal S \arrow[two heads, d]\\
 \mathcal O' \arrow[r, "D'"] & \Omega_{\mathcal O'/\mathbb K} \arrow[r, "h'"] & \mathcal T
\end{tikzcd}
\end{equation}
         The existence of $h$ follows from $\Omega_{\cO/\mathbb K}$ being projective. Thus,  $q$ is defined as the composition $h\circ D$.
         \end{proof}
         
        \begin{remark}
        \label{rem:smooth.kahler}
         Lemma \ref{lem:kahler} admits generalization to the smooth setting, i.e.,  when $\cO$ is a $C^\infty$-ring. A notable example is the ring $C^{\infty}(M)$ of smooth functions of a manifold $M$. In this setting, the K\"ahler module is replaced by a smooth one \cite{dubuc19841,joyce2019a}. The smooth K\"ahler module is obtained by taking a quotient of the free $\cO$ module $\oplus_{i\in\cO}\cO_i$, parametrized by ring elements, by the relations of \emph{smooth} derivations\footnote{{For $C^\infty(M)$, smooth derivations coincide with the usual derivations of the algebra $C^\infty(M)$. The same is true for $k$-jet determined $C^\infty$-rings \cite{Yamashita01112016}. In general, however, the module of smooth derivations is a proper submodule of the module of algebraic derivations \cite{LermanLouis}. }}. In this modification, the smooth K\"ahler module retains many of the properties of the algebraic one, e.g., the commutative diagram \eqref{eq:diag.kahler} holds \cite{LERMAN2024105062}. In general, the  algebraic K\"ahler module is strictly larger than the smooth one, and the universal derivation $D\colon \cO\to \Omega_{\cO/\mathbb K}$ of the algebraic one might not be smooth \cite{Osborn}, while the smooth K\"ahler module reproduces de Rham 1-forms \cite{joyce2019a, LERMAN2024105062}, so that it is projective and Lemma \ref{lem:kahler} holds. {We expect that the results of the present paper can be extended to the general framework of $C^\infty$-rings, although we do not pursue this direction here.} 
        \end{remark}
        \begin{convention}\label{convention-smooth}
            In this article, whenever we write “smooth K\"ahler module” we refer to the geometric case where $\cO = C^{\infty}(M)$ for some smooth manifold $M$, where the algebraic K\"ahler module is replaced by the module of de Rham $1$-forms. {For an arbitrary commutative unital algebra \(\mathcal O\), \(\Omega_{\mathcal O/\mathbb K}\) is the usual module of K\"ahler differentials.}
        \end{convention}   
    \begin{proposition}
    \label{prop:Qlift}
        Let $\cO$ be a (smooth) $\mathbb K$-algebra with a (smooth) K\"ahler module $\Omega_{\cO/\mathbb K}$ being projective, and let  $\mathcal S = S(\oplus_{i\in \mathbb Z^\times} \mathcal{W}_i)$ for some projective $\cO$-modules $\mathcal W_i$. If $Q$ is a derivation of $\mathcal S\otimes \cO/\cI$, then $Q$  lifts to a derivation $\hat{Q}$ 
        of $\mathcal S$. Moreover, if $Q$ is a differential on $\mathcal S\otimes \cO/\cI$, then $\hat Q^2(\mathcal S)\subseteq \cI\mathcal S$.
    \end{proposition}

\begin{proof}
    By Lemma \ref{lem:kahler} one can lift the derivation $Q|_{\cO/\cI}\colon \mathcal{O}/\mathcal I\to \mathcal S_1\otimes \mathcal{O}/\mathcal{I}$ to a derivation $\hat Q|_{\cO}\colon \mathcal{O}\to \mathcal{S}_1$. Furthermore, $\hat Q|_{\cO}$ can be lifted to a derivation $ \nabla$ on $\mathcal S$. This is done as follows:
    \begin{itemize}
        \item [$\bullet$] By definition, each $\mathcal W_i$ is a direct summand of a free $\cO$-module $F_i$. Let $\lbrace f_{\alpha}, \alpha\in J\rbrace$ be a basis of $F_{i}$ parametrized by some index set $J$. We can define a derivation $\nabla^{F_i}$ by choosing images $\nabla^{F_i}(e_{\alpha})\in \mathcal S_{i+1}$. The extension to $F_i$ is given by 
        $$
\nabla^{F_i}(\sum_{\alpha\in J}\lambda_{\alpha}f_\alpha) = \sum_{\alpha\in J}\hat Q|_{\cO}(\lambda_{\alpha}) f_{\alpha} + \sum_{\alpha\in J}\lambda_{\alpha}  \nabla^{F_i}(f_{\alpha}) 
        $$
  for some coefficients $\lambda_{\alpha} \in \cO$.     
        \item [$\bullet$] By means of inclusion $ \mathcal W_i \hookrightarrow F_i$ and projection $F_{i} \twoheadrightarrow \mathcal W_i$ we define a derivation $ \nabla$ on each $\mathcal W_i$, and extend it to $\mathcal S$ by the graded Leibniz rule. 
    \end{itemize} This derivation $ \nabla$ descends to  $\nabla'\colon \mathcal S\otimes \cO/\cI \longrightarrow\mathcal S \otimes \cO/\cI$. Then the difference $Q-\nabla'$ is manifestly $\cO/\cI$-linear, and it is also $\cO$-linear, taking into account the induced $\cO$-module structure. Since each $\mathcal W_i$ is projective, we can lift $Q-\nabla'$ to an $\cO$-module map as below:
        $$
        \begin{tikzcd}[row sep=large, column sep=large]
\mathcal W_i \arrow[two heads,d, "\mathrm{Pr}" ] \arrow[dashed, r, "\exists A"] & \mathcal S_{i+1} \arrow[two heads, d, "\mathrm{Pr}"]\\
 \mathcal W_i \otimes \cO/\cI \arrow[r, "Q - \nabla'"] & \mathcal S_{i+1} \otimes \cO/\cI
\end{tikzcd}
$$
The morphism $A$ is extended to $\mathcal S$ as a derivation, so that the diagram
$$
        \begin{tikzcd}[row sep=large, column sep=large]
\mathcal S \arrow[two heads,d, "\mathrm{Pr}" ] \arrow[dashed, r, "\exists A"] & \mathcal S \arrow[two heads, d, "\mathrm{Pr}"]\\
 \mathcal S \otimes \cO/\cI \arrow[r, "Q - \nabla'"] & \mathcal S \otimes \cO/\cI
\end{tikzcd}
$$
commutes. The derivation $\hat Q$ is defined as $\hat Q = \nabla +A$. By construction, it is a lift of $Q$. Moreover, if $Q$ is a differential, the image of  $\hat Q^2$ lies in the kernel of $\mathrm{Pr}$:
$$
\mathrm{Pr}\circ (\nabla + A)\circ( \nabla+A) =(\nabla' + Q-\nabla')\circ\mathrm{Pr} \circ (\hat \nabla + A) = Q^2\circ \mathrm{Pr} = 0.
$$
Therefore, $\hat Q^2(\mathcal{S}) \subseteq \cI\mathcal S$.
    \end{proof}
    
\section*{Conclusion}

The study of \(Q\)-varieties (or \(Q\)-manifolds) and their dual counterparts, Lie \(\infty\)-algebroids
\cite{Campos,Voronov,Voronov2}, as well as of ``higher groupoids'' \cite{Severa}, is commonly motivated by their
applications across various areas of mathematics and theoretical physics. These structures often arise
in contexts where, at first glance, no higher-structure framework appears to be involved, yet they
ultimately prove indispensable for addressing natural and fundamental questions. Prominent examples
include deformation quantization of Poisson manifolds \cite{Kontsevich}, recent developments in BV operator
theory (see, for instance, \cite{CamposBV}), deformations of coisotropic submanifolds \cite{CattaneoFelder},
integration problems for Lie algebroids via stacky groupoids \cite{Zhu}, and the study of complex
submanifolds and Atiyah classes \cite{Kapranov,ChenStienonXu,LSX}. Many further examples could be added to this list.

In this work, we provide explicit constructions (Theorems~\ref{thm:gen.case.computations} and~\ref{thm:main2}) and concrete examples of
\(\mathbb{Z}\)-graded \(Q\)-varieties arising from arborescent Koszul--Tate resolutions of
\(\mathcal O/\mathcal I\), for a given ideal \(\mathcal I\), together with the universal Lie
\(\infty\)-algebroid associated with a Lie--Rinehart algebra preserving \(\mathcal I\) in the sense of
\cite{LLS,CLRL}. These results extend \cite[Theorem~3.21]{KOTOV2023104908} to a purely algebraic setting and provide a
more explicit description, as well as an effective algorithm that significantly reduces the homological
computations required in the construction.

Theorems~2.4 and~2.6 admit direct applications to singular foliation theory. In particular, every
solvable singular foliation, in the sense of \cite{LLS, LLL2}, that preserves an ideal \(\mathcal I\) arises
as the image of the anchor map of a \(\mathbb{Z}\)-graded \(Q\)-manifold, thereby extending the main
results of \cite{LLS,CLRL}. As a notable consequence, to every affine variety
\(W \subset \mathbb{C}^d\) one can associate a \(\mathbb{Z}\)-graded \(Q\)-variety canonically attached
to the vanishing ideal \(\mathcal I_W\) of \(W\) and to its Lie--Rinehart algebra of vector fields
\(\mathfrak{X}(W)\).

Finally, the \(\mathbb{Z}\)-graded \(Q\)-varieties constructed here are naturally equipped with
collections of (higher) multiplications, described in Appendix~\ref{sec:h.mult}. A geometric interpretation of these
structures would be of considerable interest and will be the subject of a forthcoming paper.

\bibliographystyle{unsrt}
\bibliography{Zgraded}

@article{dubuc19841,
	title = {On 1-form classifiers},
	volume = {12},
	issn = {0092-7872, 1532-4125},
	url = {http://www.tandfonline.com/doi/abs/10.1080/00927878408823064},
	doi = {10.1080/00927878408823064},
	language = {en},
	number = {12},
	urldate = {2026-07-23},
	journal = {Communications in Algebra},
	author = {Dubuc, E. J. and Kock, A.},
	year = {1984},
	pages = {1471--1531},
}

@article{Yamashita01112016,
author = {Tatsuya Yamashita},
title = {Derivations On a ${C}^\infty$-ring},
journal = {Communications in Algebra},
volume = {44},
number = {11},
pages = {4811--4822},
year = {2016},
publisher = {Taylor \& Francis},
doi = {10.1080/00927872.2015.1113293},


URL = { 
    
        https://doi.org/10.1080/00927872.2015.1113293
    
    

},
eprint = { 
    
        https://doi.org/10.1080/00927872.2015.1113293
    
    

}

}

@article{khudaverdian2008higher,
	title = {Higher {Poisson} {Brackets} and {Differential} {Forms}},
	volume = {1079},
	issn = {0094-243X},
	url = {https://doi.org/10.1063/1.3043861},
	doi = {10.1063/1.3043861},
	number = {1},
	journal = {AIP Conference Proceedings},
	author = {Khudaverdian, H. M. and Voronov, Th. Th.},
	month = nov,
	year = {2008},
	pages = {203--215},
}

@article{LermanLouis,
  title={Lie-{R}inehart and {P}oisson algebras over ${C}^{\infty}$-rings},
  author={Lerman, Eugene and Louis, Ruben},
  journal={arXiv preprint arXiv:2606.01388},
  year={2026}
}

@incollection{LLL1,
  author    = {Laurent-Gengoux, Camille and Louis, Ruben and Ryvkin, Leonid},
  title     = {What Is a Singular Foliation?},
  booktitle = {Advances in Poisson Geometry},
  series    = {Advanced Courses in Mathematics - CRM Barcelona},
  publisher = {Birkh{\"a}user},
  address   = {Cham},
  year      = {2025},
  doi       = {10.1007/978-3-031-86657-9_3},
}

@incollection{LLL2,
  author    = {Laurent-Gengoux, Camille and Louis, Ruben and Ryvkin, Leonid},
  title     = {Canonical Geometric and Algebraic Structures Hidden Behind a Singular Foliation},
  booktitle = {Advances in Poisson Geometry},
  series    = {Advanced Courses in Mathematics - CRM Barcelona},
  publisher = {Birkh{\"a}user},
  address   = {Cham},
  year      = {2025},
  doi       = {10.1007/978-3-031-86657-9_4},
}

@phdthesis{louis2023universalhigherliealgebras,
  title={Universal higher Lie algebras of singular spaces and their symmetries},
  author={Louis, Ruben},
  year={2022},
  school={Universit{\'e} de Lorraine}
}

@book{Eisenbud1995,
  author    = {David Eisenbud},
  title     = {Commutative Algebra: with a View Toward Algebraic Geometry},
  year      = {1995},
  publisher = {Springer-Verlag},
  series    = {Graduate Texts in Mathematics},
  volume    = {150},
  address   = {New York}
}

@Inbook{Bourbaki2006,
author="Bourbaki, N.",
title="Graduations, filtrations et topologies",
bookTitle="Alg{\`e}bre commutative",
year="2006",
publisher="Springer Berlin Heidelberg",
pages="181--304",
isbn="978-3-540-33976-2",
doi="10.1007/978-3-540-33976-2_4",
url="https://doi.org/10.1007/978-3-540-33976-2_4"
}

@article{Pavol-Severa,
  author = {Pavol Ševera}, 
  title = {Some title containing the words ``homotopy'' and
``symplectic'', e.g., this one.},
  journal = {Travaux mathématiques},
  year = {2005},
  volume = {16},
  pages = {121--137},
  publisher = {Université du Luxembourg}
}

@phdthesis{hancharuk:tel-04692988,
  TITLE = {{Homological methods for Gauge Theories with singularities}},
  AUTHOR = {Hancharuk, Aliaksandr},
  URL = {https://theses.hal.science/tel-04692988},
  NUMBER = {2023LYO10193},
  SCHOOL = {{Universit{\'e} Claude Bernard - Lyon I}},
  YEAR = {2023},
  MONTH = Oct,
  TYPE = {PhD thesis},
  HAL_ID = {tel-04692988},
  HAL_VERSION = {v1},
}

@article{hancharuk2024,
      title={Koszul-{T}ate resolutions and decorated trees}, 
      author={Aliaksandr Hancharuk and Camille Laurent-Gengoux and Thomas Strobl},
      eprint={2406.03955},
      journal={arXiv preprint arXiv:2406.03955},
      archivePrefix={arXiv},
      primaryClass={math.AC},
      url={https://arxiv.org/abs/2406.03955}, 
      year = {2024},
}

@article{Siebert1996,
author = {Siebert, Thomas},
journal = {Mathematische Annalen},
number = {2},
pages = {271-286},
title = {Lie algebras of derivations and affine algebraic geometry over fields of characteristic 0.},
url = {http://eudml.org/doc/182987},
volume = {305},
year = {1996},
}

@article{Grabowski1978,
	title = {Isomorphisms and ideals of the {Lie} algebras of vector fields},
	volume = {50},
	copyright = {http://www.springer.com/tdm},
	issn = {0020-9910, 1432-1297},
	url = {http://link.springer.com/10.1007/BF01406466},
	doi = {10.1007/BF01406466},
	language = {en},
	number = {1},
	urldate = {2026-07-23},
	journal = {Inventiones Mathematicae},
	author = {Grabowski, J.},
	year = {1978},
	pages = {13--33},
}

@article{D.A.Jordan,
    author = {Jordan, D. A.},
    title = {On the Ideals of a {L}ie Algebra of Derivations},
    journal = {Journal of the London Mathematical Society},
    volume = {s2-33},
    number = {1},
    pages = {33-39},
    year = {1986},
    month = {02},
    issn = {0024-6107},
    doi = {10.1112/jlms/s2-33.1.33},
    url = {https://doi.org/10.1112/jlms/s2-33.1.33},
    eprint = {https://academic.oup.com/jlms/article-pdf/s2-33/1/33/2385357/s2-33-1-33.pdf},
}

@article{Omori1980,
  author = {Hideki Omori},
  title = {A Method of Classifying Expansive Singularities},
  journal = {Journal of Differential Geometry},
  volume = {15},
  number = {4},
  pages = {493--512},
  year = {1980},
  publisher = {Project Euclid},
  doi = {10.4310/jdg/1214435839}
}

@article{Pursell-Shanks,
 ISSN = {00029939, 10886826},
 URL = {http://www.jstor.org/stable/2031961},
 author = {M. E. Shanks and Lyle E. Pursell},
 journal = {Proceedings of the American Mathematical Society},
 number = {3},
 pages = {468--472},
 publisher = {American Mathematical Society},
 title = {The {L}ie Algebra of a Smooth Manifold},
 urldate = {2025-12-11},
 volume = {5},
 year = {1954}
}

@article{KotovSalnikov,
title = {The category of $\mathbb {Z}$-graded manifolds: What happens if you do not stay positive},
journal = {Differential Geometry and its Applications},
volume = {93},
pages = {102109},
year = {2024},
issn = {0926-2245},
doi = {https://doi.org/10.1016/j.difgeo.2024.102109},
url = {https://www.sciencedirect.com/science/article/pii/S0926224524000020},
author = {Alexei Kotov and Vladimir Salnikov}
}

@article{CLRL,
author = {Laurent-Gengoux, Camille and Louis, Ruben},
copyright = {2021 Elsevier Inc.},
issn = {0021-8693},
journal = {Journal of {A}lgebra},
language = {eng},
pages = {1-53},
publisher = {Elsevier Inc},
title = {Lie-{R}inehart algebras $\simeq$ acyclic {L}ie $\infty$-algebroids},
volume = {594},
year = {2022},
note={\url{https://doi.org/10.1016/j.jalgebra.2021.11.023}},
}

@article{KOTOV2023104908,
title = {Normal forms of $\mathbb {Z}$-graded ${Q}$-manifolds},
journal = {Journal of Geometry and Physics},
volume = {191},
pages = {104908},
year = {2023},
issn = {0393-0440},
doi = {https://doi.org/10.1016/j.geomphys.2023.104908},
note={\url{https://www.sciencedirect.com/science/article/pii/S0393044023001602}},
author = {Alexei Kotov and Camille Laurent-Gengoux and Vladimir Salnikov}
}

@book{HT,
isbn = {0-691-08775-X},
address = {Princeton (N.J.)},
booktitle = {Quantization of gauge systems},
author = {Henneaux, Marc and Teitelboim,  Claudio},
publisher = {Princeton University Press},
title = {Quantization of {G}auge {S}ystems},
year = {1992},
}

@article{stasheff_poisson,
author = {Jim Stasheff},
title = {{Homological reduction of constrained Poisson algebras}},
volume = {45},
journal = {Journal of Differential Geometry},
number = {1},
publisher = {Lehigh University},
pages = {221 -- 240},
year = {1997},
doi = {10.4310/jdg/1214459757},
URL = {https://doi.org/10.4310/jdg/1214459757}
}

@article{LLS,
 Author = {Camille {Laurent-Gengoux} and Sylvain {Lavau} and Thomas {Strobl}},
 Title = {{The universal {L}ie \(\infty\)-algebroid of a singular foliation}},
 FJournal = {{Documenta Mathematica}},
 Journal = {{Doc. Math.}},
 ISSN = {1431-0635; 1431-0643/e},
 Volume = {25},
 Pages = {1571--1652},
 Year = {2020},
 Publisher = {Deutsche Mathematiker-Vereinigung, Berlin},
 Language = {English},
 MSC2010 = {53C12 57R30 18G10 58H05},
 Zbl = {1453.53033}
}

@article {Zhu,
    AUTHOR = {Tseng, Hsian-Hua and Zhu, Chenchang},
     TITLE = {Integrating {L}ie algebroids via stacks},
   JOURNAL = {Compositio Mathematica},
  FJOURNAL = {Compositio Mathematica},
    VOLUME = {142},
      YEAR = {2006},
    NUMBER = {1},
     PAGES = {251--270},
      ISSN = {0010-437X},
   MRCLASS = {58H05 (14A20)},
  MRNUMBER = {2197411},
MRREVIEWER = {Iakovos Androulidakis},
       DOI = {10.1112/S0010437X05001752},
       NOTE={\url{https://doi.org/10.1112/S0010437X05001752}},
}

@article {CattaneoFelder,
    AUTHOR = {Cattaneo, Alberto S. and Felder, Giovanni},
     TITLE = {Coisotropic submanifolds in {P}oisson geometry and branes in
              the {P}oisson sigma model},
   JOURNAL = {Lett. Math. Phys.},
  FJOURNAL = {Letters in Mathematical Physics},
    VOLUME = {69},
      YEAR = {2004},
     PAGES = {157--175},
      ISSN = {0377-9017},
   MRCLASS = {81T45 (22A22 53D17 53D20 53D55)},
  MRNUMBER = {2104442},
MRREVIEWER = {Stefan Waldmann},
       DOI = {10.1007/s11005-004-0609-7},
       note={\url{https://doi.org/10.1007/s11005-004-0609-7}},
}

@article{Severa,
	title = {Integration of differential graded manifolds},
	volume = {2020},
	issn = {1073-7928},
	url = {https://doi.org/10.1093/imrn/rnz004},
	doi = {10.1093/imrn/rnz004},
	number = {20},
	urldate = {2026-07-23},
	journal = {International Mathematics Research Notices},
	author = {Ševera, Pavol and Širaň, Michal},
	month = oct,
	year = {2020},
	pages = {6769--6814},
}

@article {Kontsevich,
    AUTHOR = {Kontsevich, Maxim},
     TITLE = {Deformation quantization of {P}oisson manifolds},
   JOURNAL = {Lett. Math. Phys.},
  FJOURNAL = {Letters in {M}athematical {P}hysics},
    VOLUME = {66},
      YEAR = {2003},
    NUMBER = {3},
     PAGES = {157--216},
      ISSN = {0377-9017},
   MRCLASS = {53D55 (16E40 53D17 81S10)},
  MRNUMBER = {2062626},
MRREVIEWER = {David Chataur},
       DOI = {10.1023/B:MATH.0000027508.00421.bf},
       note={\url{https://doi.org/10.1023/B:MATH.0000027508.00421.bf}},
}

@article{LSX,
title = {Poincaré–{B}irkhoff–{W}itt isomorphisms and {K}apranov dg-manifolds},
journal = {Advances in Mathematics},
volume = {387},
pages = {107792},
year = {2021},
issn = {0001-8708},
doi = {https://doi.org/10.1016/j.aim.2021.107792},
url = {https://www.sciencedirect.com/science/article/pii/S0001870821002310},
author = {Camille Laurent-Gengoux and Mathieu Stiénon and Ping Xu},
}

@article{AndroulidakisIakovos,
year = {2009},
author = {Androulidakis, Iakovos and Skandalis, Georges},
copyright = {COPYRIGHT 2009 Walter de Gruyter GmbH & Co. KG},
issn = {0075-4102},
journal = {Journal für die reine und angewandte Mathematik},
language = {eng},
number = {626},
pages = {1-37},
publisher = {Walter de Gruyter GmbH & Co. KG},
title = {The holonomy groupoid of a singular foliation},
note={\url{ https://doi.org/10.5167/uzh-23589}},
volume = {2009},
}

@article{Debord,
author = "Debord, Claire",
doi = "10.4310/jdg/1090348356",
fjournal = "Journal of Differential Geometry",
journal = "J. Differential Geom.",
month = "07",
number = "3",
pages = "467--500",
publisher = "Lehigh University",
title = "Holonomy {G}roupoids of {S}ingular {F}oliations",
note={\url{https://doi.org/10.4310/jdg/1090348356}},
volume = "58",
year = "2001"
}

@book{Matsumura,
author = {Matsumura Hideyuki},
booktitle = {Commutative ring theory},
isbn = {978-0-521-36764-6},
language = {eng},
publisher = {Cambridge University Press},
series = {Cambridge {S}tudies in {A}dvanced {M}athematics},
title = {Commutative ring theory},
year = {1989},
}

@article {Poncin,
    AUTHOR = {Bonavolont\`a, Giuseppe and Poncin, Norbert},
     TITLE = {On the category of {L}ie {$n$}-algebroids},
   JOURNAL = {J. Geom. Phys.},
  FJOURNAL = {Journal of Geometry and Physics},
    VOLUME = {73},
      YEAR = {2013},
     PAGES = {70--90},
      ISSN = {0393-0440},
   MRCLASS = {58H05 (18D50 18F20 53D55 58A50)},
  MRNUMBER = {3090103},
MRREVIEWER = {Laiachi El Kaoutit},
       DOI = {10.1016/j.geomphys.2013.05.004},
       note={\url{https://doi.org/10.1016/j.geomphys.2013.05.004}},
}

@incollection {Voronov2,
    AUTHOR = {Voronov, Theodore Th.},
     TITLE = {{$Q$}-manifolds and higher analogs of {L}ie algebroids},
 BOOKTITLE = {X{XIX} {W}orkshop on {G}eometric {M}ethods in {P}hysics},
    SERIES = {AIP Conf. Proc.},
    VOLUME = {1307},
     PAGES = {191--202},
 PUBLISHER = {Amer. Inst. Phys., Melville, NY},
      YEAR = {2010},
   MRCLASS = {53D17 (17B70)},
  MRNUMBER = {2768006},
MRREVIEWER = {Andrew James Bruce},
}

@article{Voronov,
	title = {Higher derived brackets and homotopy algebras},
	volume = {202},
	issn = {0022-4049},
	url = {https://www.sciencedirect.com/science/article/pii/S0022404905000186},
	doi = {https://doi.org/10.1016/j.jpaa.2005.01.010},
	number = {1},
	journal = {Journal of Pure and Applied Algebra},
	author = {Voronov, Theodore},
	year = {2005},
	pages = {133--153},
}

@article {ChenStienonXu,
    AUTHOR = {Chen, Zhuo and Sti\'{e}non, Mathieu and Xu, Ping},
     TITLE = {From {A}tiyah classes to homotopy {L}eibniz algebras},
   JOURNAL = {Comm. Math. Phys.},
  FJOURNAL = {Communications in Mathematical Physics},
    VOLUME = {341},
      YEAR = {2016},
     PAGES = {309--349},
      ISSN = {0010-3616},
   MRCLASS = {53D17 (17B55 32C36)},
  MRNUMBER = {3439229},
MRREVIEWER = {Ryan E. Grady},
       DOI = {10.1007/s00220-015-2494-6},
       note={\url{https://doi.org/10.1007/s00220-015-2494-6}},
}

@article {Kapranov,
    AUTHOR = {Kapranov, Mikhail},
     TITLE = {Rozansky-{W}itten invariants via {A}tiyah classes},
   JOURNAL = {Compositio Mathematica},
  FJOURNAL = {Compositio Mathematica},
    VOLUME = {115},
      YEAR = {1999},
    NUMBER = {1},
     PAGES = {71--113},
      ISSN = {0010-437X},
   MRCLASS = {57R57 (17B70 32J18 53D35 55P48 57M27 57R56)},
  MRNUMBER = {1671737},
MRREVIEWER = {Philip A. Foth},
       DOI = {10.1023/A:1000664527238},
       note={\url{https://doi.org/10.1023/A:1000664527238}},
}

@article{Androulidakis-Iakovos-Georges,
author = {Iakovos Androulidakis and Georges Skandalis},
title = {{A Baum–Connes conjecture for singular foliations}},
volume = {4},
journal = {Annals of K-Theory},
number = {4},
publisher = {MSP},
pages = {561 -- 620},
year = {2019},
doi = {10.2140/akt.2019.4.561},
note={\url{https://doi.org/10.2140/akt.2019.4.561}
}}

@article {CamposBV,
    AUTHOR = {Campos, Ricardo},
     TITLE = {B{V} formality},
   JOURNAL = {Advances in Mathematics},
    VOLUME = {306},
      YEAR = {2017},
     PAGES = {807--851},
      ISSN = {0001-8708},
   MRCLASS = {18D50 (16E40 53D55 55P48)},
  MRNUMBER = {3581318},
MRREVIEWER = {Beno\^{\i}t Fresse},
       DOI = {10.1016/j.aim.2016.10.034},
       note={\url{https://doi.org/10.1016/j.aim.2016.10.034}},
}

@article {Campos,
    AUTHOR = {Campos, Ricardo},
     TITLE = {Homotopy equivalence of shifted cotangent bundles},
   JOURNAL = {J. Lie Theory},
  FJOURNAL = {Journal of {L}ie {T}heory},
    VOLUME = {29},
      YEAR = {2019},
    NUMBER = {3},
     PAGES = {629--646},
      ISSN = {0949-5932},
   MRCLASS = {58A50 (17B63 18G55)},
  MRNUMBER = {3973608},
MRREVIEWER = {L. Nourmohammadifar},
}

@book{Hartshorne,
publisher = {Springer New York},
series = {Graduate Texts in Mathematics, 52},
title = {Algebraic Geometry},
year = {1977},
isbn = {1-4757-3849-8},
language = {eng},
author = {Hartshorne, Robin},
}

@book{peeva2010graded,
copyright = {Springer-Verlag London Limited 2011},
isbn = {0857291769},
language = {eng},
publisher = {Springer London},
series = {Algebra and Applications},
title = {Graded Syzygies},
volume = {14},
year={2010},
author = {Peeva, Irena},
}

@article{Tate,
  title={Homology of {N}oetherian rings and local rings},
  author={John Tate},
  journal={Illinois Journal of Mathematics},
  year={1957},
  volume={1},
  pages={14-27}
}

@article{Barnich_1998,
   title={The sh {L}ie Structure of {P}oisson Brackets in Field Theory},
   volume={191},
   ISSN={1432-0916},
   url={http://dx.doi.org/10.1007/s002200050278},
   DOI={10.1007/s002200050278},
   number={3},
   journal={Communications in Mathematical Physics},
   publisher={Springer Science and Business Media LLC},
   author={Barnich, G. and Fulp, R. and Lada, T. and Stasheff, J.},
   year={1998},
   pages={585–601} }

@article{joyce2019a,
  edition = {},
  number = {1256},
  journal = {Memoirs of the American Mathematical Society},
  pages = {},
  publisher = {American Mathematical Society},
  school = {},
  title = {Algebraic geometry over ${C}^\infty$-rings},
  volume = {260},
  author = {Joyce, D.},
  editor = {},
  year = {2019},
  series = {}
}

@article{LERMAN2024105062,
title = {Differential forms on {$C^{\infty}$}-ringed spaces},
journal = {Journal of Geometry and Physics},
volume = {196},
pages = {105062},
year = {2024},
issn = {0393-0440},
doi = {https://doi.org/10.1016/j.geomphys.2023.105062},
url = {https://www.sciencedirect.com/science/article/pii/S0393044023003145},
author = {Eugene Lerman}
}

@article{Osborn,
  edition = {},
  number = {1},
  journal = {Illinois Journal of Mathematics},
  pages = {137-144},
  publisher = {},
  school = {},
  title = {Derivations of commutative algebras},
  volume = {13},
  author = {Osborn, H.},
  editor = {1},
  year = {1969},
  series = {}
}

@Article{zbMATH03080175,
 Author = {Eilenberg, Samuel and MacLane, Saunders},
 Title = {On the groups {{\(H(\Pi,n)\)}}. {I}},
 FJournal = {Annals of Mathematics. Second Series},
 Journal = {Ann. Math. (2)},
 ISSN = {0003-486X},
 Volume = {58},
 Pages = {55--106},
 Year = {1953},
 Language = {English},
 DOI = {10.2307/1969820},
 zbMATH = {3080175},
 Zbl = {0050.39304}
}

@article{Scarf,
  title={Monomial resolutions},
  author={Bayer, Dave and Peeva, Irena and Sturmfels, Bernd},
  journal={Mathematical Research Letters},
  volume={5},
  number={1},
  pages={31--46},
  year={1998},
  publisher={International Press of Boston}
}

@book{Macaulay,
	address = {Berlin Heidelberg},
	series = {Algorithms and {C}omputation in {M}athematics},
	title = {Computations in algebraic geometry with {Macaulay} 2},
	isbn = {9783540422303 9783642075926},
	language = {eng},
	number = {8},
	publisher = {Springer},
	editor = {Eisenbud, David and Grayson, Daniel R. and Stillman, Michael and Sturmfels, Bernd},
	year = {2002},
}

@incollection{Felder-Kazhdan,
	title = {The classical master equation},
	volume = {610},
	isbn = {9780821891704 9781470415235},
	url = {https://www.ams.org/conm/610/12124/},
	language = {en},
	urldate = {2026-07-23},
	booktitle = {Contemporary {Mathematics}},
	publisher = {American Mathematical Society},
	author = {Felder, Giovanni and Kazhdan, David},
	year = {2014},
	doi = {10.1090/conm/610/12124},
	pages = {79--137},
}

@article{Grabowski-Kotov-Poncin2013,
     author = {Janusz Grabowski and Alexei Kotov and Norbert Poncin},
     title = {Lie superalgebras of differential operators},
     journal = {Journal of Lie Theory},
     pages = {35--54},
     year = {2013},
     volume = {23},
     number = {1},
     doi = {10.5802/jolt.714},
     zbl = {1264.58004},
     url = {https://jolt.centre-mersenne.org/articles/10.5802/jolt.714/}
}

@article{Grabowski-Kotov-Poncin2011,
 author = {Grabowski, Janusz and Kotov, Alexei and Poncin, Norbert},
 title = {Geometric structures encoded in the {Lie} structure of an {Atiyah} algebroid},
 fjournal = {Transformation Groups},
 journal = {Transform. Groups},
 issn = {1083-4362},
 volume = {16},
 number = {1},
 pages = {137--160},
 year = {2011},
 language = {English},
 doi = {10.1007/s00031-011-9126-9},
 url = {orbilu.uni.lu/handle/10993/14099},
 zbMATH = {5898845},
 Zbl = {1288.58010}
}

@article{Grabowski-Kotov-Poncin2010,
     author = {Janusz Grabowski and Alexei Kotov and Norbert Poncin},
     title = {The {Lie} Superalgebra of a Supermanifold},
     journal = {Journal of Lie Theory},
     pages = {739--749},
     year = {2010},
     volume = {20},
     number = {4},
     doi = {10.5802/jolt.617},
     zbl = {1207.58005},
}

@article{Barnich_2000,
   title={Local {BRST} cohomology in gauge theories},
   volume={338},
   ISSN={0370-1573},
   url={http://dx.doi.org/10.1016/S0370-1573(00)00049-1},
   DOI={10.1016/s0370-1573(00)00049-1},
   number={5},
   journal={Physics Reports},
   publisher={Elsevier BV},
   author={Barnich, Glenn and Brandt, Friedemann and Henneaux, Marc},
   year={2000},
   month=Nov, pages={439–569} }

@incollection{voronov_graded_2002-1,
	title = {Graded manifolds and {Drinfeld} doubles for {Lie} bialgebroids},
	volume = {315},
	isbn = {9780821832011 9780821856512 9780821879054},
	url = {https://www.ams.org/conm/315/5478/},
	language = {en},
	urldate = {2026-07-03},
	booktitle = {Contemporary {Mathematics}},
	publisher = {American Mathematical Society},
	author = {Voronov, Theodore},
	year = {2002},
	doi = {10.1090/conm/315/05478},
	pages = {131--168},
}

@article{Voronov_2019,
   title={Graded Geometry, {$Q$}‐Manifolds, and Microformal Geometry},
   volume={67},
   ISSN={1521-3978},
   url={http://dx.doi.org/10.1002/prop.201910023},
   DOI={10.1002/prop.201910023},
   number={8-9},
   journal={Fortschritte der Physik},
   publisher={Wiley},
   author={Voronov, Theodore Th.},
   year={2019},
   month=May }

@book{rosenberg_algebraic_1994,
	series = {Graduate {Texts} in {Mathematics}},
	title = {Algebraic {K}-{Theory} and {Its} {Applications}},
	volume = {147},
	copyright = {http://www.springer.com/tdm},
	isbn = {9781461287353 9781461243144},
	url = {http://link.springer.com/10.1007/978-1-4612-4314-4},
	urldate = {2026-07-07},
	publisher = {Springer New York},
	author = {Rosenberg, Jonathan},
	year = {1994},
	doi = {10.1007/978-1-4612-4314-4},
}
\vfill
\begin{center}
    \textsc{Department of Mathematics, Jilin University, Changchun 130012, Jilin, China\\Department of Mathematics, University of Illinois Urbana-Champaign\\ 1409 W. Green Street, Urbana, IL 61801, USA.}
\end{center}
\end{document}